\documentclass[12pt]{article}

\usepackage{multibib}
\newcites{Supp}{REFERENCES}
\usepackage{multirow}
\usepackage{mathtools}
\usepackage{babel}
\usepackage{algorithm}
\usepackage[noend]{algpseudocode}
\usepackage[font=small,labelfont=bf]{caption}
\usepackage{subcaption}
\usepackage[authoryear]{natbib}
\usepackage{graphics}
\usepackage{amsmath,amsthm,amssymb,graphicx,epsfig}
\usepackage[utf8]{inputenc}
\usepackage{amsfonts,mathrsfs}
\usepackage{array,paralist,indentfirst,ctable,xspace}
\usepackage{paralist,indentfirst,ctable}
\usepackage{url}

\usepackage[pdftex]{hyperref}
\usepackage{soul}

\newcommand{\bc}{\begin{center}}
  \newcommand{\ec}{\end{center}}

\newtheorem{theorem}{Theorem}[section]
\newtheorem{lemma}[theorem]{Lemma}
\newtheorem{proposition}[theorem]{Proposition}

\theoremstyle{definition} \newtheorem{definition}{Definition}[section]
\theoremstyle{definition} 
\theoremstyle{remark} 
\newtheorem{assumption}{Assumption}

\newcommand{\R}{\ensuremath{\mathbb{R}}}

\newcommand{\ud}{\ensuremath{\mathrm{d}}}

\newcommand*\mc[1]{\multicolumn{4}{c}{\textbf{#1}}}

\DeclareMathOperator*{\argmin}{\arg\min}

\newcommand{\bm}[1]{\ensuremath{\boldsymbol{#1}}}

\newcommand{\blind}{0}
\newcommand{\showrevisions}{0}
\newif\ifshowrevisions
\showrevisionstrue   

\newcommand{\revised}[1]{%
    \if1\showrevisions%
        \textcolor{blue}{#1}%
    \else%
        #1%
    \fi%
}

\begin{document}

\def\spacingset#1{\renewcommand{\baselinestretch}%
{#1}\small\normalsize} \spacingset{1}


\if0\blind
{
  \title{\bf Efficient Multidimensional Functional Data Analysis Using Marginal Product Basis Systems}
\author{William Consagra\\
  Brigham and Women’s Hospital, Harvard Medical School \\
  Arun Venkataraman \\
  Department of Physics and Astronomy,
  University of Rochester \\
  Xing Qiu \\
  Department of Biostatistics and Computational Biology,
  University of Rochester
}
  \maketitle
} \fi

\if1\blind
{
  \bigskip
  \bigskip
  \bigskip
  \begin{center}
    {\LARGE\bf Efficient Multidimensional Functional Data Analysis Using Marginal Product Basis Systems}
\end{center}
  \medskip
} \fi

\bigskip
\begin{abstract}
In areas ranging from neuroimaging to climate science, advances in data storage and sensor technology have led to a proliferation in multidimensional functional datasets. A common approach to analyzing functional data is to first map the discretely observed functional samples into continuous representations, and then perform downstream statistical analysis on these smooth representations. It is well known that many of the traditional approaches used for 1D functional data representation are plagued by the curse of dimensionality and quickly become intractable as the dimension of the domain increases. In this paper, we propose a computational framework for learning continuous representations from a sample of multidimensional functional data that is immune to several manifestations of the curse. The representations are constructed using a set of separable basis functions that are defined to be optimally adapted to the data. We show that the resulting estimation problem can be solved efficiently by the tensor decomposition of a carefully defined reduction transformation of the observed data. Roughness-based regularization is incorporated using a class of differential operator-based penalties. Relevant theoretical properties are also discussed. The advantages of our method over competing methods are thoroughly demonstrated in simulations. We conclude with a real data application of our method to a clinical diffusion MRI dataset.
\end{abstract}
\noindent%
{\it Keywords:}  functional data; basis representation; tensor decomposition; functional principal component analysis
\vfill
\thispagestyle{empty}

\newpage
\spacingset{1.5} 
\pagenumbering{arabic} 

\section{Introduction}
Functional data analysis (FDA) is a subfield of statistics concerned with the analysis of samples of functions. In most applications, the functional sample is not observed directly, rather, at some discrete number of domain points $\boldsymbol{x}_{ij}\in\mathcal{M}\subset\mathbb{R}^{D}$ according to 
$$
Y_{ij} = U_i(\boldsymbol{x}_{ij}) + \epsilon_{ij}; \quad i=1,...,N; j=1,...,M_i
$$
where $N$ is the sample size, $M_i$ the number of domain points for the $i$'th sample, $U_i\sim U$ is a random function and $\epsilon_{ij}$ is additive noise. In many FDA workflows, the analyst needs to perform the initial step of estimating a smooth function $\widehat{U}_i$ from each subject's discretely observed noisy data, i.e. the ``smoothing first, then estimation'' approach \citep{chen2007}, which we refer to here as \textit{functional representation}. Various downstream analyses are then performed using these reconstructed smooth functions. In this work, we are interested in the function representation problem for \revised{\textit{multidimensional}} functional data, i.e., random functions defined on multidimensional \revised{($D>1$)} domains. \revised{Specifically,} we consider the situation where the sample is observed on a common $D$-dimensional grid, i.e. $\boldsymbol{x}_{ij}$ are the same for all subjects, a setting frequently encountered in various imaging applications. 
\par 
When $D=1$, functional representation can be accomplished using standard nonparametric approaches, e.g., local kernel regression or basis expansion \citep{ramsay2005,Hsing2015TheoreticalFO}. For $D > 1$, these approaches suffer from the curse of dimensionality: the number of observations required to obtain a desired mean-squared error \citep{stone1980} and/or the number of model parameters (e.g. the number of basis functions for a tensor product basis) grow exponentially in $D$ \citep{wasserman2010}. In this scenario, semiparametric regression approaches 
provide a route to tractable estimation, though the associated structural assumptions are often overly restrictive for real word data. Hence, developing a general framework for multidimensional functional data analysis demands a different approach to representation.
\par 
It is well known that the optimal (minimum mean integrated squared error) low-rank representation for functional data can be formed using the eigenfunctions of the covariance operator of the process, and thus these form an attractive basis system for parsimonious modeling. Unfortunately, estimating the eigenfunctions, i.e., performing functional principal component analysis (FPCA) \citep{silverman1996}, for a general $D$-dimensional random function requires the nonparametric estimation of the $2D$-dimensional covariance function, denoted $C(\boldsymbol{x},\boldsymbol{y})$. Some general techniques have recently been proposed \citep{chen2017,yingxing2019,wang2020}, though the curse of dimensionality is even more problematic in this situation, due to both the dimension doubling effect and the symmetry and positive semi-definite constraints. To alleviate the computational difficulties associated with estimating a generic covariance function of a multidimensional process, a common tactic is to assume some notion of separability for $C$ \citep{muller2017,lynch2018}. These works are mostly developed for the case when the domain naturally decomposes into a product of two spaces, e.g., space and time, though extending the theory to general $D$-dimensional domains is feasible. Computationally speaking, it is a different story. For instance, when $D > 2$, the marginal product FPCA in \cite{muller2017} requires multidimensional numerical integration or nonparametric smoothing in order to estimate the marginal covariance functions, re-introducing a manifestation of the curse of dimensionality.
\par
In what follows, we propose a framework for multidimensional function representation based on learning the optimal \emph{marginal product basis} (MPB), i.e., a collection of independent multiplicatively separable functions, that avoids both direct estimation of, and explicit structural assumptions on, $C$. 
Critically, the number of parameters needed to estimate the MPB is \textit{linear} in $D$, and therefore this structure is effective for combating the curse of dimensionality. We prove that the optimal MPB defines a representation space that can be considered nearly optimal for a particular rank, with an inefficiency cost that becomes negligible for large ranks. 
To estimate the optimal MPB from the observed data, we identify an isometric embedding which allows the reparameterization of the observed data tensor into a lower dimensional space, permitting the derivation of a fast algorithm which scales favorably with huge datasets. This is in contrast to alternative methods for optimal basis construction for multidimensional functional data which rely on a smoothed decomposition of the raw data tensor \citep{huang2009, allen2013a, allen2019}, and as a result become computationally problematic for densely observed functions. Additionally, our approach enables a ``fully functional'' treatment of the estimation problem, in which the continuous basis functions are estimated directly, as opposed to the discrete factors estimated by the tensor decomposition approaches. Working directly with the continuous representations allows analytic computation of 
partial derivatives and inner products, thus facilitating efficient and stable two-stage algorithms for a variety of subsequent FDA tasks of interest. When working directly with the discrete data, such operations can be potentially numerically unstable (finite difference derivative approximation) or ill-defined (inner product between functions at different resolutions). 
\par 
\revised{It is worth noting that the multidimensional set-up considered here has received significantly less attention than \textit{multivariate} functional data. Multivariate FDA involves observing multiple, potentially correlated functions for each of the $N$ samples, with most approaches focusing on the case where the domain of each univariate functional data is one-dimensional \citep{li2020}. Approaches to analyzing \textit{multivariate multidimensional} functional data typically first require a functional representation for each univariate multidimensional domain, which are then used for the multivariate analysis \citep{happ2018}. As a result, this work can be seamlessly integrated with such methods.} 
\par
The rest of the paper is organized as follows. In Section~\ref{sec:model_theory}, we formulate the optimal MPB system and discuss relevant theoretical properties. 
In Section~\ref{sec:marginal_product_function_estimation}, we derive an efficient estimation procedure and discuss the incorporation of roughness-based regularization using differential operators \revised{as well as hyperparameter selection}. In Section~\ref{sec:mfpca_proc}, we illustrate how to utilize the MPB to derive a fast two-stage multidimensional FPCA. Section~\ref{sec:simulation_studies} compares the proposed method with competing methods in simulation studies. In Section~\ref{sec:real_data_analysis}, we analyze a set of magnetic resonance imaging data from a traumatic brain injury study. Section~\ref{sec:discussion} offers concluding remarks and potential future directions. Proofs for all theorems can be found in Supplemental Materials. \revised{Code for our algorithms and scripts to reproduce all of the results in the simulation section have been made publicly available in a python package: \url{https://github.com/Will-Consagra/eMFDA}}.

\section{Marginal Product Models}\label{sec:model_theory}
\subsection{Background and Model Description}
Let $U$ be a multidimensional random function with real-valued square integrable realizations, i.e. $U_i \sim U \in \mathcal{H} := \mathbb{L}^{2}(\mathcal{M})$. We assume the domain can be decomposed as $\mathcal{M} = \mathcal{M}_{1} \times \dots \times \mathcal{M}_{D}$, so for $\bm{x} = (x_{1}, \dots, x_{D})' \in \mathcal{M}$; each $x_{d}$, for $d=1,\dots, D$, is a member in the marginal domain $\mathcal{M}_{d}$, which is assumed to be a compact subset of Euclidean space $\mathbb{R}^{p_{d}}$. Without loss of generality, we assume the Lebesgue measure of $\mathcal{M}$ is 1. We make the following regularity assumptions on $U$:
\begin{assumption}\label{asm:random_function_assumptions}
 (a) $\mathbb{E}[U] = 0$, (b) $\mathbb{E}[\int_{\mathcal{M}}U^2(\bm{x})\ud \bm{x}] < \infty$, and (c) $U$ is \textit{mean-square continuous}.
\end{assumption}
The mean zero assumption on $U$ is made for convenience of presentation, and the mean square integrability and mean-square continuity assumptions are a standard requirement \citep{Hsing2015TheoreticalFO}. Under Assumption~\ref{asm:random_function_assumptions}, we are guaranteed that the covariance function $C(\bm{x}, \bm{y}) := \mathbb{E} \left[U(\boldsymbol{x})U(\boldsymbol{y}) \right]$ is continuous on $\mathcal{M}\times\mathcal{M}$. By Mercer's theorem, this covariance function has an eigen-decomposition
$
C(\boldsymbol{x}, \boldsymbol{y}) = \overset{\infty}{\underset{k=1}{\sum}}\rho_{k} \psi_{k}(\boldsymbol{x})\psi_{k}(\boldsymbol{y})
$,
where $\{\psi_{k}\}_{k=1}^\infty$ forms a complete orthonormal sequence of eigenfunctions in $\mathcal{H}$ and $\{\rho_k\}_{k=1}^\infty$ is a non-increasing sequence of real, non-negative eigenvalues. Additionally, by the Karhunen-Lo\'eve theorem, with probability one we have the decomposition
$U(\boldsymbol{x}) = \overset{\infty}{\underset{k=1}{\sum}} Z_k \psi_k(\boldsymbol{x})$,
where $Z_k = \langle  U, \psi_k \rangle_{\mathcal{H}}$, which are mean zero random variables with $\mathbb{E}\left[Z_kZ_j\right] = \rho_k\mathbb{I}\{k=j\}$.
\par 
Let $H_{d} := \mathbb{L}^2(\mathcal{M}_d)$, so that $\mathcal{H} := \mathbb{L}^{2}(\mathcal{M}) = \bigotimes_{d=1}^{D} H_{d}$, the tensor product of $D$ member spaces. We assume that there exists a complete basis system, $\bm{\phi}_{d} := \{\phi_{d,j}\}_{j=1}^\infty$, for each marginal function space $H_{d}$. Denote their rank-$m_{d}$ truncations as $\boldsymbol{\phi}_{m_{d},d} = (\phi_{d,1},...,\phi_{d,m_{d}})^\prime$; $H_{m_{d},d} := \mathrm{span}(\boldsymbol{\phi}_{m_{d},d})$; and $\mathcal{H}_{\bm{m}} := \bigotimes_{d=1}^{D} H_{m_{d},d}$, where $\bm{m} = (m_{1}, \dots, m_{D})'$ are the marginal ranks. By construction,
$$
    \bm{\tau}_{\bm{m}} := \bigotimes_{d=1}^{D} \bm{\phi}_{m_{d},d} = \Big\{ \tau_{j_{1},\dots,j_{D}}(\bm{x}) = \prod_{d=1}^{D}\phi_{d,j_{d}}(x_{d}), j_{d} = 1,\dots, m_{d} \Big\}
$$
is the complete tensor product bases (TPB) for $\mathcal{H}_{\bm{m}}$.
\begin{definition}[Marginal product structure]
  \label{def:MPS}
  $\zeta \in \mathcal{H}$ is called a rank-1 marginal product function if it is \emph{multiplicatively separable}, and $u \in \mathcal{H}$ is called a rank-$K$ marginal product function if it is a linear combination of $K$ independent rank-1 marginal product functions:
  \begin{equation}
    \label{eq:rankK-mpf}
    u(\boldsymbol{x}) = \sum_{k=1}^{K} b_{k} \zeta_{k}(\mathbf{x}) = \sum_{k=1}^{K} b_{k} \prod_{d=1}^{D} \xi_{k,d}(x_d), \qquad \xi_{k,d} \in H_{d}, \quad b_{k} \in \R.
  \end{equation}
We denote the collection of rank-1 marginal product functions with marginal ranks $\boldsymbol{m}$:
\begin{equation}
  \label{eq:MPF1}
  \begin{aligned}
      \mathcal{L}_{\bm{m}} &:= \Big\{ \zeta(\bm{x}): \zeta(\bm{x}) = \prod_{d=1}^{D} \xi_{d}(x_{d}), \; \xi_{d} \in H_{m_{d},d}, \; \|\xi_{d}\|_{H_{d}}=1 \Big\}
  \end{aligned}
\end{equation}
\end{definition}
In this work, we propose to estimate the optimal basis set of $K$ elements from $\mathcal{L}_{\bm{m}}$ for representing realizations of $U$, a notion formalized as follows:
\begin{definition}[Optimal Rank-$K$ MPB]
  \label{def:best_K_MPF}
    Define the set of functions
    \begin{equation}\label{eq:indp_approx_space_finite_dim}
        \mathcal{V}_{K,\boldsymbol{m}} := \Big\{\boldsymbol{\zeta}=(\zeta_1,...,\zeta_K)^\prime: \zeta_k\in\mathcal{L}_{\boldsymbol{m}}, \zeta_1,...,\zeta_K\text{ linearly independent}\Big\},
    \end{equation}
    and define the associated optimal rank $K$ MPB, denoted $K$-oMPB, as 
    \begin{equation}\label{eqn:best_K_MPF_finite_dim}
        \boldsymbol{\zeta}^{*}_{\boldsymbol{m}} = \underset{\boldsymbol{\zeta} \in \mathcal{V}_{K,\boldsymbol{m}}}{\text{arg\,inf}} \quad  \mathbb{E}\Big\| U - P_{\boldsymbol{\zeta}}(U)\Big\|_{\mathcal{H}}^2.
    \end{equation}
  where $P_{\boldsymbol{\zeta}}$ is the projection operator onto
  $\text{span}(\boldsymbol{\zeta})$.
\end{definition}
Given a random sample of $N$ realizations $U_i\sim U$, define the corresponding empirical estimate of \eqref{eqn:best_K_MPF_finite_dim} as
\begin{equation}\label{eqn:best_K_MPF_finite_dim_empirical}
    \Breve{\boldsymbol{\zeta}}_{\boldsymbol{m},N}^{*} = \underset{\boldsymbol{\zeta} \in \mathcal{V}_{K,\boldsymbol{m}}}{\text{arg\,inf}} \quad  \frac{1}{N}\sum_{i=1}^N\Big\| U_i - P_{\boldsymbol{\zeta}}(U_i)\Big\|_{\mathcal{H}}^2
\end{equation}
\subsection{Approximation Properties}\label{ssec:approx_theory}
We now characterize the expected asymptotic approximation power of $\Breve{\boldsymbol{\zeta}}_{\boldsymbol{m},N}^{*}$. 
  Let $\mathcal{A}_k$ be the $D$-mode tensor with elements $\mathcal{A}_k(j_{1},...,j_{D})$ defined by
  $$
  P_{\mathcal{H}_{\boldsymbol{m}}}(\psi_k) = \sum_{j_{1}=1}^{m_{1}}\cdots\sum_{j_{D}=1}^{m_{D}}\mathcal{A}_k(j_{1},...,j_{D})\phi_{1,j_{1}}\cdots\phi_{D,j_{D}},
  $$
  where $P_{\mathcal{H}_{\boldsymbol{m}}}$ is the projection operator onto $\mathcal{H}_{\boldsymbol{m}}$. Denote $\mathcal{A}^{(K)}$ as the tensor obtained by the $D+1$ mode stacking of $\mathcal{A}_1,...,\mathcal{A}_K$.
Let $\boldsymbol{J}_{\boldsymbol{\phi}_{d}}\in\mathbb{R}^{m_{d}\times m_{d}}$ be the matrix of pairwise $\mathcal{H}_{m_{d},d}$ inner products of $\boldsymbol{\phi}_{m_{d},d}$. Define the inner-product space 
  $(\bigotimes_{d=1}^D\mathbb{R}^{m_{d}}\otimes\mathbb{R}^{K}, \langle\cdot,\cdot\rangle_{\tilde{F},C})$ where
  $$
  \begin{aligned}
   \langle\mathcal{T}_1,\mathcal{T}_2\rangle_{\tilde{F}, C} &= \sum_{k=1}^K\rho_k\langle\mathcal{T}_1(:,...,:,k),\mathcal{T}_2(:,...,:,k)\times_{1}\boldsymbol{J}_{\boldsymbol{\phi}_{1}}\cdots\times_{D}\boldsymbol{J}_{\boldsymbol{\phi}_{D}} \rangle_{F} \\
  \end{aligned}
  $$
  for tensors $\mathcal{T}_1,\mathcal{T}_2\in \bigotimes_{d=1}^D\mathbb{R}^{m}\otimes\mathbb{R}^{K}$, where $\times_{d}$ denotes the tensor $d$-mode multiplication.
\begin{theorem}[Generalization Error]\label{thm:generalization_error}
  With slight abuse of notation, denote $P_{\Breve{\boldsymbol{\zeta}}_{\boldsymbol{m},N}^{*}}$ as the projection operator onto $\text{span}(\Breve{\boldsymbol{\zeta}}_{\boldsymbol{m},N}^{*})$ and let the function $w_{\bm{\tau}_{\boldsymbol{m}}}(\boldsymbol{m})$ be the $\mathbb{L}^2(\mathcal{M})$ convergence rate of the TPB system $\bm{\tau}_{\boldsymbol{m}}$ (Definition
  S1.1
  in the supplemental materials). Under Assumptions~\ref{asm:random_function_assumptions} and 
  S2, S3 
  in the supplementary material,
  \begin{equation}\label{eqn:covergence_rate_general}
    \begin{aligned}
      \mathbb{E}\left\| U - P_{\Breve{\boldsymbol{\zeta}}_{\boldsymbol{m},N}^{*}}(U)\right\|_{\mathcal{H}}^2 \le & \sum_{k=K+1}^\infty\rho_k + \left\| \mathcal{A}^{(K)} - \widehat{\mathcal{A}}^{(K)}_K\right\|_{\tilde{F},C}^2 
      &+ O(w_{\bm{\tau}_{\boldsymbol{m}}}(\boldsymbol{m})) +  O_{p}(N^{-1/2})
    \end{aligned}
  \end{equation}
  where the expectation is taken with respect to a new realization of $U$ and $\widehat{\mathcal{A}}^{(K)}_K$ is the rank $K$ canonical polyadic decomposition of the coefficient tensor $\mathcal{A}^{(K)}$ under the $\left\| \cdot \right\|_{\tilde{F},C}$-norm.
\end{theorem}
Theorem~\ref{thm:generalization_error} bounds the expected generalization error of $\Breve{\boldsymbol{\zeta}}_{\boldsymbol{m},N}^{*}$ by the sum of four terms. The first term is the tail sum of the eigenvalues, which is the expected generalization error of the optimal rank $K$ basis system (the eigenfunctions). The third term is the irreducible bias from the finite truncation of the marginal ranks $\boldsymbol{m}$. The fourth term reflects the finite sample statistical approximation error and can be established using convergence results from the theory of M-estimators. The second term shows that the inefficiency cost incurred by representing the function realizations with $\Breve{\boldsymbol{\zeta}}_{\boldsymbol{m},N}^{*}$, as opposed to the eigenfunctions, is driven by the low-rank structure of $\mathcal{A}^{(K)}$. This term is unknown in practice but vanishes for large enough $K$ (often, $K\ll \prod_{d=1}^Dm_d)$ though this value will depend on the particular tensor product space $\bm{\tau}$, marginal ranks $\boldsymbol{m}$ and covariance function $C$. Please visit Section S1 of the Supplemental Materials for further technical discussion on these and related theoretical matters. We now turn our attention to the development of computationally efficient algorithms to estimate $\Breve{\boldsymbol{\zeta}}_{\boldsymbol{m},N}^{*}$ in practice. 

\section{Estimation}\label{sec:marginal_product_function_estimation}
\subsection{Discrete Observation Model}
We consider the case in which the $U_i$ are observed with noise at each discrete location of a multidimensional grid $\mathcal{X}\subset\mathcal{M}$, where 
$$
\mathcal{X} = (x_{11}, x_{12}, ..., x_{1n_{1}})^\prime\times (x_{21}, x_{22}, ..., x_{2n_{2}})^\prime  \times\cdot\cdot\cdot\times (x_{D1}, x_{D2}, ..., x_{Dn_{D}})^\prime,
$$
and each vector of marginal grid points $\boldsymbol{x}_d  := (x_{d1}, x_{d2}, ..., x_{dn_{d}})\in \mathcal{M}_d$, according to the canonical observation model
\begin{equation}\label{eqn:statistical_model}
  \mathcal{Y}(i_1, i_2, ..., i_D, i) = U_i(x_{1,i_1}, x_{2,i_2}, ..., x_{D, i_D}) + \mathcal{E}(i_1, i_2, ..., i_D, i)
\end{equation}
for $i_d = 1, 2, ..., n_d$, $d = 1, 2, ..., D$, $i=1,2,...,N$, where $\mathcal{Y}$ is a $D+1$-mode tensor with dimensions $(n_1, n_2, ..., n_D, N)$,  $\mathbb{E}[\text{vec}(\mathcal{E})] = \boldsymbol{0}$ and $\text{Var}[\text{vec}(\mathcal{E})] = \sigma^2\boldsymbol{I}$. The discretized counterpart to \eqref{eqn:best_K_MPF_finite_dim_empirical} is given by 
\begin{equation}\label{eqn:best_K_MPF_finite_dim_empirical_discrete}
          \widehat{\boldsymbol{\zeta}}_{N,\boldsymbol{m}}^{*} := \quad \underset{\boldsymbol{\zeta} \in \mathcal{V}_{K,\boldsymbol{m}}}{\text{arg\,inf}} \min_{\boldsymbol{B}\in\mathbb{R}^{N\times K}}\quad  \frac{1}{N}\sum_{i=1}^N\left\| \mathcal{Y}_i - \sum_{k=1}^K \boldsymbol{B}_{i,k}\bigotimes_{d=1}^D\boldsymbol{\xi}_{d,k}\right\|_{F}^2.
\end{equation}
where $\mathcal{Y}_i\in\mathbb{R}^{n_{1}\times\cdots\times n_{D}}$ is the observed data tensor for the $i$th realization, $\boldsymbol{\xi}_{d,k} \in \mathbb{R}^{n_{d}}$ is the evaluation of $\xi_{k,d}$ on $\boldsymbol{x}_d$ and $\boldsymbol{B}$ is the matrix of coefficients for the $\zeta_{k}$'s for each of the $N$ samples. 

\subsection{A Convenient Reparameterization}\label{sssec:mpf_representation}
First we note that for $\boldsymbol{\zeta}\in\mathcal{V}_{K,\boldsymbol{m}}$, we have the representation $\xi_{k,d}(x_d) = \sum_{j=1}^{m_{d}} c_{d,k,j}\phi_{d,j}(x_d)$. Consequently  $\widehat{\boldsymbol{\zeta}}_{N,\boldsymbol{m}}^{*}$ is equivalently defined by the solutions to the following optimization problem
\begin{equation}\label{eqn:empirical_optimal_mpf_problem_finite_dim_reparam}
      (\widehat{\boldsymbol{C}}_1,...,\widehat{\boldsymbol{C}}_D) := \underset{(\boldsymbol{C}_1,...,\boldsymbol{C}_D)}{\text{arg\,inf}} \min_{\boldsymbol{B}}\frac{1}{N}\sum_{i=1}^N\left\| \mathcal{Y}_i - \sum_{k=1}^K \boldsymbol{B}_{i,k}\bigotimes_{d=1}^D\boldsymbol{\Phi}_d\boldsymbol{c}_{d,k}\right\|_{F}^2
\end{equation}
where $\boldsymbol{\Phi}_{d} \in \mathbb{R}^{n_{d}\times m_{d}}$ is the evaluation of $\boldsymbol{\phi}_{m_{d},d}$ on the marginal grid $\boldsymbol{x}_d$, i.e. $\boldsymbol{\Phi}_{d, i_{d} j_{d}} := \phi_{d,j_{d}}(x_{d,i_{d}})$, $\boldsymbol{c}_{d,k}\in\mathbb{R}^{m_{d}}$ is the coefficient of $\xi_{k,d}$, and $\boldsymbol{C}_d$ is the matrix whose columns are the $\boldsymbol{c}_{d,k}$. \revised{Note that the unit $\|\cdot\|_{H_{d}}$ norm constraint in the definition of $\mathcal{V}_{K,\boldsymbol{m}}$ is translated to a norm constraint on the columns of the $\boldsymbol{C}_{d}$'s that must be incorporated for identifiability reasons. For clarity of presentation, we defer additional discussion of this and related identifiability matters to Section 
S1 of the supplemental materials.} 
\par 
Denote the SVD of the basis evaluation matrices $\boldsymbol{\Phi}_{d} = \boldsymbol{U}_{d}\boldsymbol{D}_{d}\boldsymbol{V}_{d}'$. In general, we have $n_{d} > m_{d}$, so $\boldsymbol{U}_{d} \in \mathbb{R}^{n_{d}\times m_{d}}$ is a semi-orthogonal matrix; $\boldsymbol{D}_{d} \in \mathbb{R}^{m_{d}\times m_{d}}$ is an invertible diagonal matrix; and $\boldsymbol{V}_{d} \in \mathbb{R}^{m_{d}\times m_{d}}$ is an orthogonal matrix. For any $\boldsymbol{\zeta}\in\mathcal{V}_{K,\boldsymbol{m}}$, the evaluation of the MPB functions $\boldsymbol{\xi}_{d}=(\xi_{d1}, ..., \xi_{dK})^\prime$ on $\boldsymbol{x}_d$, is represented as
$$
\boldsymbol{\Xi}_{d} = \boldsymbol{\Phi}_{d} \boldsymbol{C}_{d} = \boldsymbol{U}_{d} \boldsymbol{D}_{d}\boldsymbol{V}_{d}' \boldsymbol{C}_{d} = \boldsymbol{U}_{d} \tilde{\boldsymbol{C}}_{d}, \qquad \boldsymbol{\tilde{C}}_{d} := \boldsymbol{D}_{d}\boldsymbol{V}_{d}' \boldsymbol{C}_{d}.
$$
\par 
The following theorem proves the equivalence between the solution of Equation~\eqref{eqn:empirical_optimal_mpf_problem_finite_dim_reparam} and the rank-$K$ canonical polyadic decomposition (CPD) of an appropriately defined transformation of the observed data tensor $\mathcal{Y}$.
\par
\begin{theorem}[Functional Tensor Decomposition Theorem]\label{thm:td_estimator_equivalence}
  Define $\widehat{\mathcal{G}} := \mathcal{Y} \times_{1} \boldsymbol{U}_{1}^\prime \times_{2} \boldsymbol{U}_{2}^\prime \cdots \times_{D} \boldsymbol{U}_{D}^\prime$, which is a $(D+1)$-mode tensor with dimensions $(m_1, m_2, ..., m_D, N)$, and denote its rank-$K$ decomposition by $\widehat{\mathcal{G}}_{K}(\boldsymbol{B}, \boldsymbol{\tilde{C}}) = \overset{K}{\underset{k=1}{\sum}} \Big[\overset{D}{\underset{d=1}{\bigotimes}} \boldsymbol{\tilde{c}}_{d,k} \Big]\otimes \boldsymbol{b}_{k}$, with factor matrices  $\boldsymbol{B}\in\mathbb{R}^{N\times K}$ and $\boldsymbol{\tilde{C}} = [\boldsymbol{\tilde{C}}_1, \dots, \boldsymbol{\tilde{C}}_D]$, $\boldsymbol{\tilde{C}}_d \in \mathbb{R}^{m_d\times K}$. $\boldsymbol{\tilde{c}}_{d,k}$ and $\boldsymbol{b}_{k}$ are the $k$th column of $\boldsymbol{\tilde{C}}_d$ and $\boldsymbol{B}$, respectively.  The optimization problem \eqref{eqn:empirical_optimal_mpf_problem_finite_dim_reparam} has the following solutions: 
  $\widehat{\boldsymbol{B}}=\boldsymbol{B}$ and $\widehat{\boldsymbol{C}}_d=\boldsymbol{V}_{d}\boldsymbol{D}_{d}^{-1}\boldsymbol{\tilde{C}}_d$, for $d=1, \dots ,D$.
\end{theorem}
Theorem~\ref{thm:td_estimator_equivalence} reveals that estimating the $K$-oMPB is equivalent to the CPD of the compressed $\widehat{\mathcal{G}}$ tensor. As the dimensionality of $\widehat{\mathcal{G}}$ is controlled by the marginal ranks $m_{d}$, as opposed to the number of marginal grid points $n_d$, this defines a practical reduction transformation which permits user control of the dimensionality of the optimization problem.

\subsection{Regularization in the Transformed Space}\label{sssec:regularization}
In order to ameliorate the influence of noise and discretization and to promote smoothness in the estimated $K$-oMPB basis, we incorporate a regularization term to the objective function in Equation~\eqref{eqn:empirical_optimal_mpf_problem_finite_dim_reparam} of the form: 
$$
\sum_{k=1}^K\mathrm{Pen}(\zeta_k) = \sum_{k=1}^K\int_{\mathcal{M}} \sum_{d=1}^D \lambda_{d} L^2_d(\xi_{k,d})
$$
for some $\lambda_d > 0$, where $L_d:\mathbb{W}^{\alpha_{d},2}(\mathcal{M}_d)\rightarrow\mathbb{L}^2(\mathcal{M}_d)$ is a linear (partial) differential operator, and $\mathbb{W}^{\alpha_{d},2}(\mathcal{M}_d)$ is the Sobolev space over the $d$th marginal domain, with order $\alpha_d$ defined appropriately.
\begin{proposition}\label{prop:convex_penalty_prop}
\revised{Let $\boldsymbol{T}_d:=\boldsymbol{D}_d^{-1}\boldsymbol{V}_d^{\prime}\boldsymbol{R}_d\boldsymbol{V}_d\boldsymbol{D}_d^{-1}$, with $\boldsymbol{R}_d(i,j) = \int_{\mathcal{M}_{d}}L_d(\phi_{d,i})L_d(\phi_{d,j})$, then $\sum_{k=1}^K\mathrm{Pen}(\zeta_k)= \sum_{d=1}^D\lambda_d\mathrm{tr}(\boldsymbol{\tilde{C}}^\prime_d\boldsymbol{T}_{d}\boldsymbol{\tilde{C}}_d)$.}
\end{proposition}
As a consequence, such penalties are quadratic in the transformed coordinate matrices $\boldsymbol{\tilde{C}}_d$ and therefore convex. This permits the derivation of efficient numerical algorithms to estimate the optimal MPB functions, which will be discussed in Section~\ref{ssec:algorithm}.
\par
Penalization of the coefficient matrix $\boldsymbol{B}$ is incorporated using the penalty function denoted $l(\boldsymbol{B})$. We assume that $l(\boldsymbol{B})$ is convex, which is a requirement to guarantee the convergence of the algorithm in Section~\ref{ssec:algorithm}, but otherwise leave it's form unspecified. \revised{For example, lasso-type penalties can be integrated to promote sparsity in the MPB functional representation for interpretability or feature extraction.}
Using the results from Theorem~\ref{thm:td_estimator_equivalence} and Proposition~\ref{prop:convex_penalty_prop}, the solution to the regularized augmentation of Equation~\eqref{eqn:empirical_optimal_mpf_problem_finite_dim_reparam} is a linear transformation of the solution to
\begin{equation}\label{eqn:Common_Basis_Optimizer_Transform_Tilde_with_regularization}
  \argmin_{\boldsymbol{\tilde{C}_1}, ..., \boldsymbol{\tilde{C}_D}, \boldsymbol{B}}\left\| \mathcal{\widehat{G}}- \sum_{k=1}^{K} \bigotimes_{d=1}^{D} \boldsymbol{\tilde{c}}_{d,k} \otimes \boldsymbol{b}_{k} \right\|_{F}^2 + \sum_{d=1}^D\lambda_{d}\mathrm{tr}(\boldsymbol{\tilde{C}}^\prime_d\boldsymbol{T}_{d}\boldsymbol{\tilde{C}}_d) + \lambda_{D+1}l(\boldsymbol{B}).
\end{equation}

\subsection{Algorithm}\label{ssec:algorithm}
In general, it can be shown that the optimization problem~\eqref{eqn:Common_Basis_Optimizer_Transform_Tilde_with_regularization} is non-convex and NP-hard \citep{hillar2013}. To derive a computationally tractable approximation algorithm, we propose a block coordinate descent based approach in which, for the $(r+1)$'th iteration, the variables are updated according to the sequence of conditional minimization problems
\begin{equation}\label{eqn:BCD}
  \tilde{\boldsymbol{C}}_{d}^{(r+1)} = \min_{\boldsymbol{X}}g(\tilde{\boldsymbol{C}}^{(r+1)}_1, ...,\tilde{\boldsymbol{C}}^{(r+1)}_{d-1}, \boldsymbol{X}, \tilde{\boldsymbol{C}}^{(r)}_{d+1}, ..., \tilde{\boldsymbol{C}}^{(r)}_{D}, \boldsymbol{B}^{(r)}), 
\end{equation}
for $d=1,\dots,D$ and likewise for $\boldsymbol{B}^{(r+1)}$, where  $g$ denotes the objective function from \eqref{eqn:Common_Basis_Optimizer_Transform_Tilde_with_regularization}.
\par
Using the properties of the $d$-mode matricization, we can write the conditional minimization problem defining the update of $\boldsymbol{\tilde{C}}_d$ as
\begin{equation}\label{eqn:tilde_d_update}
  \boldsymbol{\tilde{C}}_d^{(r+1)} = \min_{\boldsymbol{\tilde{C}}_{d}} \|\boldsymbol{G}_{(d)}^\prime - \boldsymbol{\tilde{C}}_d \boldsymbol{W}_d^{(r)\prime}\|_{F}^2 + \lambda_d\mathrm{tr}(\boldsymbol{\tilde{C}}^\prime_d\boldsymbol{T}_{d}\boldsymbol{\tilde{C}}_d).
\end{equation}
The update for $\boldsymbol{B}$ is given by
\begin{equation}\label{eqn:s_update}
  \boldsymbol{B}^{(r+1)} = \min_{\boldsymbol{B}}\|\boldsymbol{G}_{(D+1)} - \boldsymbol{W}^{(r)}_{D+1}\boldsymbol{B}^\prime\|_{F}^2 +\lambda_{D+1}l(\boldsymbol{B}),
\end{equation}
where $\boldsymbol{W}_d^{(r)} = (\bigodot_{j< d}^D \boldsymbol{\tilde{C}}_j^{(r+1)}\bigodot_{j > d}^D \boldsymbol{\tilde{C}}_j^{(r)})\odot \boldsymbol{B}^{(r)}$ for $d=1,...,D$, $\boldsymbol{W}_{D+1}^{(r)} = \bigodot_{d=1}^D \boldsymbol{\tilde{C}}_d^{(r+1)}$ and $\boldsymbol{G}_{(d)}$ is the $d$-mode unfolding of tensor $\mathcal{\widehat{G}}$. Here $\odot$ is the Khatri–Rao product. From here on the superscript $r$ denoting iteration is dropped for clarity.
\par
In fact, the solution to the subproblem \eqref{eqn:tilde_d_update} is equivalent to the solution to
\begin{equation}\label{eqn:sylvester_equation_formulation}
  \boldsymbol{\tilde{C}}_d\boldsymbol{W}_d^{\prime}\boldsymbol{W}_d + \lambda_d\boldsymbol{T}_d\boldsymbol{\tilde{C}}_d = \boldsymbol{W}_d^{\prime}\boldsymbol{G}_{(d)}.
\end{equation}
This equivalence can be verified by noting that  \eqref{eqn:sylvester_equation_formulation} defines the gradient equations of \eqref{eqn:tilde_d_update} and that the solution is globally optimum due to convexity. Equation \eqref{eqn:sylvester_equation_formulation} is known as the Sylvester equation and has a unique solution under very mild conditions (specifically $\boldsymbol{W}_d^\prime\boldsymbol{W}_d$ and $\lambda_d\boldsymbol{T}_d$ must have no common eigenvalues). Efficient algorithms for solving the Sylvester equation  \citep{bartels1972} are readily available in most common numerical computing languages.
\par
Notice that by introducing the auxiliary variable $\boldsymbol{Z} = \boldsymbol{B}^\prime$, the subproblem \eqref{eqn:s_update} can be written in separable form as
\begin{equation}\label{eqn:S_update_ADMM}
  \begin{aligned}
    & \min_{\boldsymbol{B}, \boldsymbol{Z}} \|\boldsymbol{G}_{(D+1)} - \boldsymbol{W}_{D+1}\boldsymbol{Z}\|_{F}^2+ \lambda_{D+1}l(\boldsymbol{B}) \\
    & \text{subject to } \boldsymbol{B} - \boldsymbol{Z}^\prime = \boldsymbol{0}.
  \end{aligned}
\end{equation}
A numerical approximation to problems of the form \eqref{eqn:S_update_ADMM} can be found using an alternating direction method of multipliers (ADMM) algorithm \citep{boyd2011}. 
The ADMM scheme consists of the iterates
\begin{equation}\label{eqn:S_admm_update}
  \boldsymbol{B}_{\text{update}}  \leftarrow \min_{\boldsymbol{B}}\left( \lambda_{D+1}l(\boldsymbol{B}) + \gamma \|\boldsymbol{B} - \boldsymbol{Z}^\prime +\boldsymbol{A}^{*}\|_{F}^2\right)
\end{equation}
\begin{equation}\label{eqn:aux_admm_update}
  \boldsymbol{Z}_{\text{update}}  \leftarrow \min_{\boldsymbol{Z}}\left( \|\boldsymbol{G}_{(D+1)} - \boldsymbol{W}_{D+1}\boldsymbol{Z}\|_{F}^2 + \gamma \|\boldsymbol{B} - \boldsymbol{Z}^\prime +\boldsymbol{A}^{*}\|_{F}^2\right)
\end{equation}
\begin{equation}\label{eqn:dual_admm_update}
  \boldsymbol{A}^{*}_{\text{update}} \leftarrow \boldsymbol{A}^{*} +  \boldsymbol{B} - \boldsymbol{Z}^\prime
\end{equation}
for some choice of $\gamma > 0$, where $\boldsymbol{A}^{*}$ is the scaled dual variable associated with the constraint. Since $l$ is assumed to be convex, the ADMM iterates are guaranteed to converge.
\par
The update \eqref{eqn:aux_admm_update} is a matrix ridge regression and has analytic solution given by
\begin{equation}\label{eqn:aux_admm_update_closed_form}
  \boldsymbol{Z}_{\text{update}}  = [\boldsymbol{W}^{\prime}_{D+1}\boldsymbol{W}_{D+1} + \gamma\boldsymbol{I}]^{-1}[\boldsymbol{W}^{\prime}_{D+1}\boldsymbol{G}_{(D+1)} + \gamma (\boldsymbol{B}+\boldsymbol{A}^{*})^\prime].
\end{equation}
The update \eqref{eqn:S_admm_update} defines the so-called \textit{proximal operator} of $l$ and is uniquely minimized. The exact solution will depend on the form of $l$, but it can be shown that many reasonable choices permit an analytic result. For example, if $l(\cdot) = \|\cdot\|_1$, the update is given by the element-wise soft thresholding operator applied to matrix $\boldsymbol{Z}^\prime - \boldsymbol{A}^{*}$.
\par
Algorithm 
1 \revised{in the Supplemental Material} provides pseudocode for estimating the $K$-oMPB utilizing the block coordinate descent scheme, \revised{referred to from here on as \texttt{MARGARITA} (\underline{MARG}inal-product b\underline{A}sis \underline{R}epresentation w\underline{I}th \underline{T}ensor \underline{A}nalysis).} The convergence of \revised{\texttt{MARGARITA}} to a stationary point can be guaranteed if each of the sub-problems is convex and has a unique solution \citep{bertsekas1997}. The former property is satisfied by our construction, while the latter is difficult to verify in practice but can be enforced with minor augmentations. In particular, adding an additional proximal regularization of the form $\frac{\mu_d^{(r)}}{2}\left\|\boldsymbol{X} - \boldsymbol{\tilde{C}}_d^{(r)}\right\|_{F}^2$, for $\mu_{d}^{(r)} > 0$ to \eqref{eqn:BCD} guarantees strong convexity, and hence convergence.
\par
We conclude this section with several remarks on practical implementation. Forming the matrix products $\boldsymbol{W}_{d}^\prime\boldsymbol{W}_{d}$ and $\boldsymbol{W}_d^\prime\boldsymbol{G}_{(d)}$ can become computationally expensive when $D$ and/or $m_d$ become sufficiently large. To avoid this computational bottleneck, the former can be calculated efficiently by leveraging the identity $[\bigodot_i A_i]^\prime[\bigodot_i A_i] = \bigcirc_{i}A_i^\prime A_i$, where $\circ$ is the Hadamard product. Algorithms for efficient computation of the latter have been developed, see \cite{phan2013}. Following the suggestion of \cite{sidiropoulos2016}, we found success setting $\gamma = \|\boldsymbol{W}_{D+1}^\prime\boldsymbol{W}_{D+1}\|_{F}/K$. Finally, if $l(\cdot)=\|\cdot\|_{F}^2$, then it is easy to show that \eqref{eqn:s_update} has a closed form solution and thus the ADMM scheme need not be invoked for this special case. 
\subsection{Hyperparameter Selection}\label{ssec:hyperparameters}
\revised{A distinct advantage of our methodology is its flexibility in allowing the user to incorporate different notions of smoothness, via linear differential operator $L_d$, different choices of marginal basis systems and alternative coefficient penalty methods. In some cases, these can be selected using a-priori knowledge of the problem of interest, though for many applications a data-driven approach to hyperparameter selection may be of interest or of necessity. While our method provides the flexibility of setting $\lambda_d$ independently for all $d=1,...,D+1$, using a data-driven method to select all these parameters is computationally infeasible for even moderately large $D$. Therefore, absent a-priori knowledge of different behavior in different dimensions, we suggest setting $\lambda_d = \lambda_f$ for $d=1,...,D$ and selecting the parameters $(\lambda_{f}, \lambda_{D+1})^\prime$ by minimizing the $n$-fold cross-validation error over a 2-dimensional grid. Pseudocode for this scheme is provided in Algorithm 2 
in the Supplemental Materials. Our method also requires the specification of both marginal ranks $\boldsymbol{m}$ and a global rank $K$. We propose the following two proportion of variance explained measures 
$$
PVG(K) = \|\mathcal{G} -\sum_{k=1}^K\boldsymbol{b}_k\otimes(\bigotimes_{d=1}^D\boldsymbol{\tilde{c}_{d,k}})\|_{F}^2/\|\mathcal{G}\|_{F}^2, \quad \text{PVM}(\boldsymbol{m}):=\|\mathcal{Y}\bigtimes_{d=1}^D\boldsymbol{U}_d^\prime\|_{F}^2/\|\mathcal{Y}\|_{F}^2,
$$
which can be used along with an elbow type criteria for selection. A detailed elaboration, justification and numerical evaluation for the proposed hyperparameter selection criteria 
is provided in Section S3
of the supplemental text.}

\section{Multidimensional Penalized FPCA}\label{sec:mfpca_proc}

In this section, we demonstrate how to leverage the MPB structure to define a fast multidimensional FPCA which avoids the curse of dimensionality, incurring only trivial additional computational expense beyond \revised{\texttt{MARGARITA}}. Consider the method for FPCA proposed in \cite{silverman1996}, in which the $j$th eigenfunction $\psi_j$ is defined as the function maximizing the penalized sample variance with modified orthogonality constraints
\begin{equation}\label{eqn:penalized_sample_variance}
  \begin{gathered}
    \hat{\psi}_{j} = \max_{\psi \in \mathbb{W}^{\alpha, 2}(\mathcal{M})} \frac{\overset{N}{\underset{i=1}{\sum}}\text{Var}(\langle \psi, U_{i}\rangle_{\mathcal{H}})}{\langle\psi,\psi\rangle_{\lambda}} \\
    \textrm{s.t.} \quad  \|\psi\|_{\mathcal{H}}^2 = 1, \qquad \langle \psi, \psi_k \rangle_{\lambda} = 0, \text{ for } k =1, 2, ..., j-1.
  \end{gathered}
\end{equation}
Here $\langle \psi, \psi_k \rangle_{\lambda} := \langle \psi, \psi_k\rangle_{\mathcal{H}} + \lambda \langle L(\psi_j), L(\psi_k)\rangle_{\mathcal{H}}$ and  $L:\mathbb{W}^{\alpha, 2}(\mathcal{M})\rightarrow \mathcal{H}$ is an $\alpha$'th order linear differential operator quantifying the global roughness. For simplicity,  hereafter we define $L:=\Delta_{\mathcal{M}}$, the Laplacian operator on $\mathcal{M}$, though other linear differential operators can be incorporated effortlessly. In our set-up, $L$ facilitates the optional incorporation of a flexible global notion of smoothness in addition to the marginally independent regularization in Equation~\eqref{eqn:Common_Basis_Optimizer_Transform_Tilde_with_regularization}, e.g. penalizing mixed partial derivatives.
\par
In the 1-dimensional case, the optimization problem \eqref{eqn:penalized_sample_variance} is solved using a two-stage approach: first computing $\widehat{U}_i$ through expansion over some suitable basis system and then looking for solutions $\hat{\psi}_j$ in the span of that set of basis functions. Analogously, we can first represent the realizations with the $K$-oMPB: $\widehat{U}_i(\boldsymbol{x})=\boldsymbol{b}_i^\prime\boldsymbol{\zeta}^{*}_{\boldsymbol{m}}(\boldsymbol{x})$, and then solve Equation~\eqref{eqn:penalized_sample_variance} with the additional constraint $\psi_j \in \text{span}(\boldsymbol{\zeta}^{*}_{\boldsymbol{m}})$, i.e. $\psi_j(\boldsymbol{x}) = \boldsymbol{s}_j^{\prime}\boldsymbol{\zeta}^{*}_{\boldsymbol{m}}(\boldsymbol{x})$ for some $\boldsymbol{s}_j\in\mathbb{R}^{K}$. Under this setup, the optimization problem~\eqref{eqn:penalized_sample_variance} is equivalent to
\begin{equation}\label{eqn:penalized_sample_variance_finite_dim}
  \begin{gathered}
    \boldsymbol{s}_{j} = \max_{\boldsymbol{s}} \frac{\boldsymbol{s}^\prime \boldsymbol{J}_{\boldsymbol{\zeta}^{*}_{\boldsymbol{m}}}\boldsymbol{\Sigma}_{\boldsymbol{b}}\boldsymbol{J}_{\boldsymbol{\zeta}^{*}_{\boldsymbol{m}}}\boldsymbol{s}}{\boldsymbol{s}^\prime \boldsymbol{J}_{\boldsymbol{\zeta}^{*}_{\boldsymbol{m}}}\boldsymbol{s} + \lambda \boldsymbol{s}^\prime \boldsymbol{R}_{\boldsymbol{\zeta}^{*}_{\boldsymbol{m}}}\boldsymbol{s}} \\
    \textrm{s.t.} \quad \boldsymbol{s}^\prime \boldsymbol{J}_{\boldsymbol{\zeta}^{*}_{\boldsymbol{m}}}\boldsymbol{s} = 1, \quad  \boldsymbol{s}^\prime[\boldsymbol{J}_{\boldsymbol{\zeta}^{*}_{\boldsymbol{m}}} + \lambda \boldsymbol{R}_{\boldsymbol{\zeta}^{*}_{\boldsymbol{m}}}]\boldsymbol{s}_{k} = 0, \text{ for } k =1, 2, ..., j-1.
  \end{gathered}
\end{equation}
Here $\boldsymbol{\Sigma}_{\boldsymbol{b}} = \text{Cov}(\boldsymbol{b})$ and $\boldsymbol{J}_{\boldsymbol{\zeta}^{*}_{\boldsymbol{m}}}$, $\boldsymbol{R}_{\boldsymbol{\zeta}^{*}_{\boldsymbol{m}}}$ are symmetric PSD matrices with elements $[\boldsymbol{J}_{\boldsymbol{\zeta}^{*}_{\boldsymbol{m}}}]_{ij} = \langle \zeta^{*}_i, \zeta^{*}_j\rangle_{\mathcal{H}}$ and $[\boldsymbol{R}_{\boldsymbol{\zeta}^{*}_{\boldsymbol{m}}}]_{ij} = \langle \Delta_{\mathcal{M}}(\zeta^{*}_i), \Delta_{\mathcal{M}}(\zeta^{*}_j) \rangle_{\mathcal{H}}$, respectively. The objective function in Equation~\eqref{eqn:penalized_sample_variance_finite_dim} is a generalized Rayleigh quotient and it can be shown that the solutions for $j=1,...,K^{\ddagger}$ are equivalently defined by the first $K^{\ddagger}$ solutions to the generalized eigenvalue problem
\begin{equation}\label{eqn:generalized_eigenvalue_problem}
  \boldsymbol{J}_{\boldsymbol{\zeta}^{*}_{\boldsymbol{m}}}\boldsymbol{\Sigma}_{\boldsymbol{b}}\boldsymbol{J}_{\boldsymbol{\zeta}^{*}_{\boldsymbol{m}}}\boldsymbol{s}_{j} = \nu_j[\boldsymbol{J}_{\boldsymbol{\zeta}^{*}_{\boldsymbol{m}}} + \lambda\boldsymbol{R}_{\boldsymbol{\zeta}^{*}_{\boldsymbol{m}}}]\boldsymbol{s}_{j},
\end{equation}
hence, the vector of estimated eigenfunctions is $\boldsymbol{\hat{\psi}}(\boldsymbol{x}) := (\boldsymbol{s}_{1}^\prime\boldsymbol{\zeta}^{*}_{\boldsymbol{m}}(\boldsymbol{x}), ..., \boldsymbol{s}_{K^{\ddagger}}^\prime\boldsymbol{\zeta}^{*}_{\boldsymbol{m}}(\boldsymbol{x}))^\prime$. 
\par 
There are a variety of algorithms to solve the generalized eigenvalue problem \eqref{eqn:generalized_eigenvalue_problem}. In practice, we use Algorithm 9.4.2 in \cite{ramsay2005}, which requires the computation of the marginal inner product matrices:
$
  \boldsymbol{J}_{\phi_{d}}(i,j) = \langle \phi_{d,i}, \phi_{d,j}\rangle_{\mathcal{H}_d}$,
  $\boldsymbol{R}_{\phi_{d}}(i,j) = \langle \Delta_{\mathcal{M}_d}(\phi_{d,i}), \Delta_{\mathcal{M}_d}(\phi_{d,j})\rangle_{\mathcal{H}_d}$, and 
  $\boldsymbol{E}_{\phi_{d}}(i,j) = \langle \phi_{d,i}, \Delta_{\mathcal{M}_d}(\phi_{d,j})\rangle_{\mathcal{H}_d}$.
Given the $\boldsymbol{C}_d$'s, simple derivations show that 
the marginal product structure of the $\zeta_k$'s permits fast analytic computation of $\boldsymbol{J}_{\boldsymbol{\zeta}^{*}_{\boldsymbol{m}}}$ and $\boldsymbol{R}_{\boldsymbol{\zeta}^{*}_{\boldsymbol{m}}}$ based on the element-wise formulas
\begin{equation}\label{eqn:assemble_Jxi}
  \boldsymbol{J}_{\boldsymbol{\zeta}^{*}_{\boldsymbol{m}}}(i,j) = \prod_{d=1}^D \boldsymbol{c}_{d,i}^\prime\boldsymbol{J}_{\phi_{d}}\boldsymbol{c}_{d,j}
\end{equation}
\begin{equation}\label{eqn:assemble_Rxi}
    \boldsymbol{R}_{\boldsymbol{\zeta}^{*}_{\boldsymbol{m}}}(i,j) =  \sum_{d=1}^D(\prod_{b\neq d}^D \boldsymbol{c}^\prime_{b,i}\boldsymbol{J}_{\phi_{b}}\boldsymbol{c}_{b,j})\boldsymbol{c}_{d,i}^\prime\boldsymbol{R}_{\phi_{d}}\boldsymbol{c}_{d,j}  + \sum_{\underset{a\neq d}{a, d}}(\prod_{\underset{b \neq d}{b \neq a}}^D\boldsymbol{c}^\prime_{b,i}\boldsymbol{J}_{\phi_{b}}\boldsymbol{c}_{b,j})(\boldsymbol{c}_{d,i}^\prime\boldsymbol{E}_{\phi_{d}}\boldsymbol{c}_{d,j})(\boldsymbol{c}^\prime_{a,i}\boldsymbol{E}_{\phi_{a}}\boldsymbol{c}_{a,j}).
\end{equation}
Notably, due to the marginal product structure of $\boldsymbol{\zeta}^{*}_{\boldsymbol{m}}$, the $D$-dimensional integrals and partial derivatives required for the computation of $\boldsymbol{J}_{\boldsymbol{\zeta}^{*}_{\boldsymbol{m}}}$ and $\boldsymbol{R}_{\boldsymbol{\zeta}^{*}_{\boldsymbol{m}}}$ decompose into simple sums and products of integrals and partial derivatives over the marginal spaces. In contrast, computing such quantities for an arbitrary $D$-dimensional function is computationally prohibitive for moderately large $D$. This highlights an important practical advantage of working with the marginal product structure: it facilitates efficient computation of $D$-dimensional integrals and partial derivatives which can serve as primitives for developing fast two-stage algorithms for more complex FDA procedures.
\par 
In practice, we form estimates $\widehat{\boldsymbol{C}}_d$ using \revised{\texttt{MARGARITA}} and then estimate the inner product matrices $\widehat{\boldsymbol{J}}_{\boldsymbol{\zeta}^{*}_{\boldsymbol{m}}}$ and $\widehat{\boldsymbol{R}}_{\boldsymbol{\zeta}^{*}_{\boldsymbol{m}}}$ by plugging  $\widehat{\boldsymbol{C}}_d$ into \eqref{eqn:assemble_Jxi} and \eqref{eqn:assemble_Rxi}, respectively. Standard FDA techniques for rank and penalty parameter selection can be adopted to select $K^{\ddagger}$ and $\lambda$.

\section{Simulation Study}\label{sec:simulation_studies}

\subsection{Representing Random Marginal Product Functions}\label{ssec:representing_random_sample}
In this section, we compare three methods for constructing the functional representation of a random sample generated from a marginal product functional model: 1) a TPB system estimated by the sandwich smoother \citep{xiao2013}, 2) the FCP-TPA algorithm \citep{allen2013a}, and 3) the $K$-oMPB estimated using \texttt{MARGARITA}. The two competitors are widely used for multidimensional function representation, see Section S4 
of the Supplementary Text for more details.
\par
The random function in our simulation is defined by the marginal product form:
$
U(\boldsymbol{x}) = \sum_{k=1}^{K^{t}} A_k^{t}\prod_{d=1}^D \left(\boldsymbol{c}^{t}_{d,k}\right)' \boldsymbol{\phi}_j^{t}(x_d).
$
Here $\boldsymbol{\phi}^t_j$ is the period-1 Fourier basis, $\boldsymbol{c}^{t}_{d,k}$ is the $k$th column vector of $\boldsymbol{C}^{t}_d$, the fixed marginal factor matrix such that each element is an $i.i.d.$ sample from $\mathcal{N}(0, 0.3^2)$; and $(A_1^t, ..., A_K^t)^\prime \sim \mathcal{N}(\boldsymbol{0}, \boldsymbol{\Sigma}_{A}^t)$. The covariance matrix is constructed as $\boldsymbol{\Sigma}_{A}^t = \boldsymbol{O}\boldsymbol{D}\boldsymbol{O}^\prime$, where $\boldsymbol{O}$ is a random $K^t\times K^t$ orthogonal matrix, and $\boldsymbol{D}$ is a diagonal matrix with $\boldsymbol{D}_{kk} = \text{exp}(-0.7k)$ for $k=1,\dots, K^{t}$. We took the function domain to be the unit cube $\mathcal{M} = [0,1]^3$. We fixed the true marginal basis dimensions to be $m_d^{t}=11$ for all $d$ and considered true ranks $K_t=10 \text{ and } 20$.
\par
For both ranks, all combinations of the following sampling settings are considered. High vs low SNR; obtained by taking of $\sigma^2$ to be $0.5$ or $10$, small vs. large domain sample size; $n_d=30$ or $50$ for all $d$, respectively, and small vs. large subject sample size; where $N$ is taken to be 5 or 50, respectively. For each of these settings, 100 replications are simulated according to Model~\eqref{eqn:statistical_model}. The performance of the fitting methods are assessed by computing the mean integrated squared error (MISE) for each replication $r$:
$\text{MISE}^{(r)} = \sum_{i=1}^{N^{(r)}} \int_{[0,1]^3} \left[ U^{(r)}_i(\boldsymbol{x}) - \widehat{U}^{(r)}_i(\boldsymbol{x}) \right]^2 \ud\boldsymbol{x},$
where $\widehat{U}^{(r)}_i$ is an estimate of $U^{(r)}_i$ from the $r$th simulated dataset. Denote the Monte Carlo average of the MISE as $\text{moMISE} = 100^{-1}\sum_{r=1}^{100} \text{MISE}^{(r)}$. \revised{For fitting, the second order derivative was used to define the marginal roughness penalties and a ridge penalty was used for regularization on the coefficients. Cubic b-splines were used as the marginal basis system.}
\par
\revised{
We begin by investigating the performance as a function of rank for each combination of $m_d\in\{15,25\}$ and $K_{\text{fit}}\in\{8, 15, 25\}$. A fair comparison between the TPB and \texttt{MARGARITA} should be based on enforcing
(roughly) equivalently sized parameter spaces, i.e. total number of degrees of freedom, so for TPB we take the smallest integer $m_d^{(TPB)}$ such that $\prod_{d=1}^D m_d^{(TPB)} \ge K_{\text{fit}}\sum_{d=1}^D m_d$ for comparison. To isolate the effects of the ranks, for each simulated dataset the models are estimated over a grid of smoothing parameters and the performance of the model with the lowest MISE is recorded. Section S5 
in the Supplementary Text presents a comprehensive comparison of the moMISE for each simulation setting and model parameterization. The results demonstrate that \texttt{MARGARITA} outperforms its competitors consistently, particularly in comparison to TPB fits that have similar degrees of freedom.}

\begin{figure}[t]
  \includegraphics[scale=0.37]{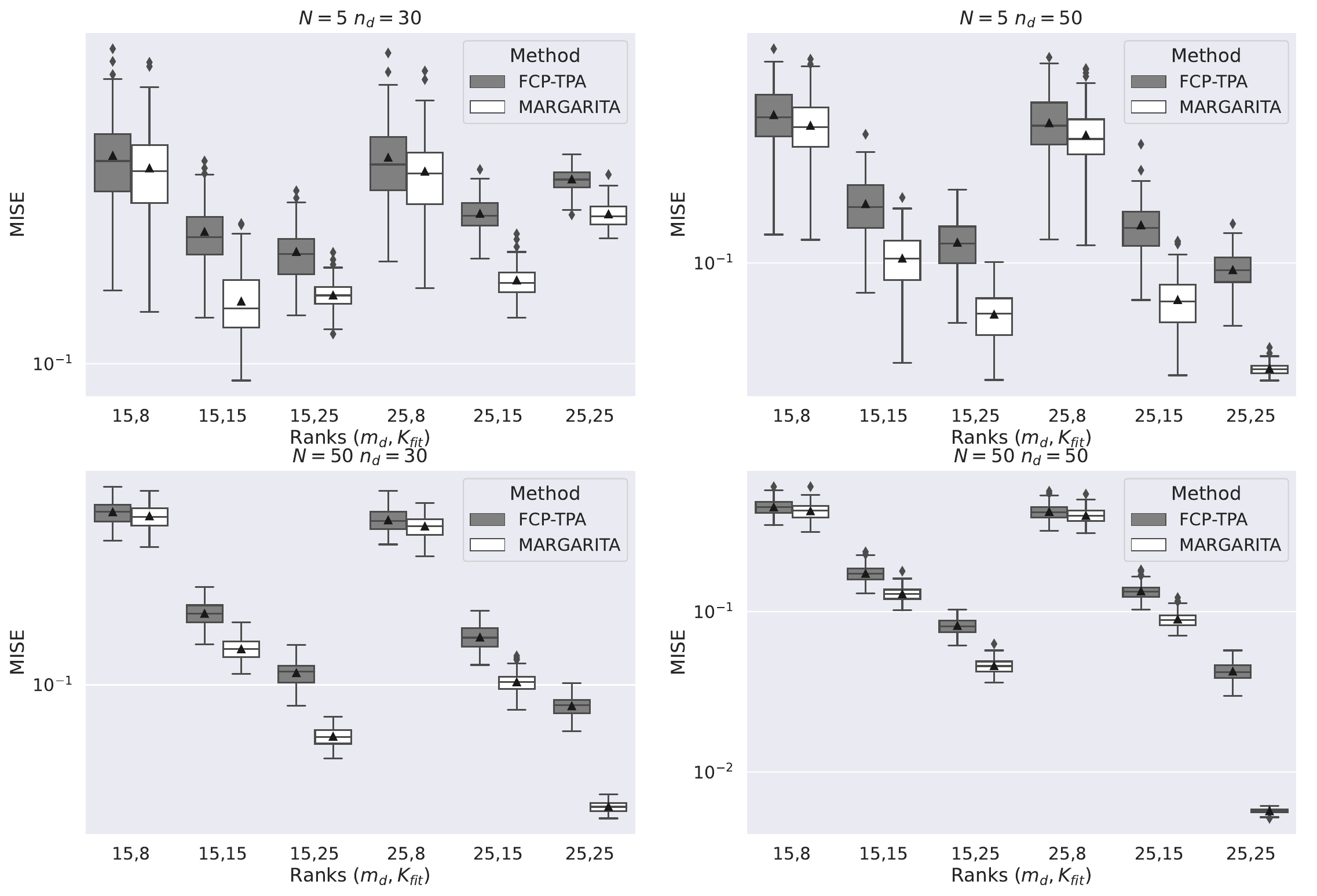}
  \caption{MISE of the fits resulting from both FCP-TPA (gray) and \texttt{MARGARITA} (white). moMISE is denoted by a triangle. The Y-axis is plotted on log-scale for clarity.}
  \label{fig:mpf_mise_comparison}
\end{figure}
\par 
\revised{Figure \ref{fig:mpf_mise_comparison} presents a comparison of FCP-TPA and \texttt{MARGARITA} for each combination of $m_d$ and $K_{fit}$, for several combinations of $N$ and $n_d$, with $K_t=20$ and $\sigma^2=10$. While we see that in all cases \texttt{MARGARITA} results in fits with lower moMISE than FCP-TPA, we also observe that the ranks of the model have a significant impact on the performance. In practical settings, it is often necessary to automate the selection of these ranks as well as the smoothing parameters. Therefore, we compared the automated hyperparameter selection strategies for \texttt{MARGARITA} to the competitors automated smoothing approaches. Specifically, for the TPB method, we selected the smoothing parameters by minimizing the GCV criterion from \cite{xiao2013}. For FCP-TPA, we implemented a $D$-dimensional extension of the nested cross validation method from \cite{huang2009}, as suggested by the authors in \cite{allen2013a}. For our method, we selected penalty parameters using the cross validation scheme outlined in Algorithm 2 
of the Supplemental Text. To focus our analysis, we consider the large domain large sample case ($n_{d}=50, N=50)$, with true rank $K_t=20$ for both low and high SNRs ($\sigma^2=10,\sigma^2=0.5$) for 100 replications.}
\begin{table}
\centering
\begin{tabular}{lrrr}
\toprule
Method &  FCP-TPA &  MARGARITA &     TPB \\
\midrule
High SNR &   $0.0729 \pm 0.0009$ &     $0.0418 \pm 0.0004$ &  $0.5927 \pm  0.0060$ \\
Low SNR  &   $0.0886 \pm 0.0009$ &     $0.0458 \pm 0.0004$ &  $0.6681 \pm 0.0061$\\
\bottomrule
\end{tabular}
\caption{Monte Carlo average MISE for the $n_{d}=50, N=50,K_t=20$ regime for both high and low SNRs. Each methods proposed automatic penalty parameter selection method was used for estimation.}
\label{tab:automated_hyperparam_mise_comparison}
\end{table}
\par 
\revised{We use an elbow criteria to select the marginal ranks and set a threshold of $\text{PVG}(K)\ge 99.5\%$ for global rank selection. The Monte-Carlo averages and standard error of these quantities are plotted for a range of $m$ and $K$ in the top right panels of Figures S1 
and S2, respectively, found in the Supplemental Materials. We consistently identify a clear elbow at $\text{PVM}(m)=15$ for both SNRs, which is in line with the results in Figure~\ref{fig:mpf_mise_comparison} showing a significant increase in performance for $m_d=15$ compared to $m_d=8$, while the performance boost from $m_d=15$ to $m_d=25$ is less pronounced. A $K_{fit}=25$ is consistently selected across simulations. These ranks are fixed for subsequent comparison of the performance of smoothing parameter selection. Table~\ref{tab:automated_hyperparam_mise_comparison} records the moMISE and accompanying standard errors for all methods, showing that \texttt{MARGARITA}'s automatic hyperparameter augmentation outperforms the competitors and is robust to noise. Furthermore, comparing these results to the corresponding results in the bottom right panel of Figure~\ref{fig:mpf_mise_comparison}, we observe that our automated hyperparameter selection estimates models with similar performance to the ones obtained by selecting the oracle best fits over the hyperparameter grid.}
\par
\revised{Due to the super high-dimensional settings encountered, computational efficiency is as important a consideration as estimation performance in multidimensional FDA. Figure S4 in the Supplemental Text compares the computational time of FCP-TPA and \texttt{MARGARITA} for different simulation settings. We find that while both algorithms are comparable in computational speed for small $N$ and small $n_d$, \texttt{MARGARITA} outperforms FCP-TPA as $N$ and $n_d$ increase. This trend is expected, since increasing $n_{d}$ does not increase the dimension of the optimization problem \eqref{eqn:Common_Basis_Optimizer_Transform_Tilde_with_regularization}, while the factors estimated with FCP-TPA are of dimension $n_{d}$, and thus the computational performance of the method can be expected to degrade as $n_{d}$ increases and ultimately become infeasible in the fine grid limit.}
\subsection{Generalization Performance}\label{ssec:sim_generalization_error}
\begin{figure}
    \centering
    \includegraphics[scale=0.47]{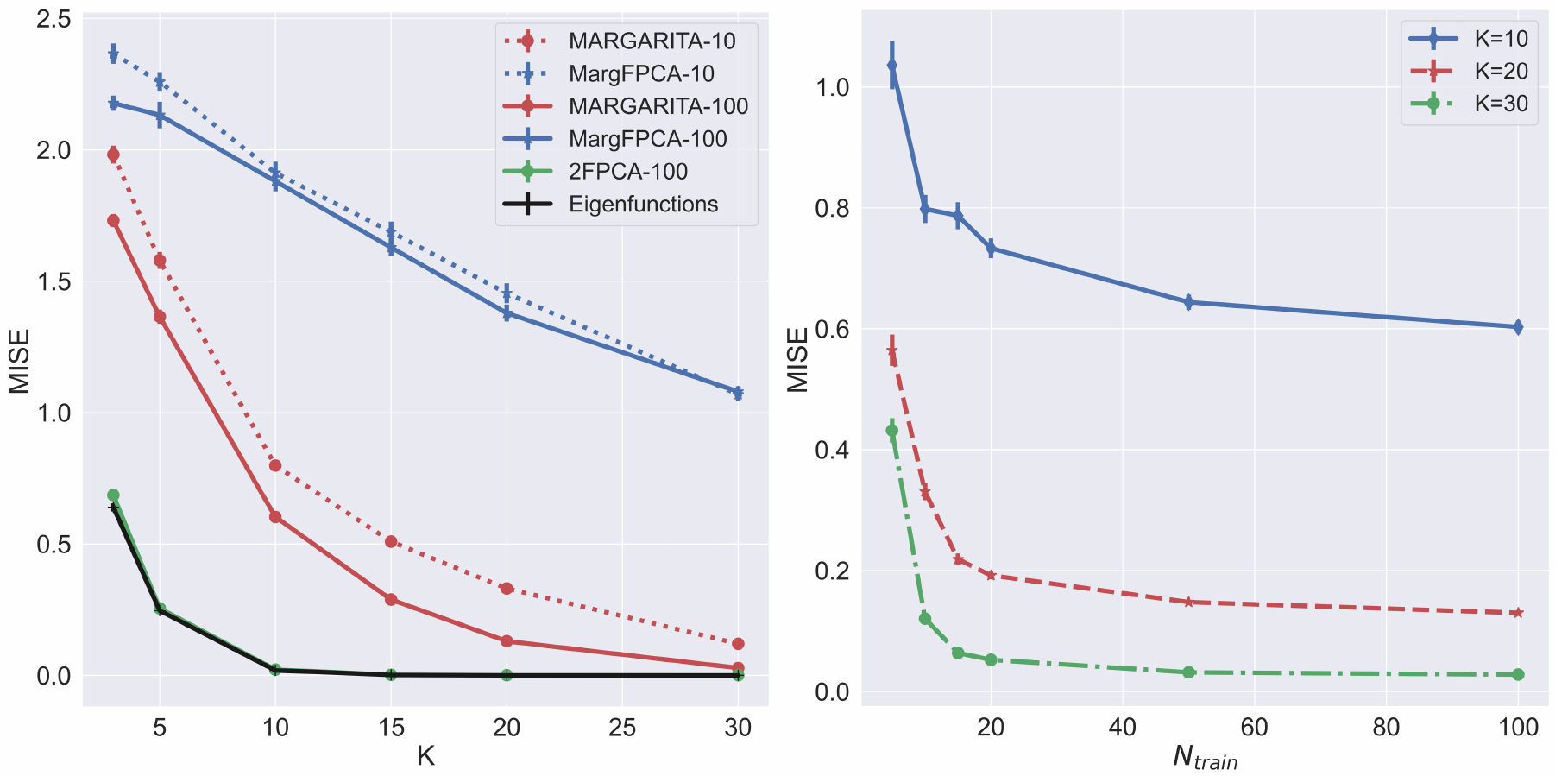}
    \captionof{figure}{(left) Comparison of the generalization performance as a function of $K$ for both \texttt{MARGARITA} (red) and \texttt{MargFPCA} (blue), for $N_{train}=10$ (dotted) $N_{train}=100$ (solid), as well as the 2-stage estimated (green) and true (black) eigenfunctions. (right) Generalization error of \texttt{MARGARITA} as a function $N_{train}$ for several $K$.}
    \label{fig:MARGARITA_vs_MARG_FPCA}
\end{figure}
In this section, we consider the generalization performance of \texttt{MARGARITA}, that is, how efficiently the $K$-oMPB estimated from a training sample of size $N_{train}$ represents new realizations from the same distribution. The results of Section~\ref{sec:model_theory} indicate that we can expect near optimal generalization performance, with an inefficiency due to a ``separability cost'' that vanishes for increasing $K$. We compare our a method to the marginal product FPCA procedure proposed in \cite{muller2017}, referred to here as \texttt{MargFPCA}, which provides a similar near optimality result. In brief, \texttt{MargFPCA} constructs the marginal basis functions by applying FPCA to smoothed estimates of the marginal covariance functions. 
\par 
The development of \texttt{MargFPCA} focuses on the $D=2$ case, so in this study we let the functional domain be $\mathcal{M}=[0,1]^2$ \revised{and evaluate the generalization error of both \texttt{MARGARITA} and \texttt{MargFPCA} as a function of  $N_{train}$ and $K_{fit}$. We define random function $U$ to be a non-stationary, non-separable anisotropic Gaussian process which is observed over an equispaced $200\times 200$ grid on $\mathcal{M}$. For each combination of $N_{train}$ and rank $K_{fit}$, both \texttt{MARGARITA} and \texttt{MargFPCA} are used to construct the representations for each of 50 realizations from an independent test set using least squares basis expansion. Each experimental set-up is repeated for $25$ replications. Additional details on the definition of $U$ and other simulation settings can be found in Supplemental Section S5.2}.  
\par 
Figure~\ref{fig:MARGARITA_vs_MARG_FPCA} (left plot) displays the average MISE on the test set, i.e. the generalization error, as a function of $K$ for both \texttt{MARGARITA} (red) and \texttt{MargFPCA} (blue). The dotted and solid lines correspond to $N_{train}=10$ and $N_{train}=100$, respectively. For both training sample sizes, we observe that our method both uniformly outperforms \texttt{MargFPCA} for all ranks considered and displays much faster convergence in $K$. The green line shows the generalization performance of the eigenfunctions estimated using the two-stage FPCA procedure outlined in Section~\ref{sec:mfpca_proc} for $N_{train}=100$, with an initial \texttt{MARGARITA} of rank 60. The performance is nearly identical with that of the true eigenfunctions (black). \revised{Table S4 and Figure S5 in the Supplemental Material evaluate the two-stage estimates of first three eigenfunctions, showing accurate recovery as $N$ increases.} The right plot of Figure~\ref{fig:MARGARITA_vs_MARG_FPCA} gives the average MISE as a function of $N_{train}$ for several $K$. Recalling that \texttt{MARGARITA} is only guaranteed to converge to a local solution, these result indicate that, at least in some cases, the local (computable) solution still exhibits good convergence properties. Results for more ranks and training sample sizes are recorded in Section S5 of the Supplemental Materials and yield similar conclusions. 

\section{Real Data Analysis}\label{sec:real_data_analysis}

The white matter (WM) of the human brain consists of large collections of myelinated nueral fibers that permit fast communication between different regions of the brain. Diffusion magnetic resonance imaging (dMRI) is a non-invasive imaging technique which uses spatially localized measurements of the diffusion of water molecules to probe the WM microstructure. At each 3-dimensional voxel in the brain, the diffusion image can be used to compute scalar summaries of local diffusion, e.g. fractional anistropy (FA) or mean diffusivity. The resulting data can be organized as a mode-$3$ tensor. For this application, we consider a dataset consisting of the brain images of 50 subjects in an age matched balanced case-control traumatic brain injury (TBI) study. Previous studies have shown the potential for using FA to identify white matter abnormalities associated with TBI and post concussive syndrome 
\citep{kraus2007}.
Typically, voxel-based analysis are performed for group-wise analysis of FA using Tract-Based Spatial Statistics (TBSS) \citep{smith2006}, though such analysis are often not able to establish significant group differences \citep{khong2016}, partially due to low power resulting from the large voxel-based multiple testing problem. Due to the continuity of the diffusion process, the FA tensor can be considered as discrete noisy observations of an underlying multidimensional random field, hence we may adopt the statistical model in Equation~\eqref{eqn:statistical_model}. In this analysis, we focus on a functional approach to predict disease status and identify regions in the WM which differ significantly between TBI and control. For details on the study design, MRI scanning protocol, and dMRI preprocessing, please visit Section S6 in the Supplementary Material.
\par 
The voxel grid is of size $115\times 140 \times 120$. Point-wise estimates of the mean function at each voxel are obtained using the sample mean tensor, which is then used to center the data. Equispaced cubic b-splines of ranks $m_1=57, m_2=70, m_3=60$, selected using a $90\%$ threshold on the quantity described in Section~\ref{ssec:hyperparameters}, are used as marginal basis systems. Marginal roughness is penalized by the second order derivative and coefficients were regularized with a ridge penalty, with penalty parameters $\lambda_d=10^{-10}$ for $d=1,2,3$ and $\lambda_4=10^{-8}$. A rank $K=500$ model is estimated from the mean centered data tensor using \texttt{MARGARITA}. FPCA is then performed on the represented data using the fast two-stage approach outlined in Section~\ref{sec:mfpca_proc}. The first $45$ eigenfunctions, denoted collectively as $\boldsymbol{\psi}$, explain $\approx 99\%$ of the represented variance and are used in constructing the final continuous representations of data. A lasso penalized logistic regression classifier is trained to predict disease status using the subject coefficient vectors obtained by their representation over $\boldsymbol{\psi}$. The resulting classification performance is evaluated using leave-one-out cross validation (LOOCV). To localize group differences to particular eigenfunctions, univariate permutation test are performed on the coefficients and the resulting p-values are corrected to maintain a false discovery rate (FDR) $\le 5\%$ using \cite{benjamini1995}. Finally, data-driven regions of interest (ROIs) are defined as spatial volumes where the values of the significant eigenfunctions are ``extreme'', i.e. outside the $0.5\%$ and $99.5\%$ quantiles.
\par
The LOOCV accuracy, precision and recall are 0.96, 1.0, and 0.92, respectively, indicating substantial discriminatory power of the learned basis functions. Additionally, the testing procedure identified significant group differences in the coefficients of three eigenfunctions. For comparison, we applied TBSS to this data and no significant group differences were identified. Figure~\ref{fig:data_driven_ROI} shows two cross sections of the brain, with the data-driven ROIs corresponding to the identified eigenfunctions displayed in blue, red and green. The ROIs in Figure~\ref{fig:data_driven_ROI} (a) are found within areas of the middle cerebellar peduncle (MCP) and, in Figure~\ref{fig:data_driven_ROI} (b), in areas along the superior longitudinal fasciculus (SLF) fiber bundle. \cite{wang2016altered} found increased FA in the MCP is associated with increased cognitive impairment. \cite{xiong2014white} found decreased FA in the SLF in patients with TBI. We note that both of these studies were completed in acute cases of TBI, whereas our data represents a more chronic state of TBI, often called post-concussive syndrome. That being said, these tracts are thought to be altered because of the nature of biophysical forces suffered in TBI. In all TBI, there is rotation of the head around the neck, which causes shearing and stretching of the brain stem tracts. In addition, the longer tracts in the brain, including the SLF, are subject to shearing forces on left to right rotation of the head around the neck. In fact, \cite{post2013examination} found that mechanical strain in the brain stem and cerebellum are significantly correlated with angular acceleration of the brain, suggesting fibers in this area are susceptible to changes related to TBI. Therefore, our findings of changes in the MCP and SLF are consistent with the hypothesized mechanism and previous findings in TBI.
\begin{figure}
  \centering
    \includegraphics[scale=0.41]{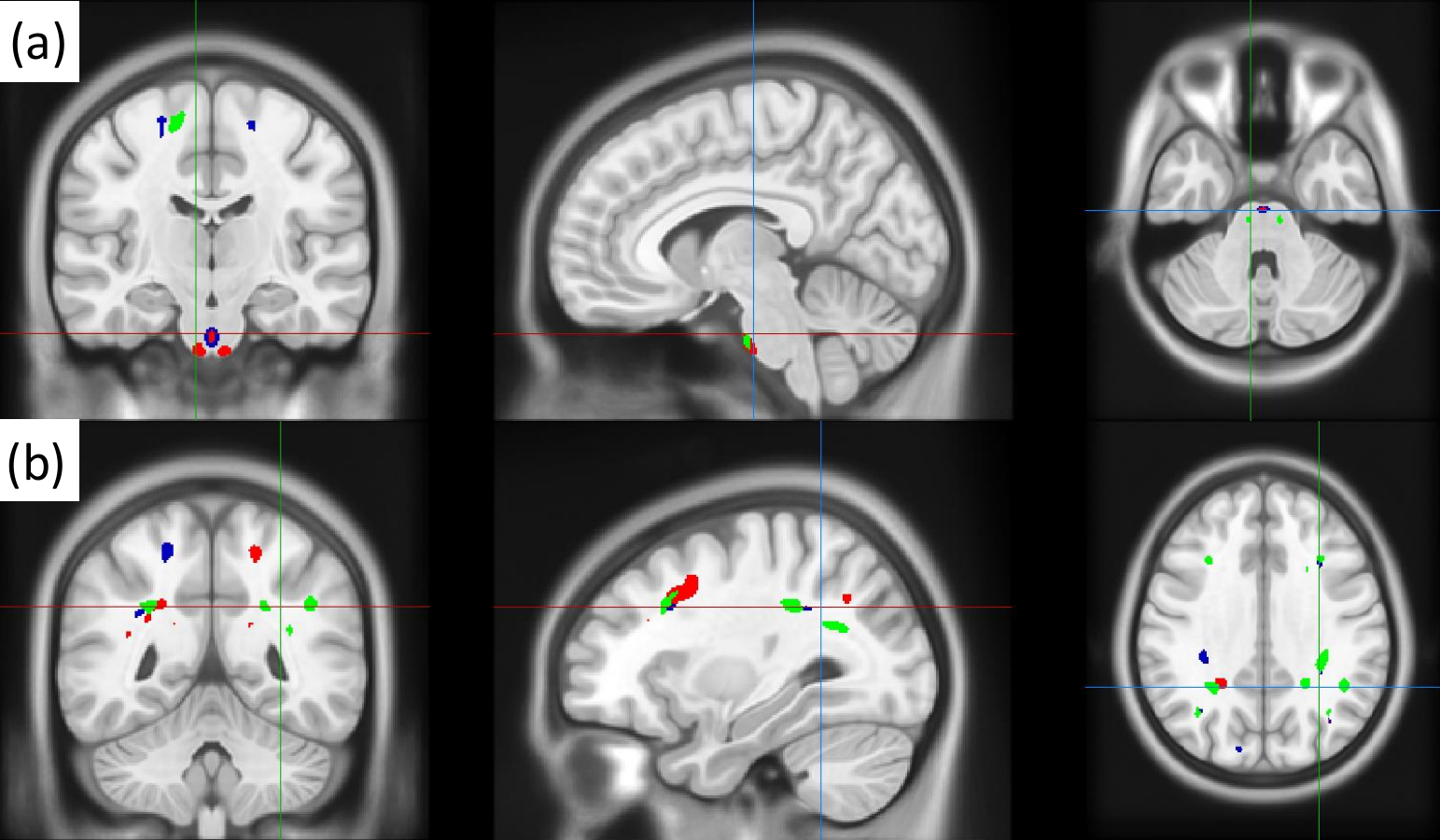}
  \caption{Data-driven ROIs created from thresholding the $0.5\%$ and $99.5\%$ quantiles of the three identified eigenfunctions in blue, red, and green. 
  }
  \label{fig:data_driven_ROI}
\end{figure}

\section{Discussion and Future Work}\label{sec:discussion}
Our work introduces a methodological framework and accompanying estimation algorithm for constructing a flexible and efficient continuous representation of multidimensional functional data. We consider basis functions that exhibit a marginal product structure and prove that an optimal set of such functions can be defined by the penalized tensor decomposition of an appropriate transformation of the raw data tensor. A variety of separable roughness penalties can be used to promote smoothness. Regularized parameter estimation is performed using a block coordinate descent scheme and we describe globally convergent numerical algorithms for solving the subproblems. Using extensive simulation studies, we illustrate the superiority of our proposed method compared to competing alternatives. In a real data application of the group-wise analysis of diffusion MRI, we show that our method can facilitate the prediction of disease status and identify biologically meaningful ROIs.
\par
\par
This work can be extended in several interesting directions. A principled \revised{and computationally efficient approach to} both a  generalized cross validation criteria and information criteria for faster penalty parameter selection and model-based selection of the global rank, respectively, are of interest. 
\revised{Additionally, many modern functional datasets are observed irregularly over the domain, rather than the common grid we consider here. To use our method on dense irregular data, we can bin the data using a common grid and define the observed data tensor to be the bin-specific sample means for each subject. However, this approach is problematic for sparsely sampled irregular data, requiring further extension of the method.}

\bibliographystyle{chicago}
\bibliography{references}

\clearpage
\pagebreak

\begin{center}
{\large\bf SUPPLEMENTARY MATERIAL}
\end{center}

\setcounter{equation}{0}
\setcounter{figure}{0}
\setcounter{table}{0}
\setcounter{section}{0}
\setcounter{page}{1}
\makeatletter
\renewcommand{\theequation}{S\arabic{equation}}
\renewcommand{\thefigure}{S\arabic{figure}}
\renewcommand{\thetable}{S\arabic{table}}
\renewcommand{\thesection}{S\arabic{section}}
\renewcommand{\bibnumfmt}[1]{[S#1]}
\renewcommand{\citenumfont}[1]{S#1}
\renewcommand{\theequation}{S.\arabic{equation}}
\renewcommand{\thesection}{S\arabic{section}}
\renewcommand{\thesubsection}{S\arabic{section}.\arabic{subsection}}
\renewcommand{\thetable}{S\arabic{table}}
\renewcommand{\thefigure}{S\arabic{figure}}
\renewcommand{\thetheorem}{S\arabic{theorem}}
\renewcommand{\theproposition}{S\arabic{proposition}}
\renewcommand{\thelemma}{S\arabic{lemma}}
\renewcommand{\theassumption}{S\arabic{assumption}}

\section{Theory and Proofs}\label{sec:theory_and_proofs}

\noindent{\textbf{Additional Definitions, Assumptions and Technical Conditions}}
\par\bigskip
\begin{definition}\label{defn:marginal_convergence_rate}
    Let the function $w_{\phi_{d}}(m_d)$ be the $\mathbb{L}^2(\mathcal{M}_d)$ convergence rate of the $d$th marginal basis system $\boldsymbol{\phi}_d$ and $w_{\bm{\tau}_{m}}(\boldsymbol{m})$ be the $\mathbb{L}^2(\mathcal{M})$ convergence rate of the tensor product basis system $\bm{\tau}_{\boldsymbol{m}}$. That is, for any $f_d \in \mathcal{H}_d$, $f \in \mathcal{H}$
    $$
    \left\|P_{\mathcal{H}_{m_{d},d}^{\perp}}(f_d)\right\|_{\mathcal{H}_d} = O(w_{\phi_{d}}(m_d)), \qquad\left\|P_{\mathcal{H}_{\boldsymbol{m}}^{\perp}}(f)\right\|_{\mathcal{H}} = O(w_{\bm{\tau}_{\boldsymbol{m}}}(\boldsymbol{m}))
    $$
    where $P_{\mathcal{H}_{m_{d},d}^{\perp}}$, $P_{\mathcal{H}_{\boldsymbol{m}}^{\perp}}$ are the projection operators onto ${\mathcal{H}_{m_{d},d}^{\perp}}$ and ${\mathcal{H}_{\boldsymbol{m}}^{\perp}}$, the orthogonal complements of $\mathcal{H}_{m_{d},d}$ in $\mathcal{H}_d$ and $\mathcal{H}_{\boldsymbol{m}}$ in $\mathcal{H}$, respectively.
\end{definition}
\begin{definition}
 For ease of presentation, we define the inner product space $(\bigotimes_{d=1}^D\mathbb{R}^{m_{d}}, \langle\cdot,\cdot\rangle_{\tilde{F}})$, where $\langle\mathcal{T}_1,\mathcal{T}_2\rangle_{\tilde{F}} = \langle\mathcal{T}_1,\mathcal{T}_2\times_{1}\boldsymbol{J}_{\boldsymbol{\phi}_{1}}\cdots\times_{D}\boldsymbol{J}_{\boldsymbol{\phi}_{D}} \rangle_{F}$ for $\mathcal{T}_1,\mathcal{T}_2\in\bigotimes_{d=1}^D\mathbb{R}^{m_{d}}$
\end{definition}
We assume the following conditions related to the boundedness and tail-behavior of $U$.
\begin{assumption}\label{asm:spectral_decay_rate}
     $$
     \begin{aligned}
       &\text{(i) } \sum_{k=K+1}^\infty\sqrt{\rho}_k = o(1)  \quad 
       \text{(ii) }  \sum_{k=1}^\infty\mathbb{E}\left[|Z_k|^r\right] < \infty, \text{ for } r=3,4\\
     \end{aligned}
     $$
\end{assumption}
Assumption~\ref{asm:spectral_decay_rate}.i introduces a slightly stronger condition on the decay rate of the eigenvalues than the one that comes for free from 
Assumption 1, 
i.e. $\sum_{k=K+1}^\infty\rho_k = o(1)$. Assumption~\ref{asm:spectral_decay_rate}.ii is a technical moment condition which controls the fatness of the ``high-frequency tail'' of $U$. These conditions are satisfied for many standard distributions and covariance kernels.
\par 
To ensure the existence and uniqueness of $\Breve{\boldsymbol{\zeta}}_{\boldsymbol{m},N}^{*}$, we address the identifiability issues resulting from the inherent ambiguity in the ordering (permutation indeterminacy) of the basis functions. Since $\boldsymbol{\zeta}\in\mathcal{V}_{K,\boldsymbol{m}}$, we have the representation $\xi_{k,d}(x_d) = \sum_{j=1}^{m_{d}} c_{d,k,j}\phi_{d,j}(x_d)$, and hence we can identify any $\boldsymbol{\zeta}\in\mathcal{V}_{K,\boldsymbol{m}}$ with the parameter $(\boldsymbol{C}_1,...,\boldsymbol{C}_{D})$, where $\boldsymbol{C}_d\in\mathbb{R}^{m_{d}\times K}$ and  $\boldsymbol{c}_{d,k}$ is the $k$'th column vector. We can now introduce the reparameterized relaxation of the parameter space $\mathcal{V}_{K,\boldsymbol{m}}$: 
$$
\begin{aligned}
   \boldsymbol{\Theta}_{K,\boldsymbol{m}} := &\{(\boldsymbol{C}_1,...,\boldsymbol{C}_{D}): \boldsymbol{c}^\prime_{d,k} \boldsymbol{J}_{\boldsymbol{\phi}_d}\boldsymbol{c}_{d,k} \le 1\text{ for }d=1,...,D;k=1,...,K;\boldsymbol{C}_{1}(1,1) > ... > \boldsymbol{C}_{D}(1,K)\} 
\end{aligned}
$$
which alleviates this identifiability problem. We must also address the ill-posedness of the best constrained rank approximations for $D > 2$ in general \citep{silva2006}. We invoke a sufficient but not necessary condition on $K$ to resolve this issue \citep{bro2000}:
\begin{assumption}\label{asm:existance_assumption}
Let $\mathcal{A}^{(K)}$ be the mode $D+1$ tensor obtained from stacking $\mathcal{A}_1, ..., \mathcal{A}_{K}$ for some finite integer $K$. Suppose it's rank is $K^{*}$. We assume that $K \ge \left(2K^{*} + D\right)/\left(D+1\right)$.
\end{assumption}

\noindent{\textbf{Consistency}}
\par\bigskip
We establish the point-wise consistency of  $\Breve{\boldsymbol{\zeta}}_{\boldsymbol{m},N}^{*}$. Throughout this section, let $0 < R < \infty$ be a generic constant, that is perhaps different depending on context. For notational convenience, we establish the following definition.
\begin{definition}
  Let the function $h(K)$ be convergence rate of the tail-sum of the eigenvalues of the covariance operator associated with $C(\bm{x}, \bm{y})$, that is
  $$
  \sum_{k=K+1}^\infty\rho_k =  O(h(K))
  $$
\end{definition}
\begin{proposition}\label{prop:kOMPB_projection_coefficients}
The coefficients of the projection $P_{\boldsymbol{\zeta}_{\boldsymbol{m}}}(U) := \sum_{k=1}^Kb_k\xi_k$ are given by
$$
  \boldsymbol{b}_K := (b_1, ..., b_K)^\prime = \sum_{l=1}^\infty Z_l\boldsymbol{b}_l 
$$
where 
\begin{equation}\label{eqn:projected_coefs}
\boldsymbol{b}_{l}  =
  \begin{pmatrix}
    \prod_{d=1}^D\boldsymbol{c}_{d,1}^{\prime}\boldsymbol{J}_{\boldsymbol{\phi}_d}\boldsymbol{c}_{d,1} & \prod_{d=1}^D\boldsymbol{c}_{d,1}^{\prime}\boldsymbol{J}_{\boldsymbol{\phi}_d}\boldsymbol{c}_{d,2} & \cdots & \prod_{d=1}^D\boldsymbol{c}_{d,1}^{\prime}\boldsymbol{J}_{\boldsymbol{\phi}_d}\boldsymbol{c}_{d,K} \\
    \prod_{d=1}^D\boldsymbol{c}_{d,2}^{\prime}\boldsymbol{J}_{\boldsymbol{\phi}_d}\boldsymbol{c}_{d,1} & \prod_{d=1}^D\boldsymbol{c}_{d,2}^{\prime}\boldsymbol{J}_{\boldsymbol{\phi}_d}\boldsymbol{c}_{d,2} & \cdots &  \\
    \vdots & & \ddots & \\
    \prod_{d=1}^D\boldsymbol{c}_{d,K}^{\prime}\boldsymbol{J}_{\boldsymbol{\phi}_d}\boldsymbol{c}_{d,1}  & \cdots & & \prod_{d=1}^D\boldsymbol{c}_{d,K}^{\prime}\boldsymbol{J}_{\boldsymbol{\phi}_d}\boldsymbol{c}_{d,K}
  \end{pmatrix}^{-1}
  \begin{bmatrix}
    \langle\mathcal{A}_l, \bigotimes_{d=1}^D\boldsymbol{c}_{d,1}\rangle_{\tilde{F}} \\
    \langle\mathcal{A}_l, \bigotimes_{d=1}^D\boldsymbol{c}_{d,2}\rangle_{\tilde{F}} \\
    \vdots \\
    \langle\mathcal{A}_l, \bigotimes_{d=1}^D\boldsymbol{c}_{d,K}\rangle_{\tilde{F}}
  \end{bmatrix}
\end{equation}
and $^{-1}$ refers to the (generalized) inverse of the inner product matrix.
\end{proposition}
\begin{proof}
  This follows from the definition of the $\mathbb{L}^2$ projection operator. 
\end{proof}
\begin{proposition}\label{prop:KOMPB_ip_bounded}
Under Assumption 1, we have that (i) $\mathbb{E}\left[ |\left\langle U, P_{\boldsymbol{\zeta}_{\boldsymbol{m}}}(U) \right\rangle_{\mathcal{H}}|\right] < R$ and (ii) $\mathbb{E}\left[\left\| P_{\boldsymbol{\zeta}_{\boldsymbol{m}}}(U)\right\|_{\mathcal{H}}^2\right] < R$ for any $\boldsymbol{\zeta}_{\boldsymbol{m}}\in\mathcal{V}_{K,\boldsymbol{m}}$.
\end{proposition}
\begin{proof}
The results follow immediately by noting that $\|U\|_{\mathcal{H}}^2 \ge \|P_{\boldsymbol{\zeta}_{\boldsymbol{m}}}(U)\|_{\mathcal{H}}^2$ and $\mathbb{E}\left[\|U\|_{\mathcal{H}}^2 \right] = \sum_{k=1}^\infty \rho_k< \infty.$
\end{proof}
\begin{lemma}\label{lem:generalization_error}
The expected generalization error of $\boldsymbol{\zeta}_{m}\in\mathcal{V}_{K,\boldsymbol{m}}$ can be written as
\begin{equation}\label{eqn:coef_tensor_decomposition}
    \mathbb{E}\left\|U - P_{\boldsymbol{\zeta}_{\boldsymbol{m}}}(U)\right\|_{\mathcal{H}}^2 = \min_{\boldsymbol{B}} \quad 
    \left\|\mathcal{A}^{(K)} - \sum_{k=1}^K \left[\bigotimes_{d=1}^D\boldsymbol{c}_{d,k}\right]\otimes \boldsymbol{B}_{:,k}\right\|_{\tilde{F}, C}^2 + O(w_{\bm{\tau}_{\boldsymbol{m}}}(\boldsymbol{m})) +  O(h(K))
\end{equation}
where $\boldsymbol{B}_{:,k}$ is the $k$'th column of $\boldsymbol{B}\in\mathbb{R}^{K\times K}$, for some $\boldsymbol{c}_{d,k}\in\mathbb{R}^{m_{d}}$.
\end{lemma}
\begin{proof}
$$
    \begin{aligned}
      \mathbb{E}\left\|U - P_{\boldsymbol{\zeta}_{\boldsymbol{m}}}(U)\right\|_{\mathcal{H}}^2 &= \mathbb{E}\left\|P_{\mathcal{H}_{\boldsymbol{m}}}(U) - P_{\boldsymbol{\zeta}_{\boldsymbol{m}}}(U) + P_{\mathcal{H}_{\boldsymbol{m}}^{\perp}}(U) \right\|_{\mathcal{H}}^2  \\ 
      &= \mathbb{E}\left\|P_{\mathcal{H}_{\boldsymbol{m}}}(U) - P_{\boldsymbol{\zeta}_{\boldsymbol{m}}}(U)\right\|_{\mathcal{H}_{\boldsymbol{m}}} + 
      \\ & + \mathbb{E}\left\langle (P_{\mathcal{H}_{\boldsymbol{m}}}(U) - P_{\boldsymbol{\zeta}_{\boldsymbol{m}}}(U)),P_{\mathcal{H}_{\boldsymbol{m}}^{\perp}}(U)\right\rangle_{\mathcal{H}} \\
      & + \mathbb{E}\left\|P_{\mathcal{H}_{\boldsymbol{m}}^{\perp}}(U) \right\|_{\mathcal{H}_{\boldsymbol{m}}^{\perp}}\\
      & :=  \text{Term}_1 + \text{Term}_2 + \text{Term}_3.
    \end{aligned}
$$

$\text{Term}_{3}$ is independent of $\boldsymbol{\zeta}_{\boldsymbol{m}}$ and represents the expected irreducible error due to the finite dimensional truncation of the marginal basis systems. We have
\begin{equation}\label{eqn:irreducible_error_tpb}
\begin{aligned}
     \text{Term}_3 &= \mathbb{E}\left\|P_{\mathcal{H}^{\perp}_{\boldsymbol{m}}}(U)\right\|_{\mathcal{H}^{\perp}_{\boldsymbol{m}}}^2  = \mathbb{E} \left\|\sum_{k=1}^\infty Z_kP_{\mathcal{H}^{\perp}_{\boldsymbol{m}}}(\psi_k) \right\|_{\mathcal{H}^{\perp}_{\boldsymbol{m}}}^2  \\
     &= \sum_{k=1}^\infty \mathbb{E}\left[Z_k^2\right] \cdot \left\|P_{\mathcal{H}^{\perp}_{\boldsymbol{m}}}(\psi_k)\right\|_{\mathcal{H}^{\perp}_{\boldsymbol{m}}}^2 \\
     &= O(w_{\bm{\tau}_{\boldsymbol{m}}}(\boldsymbol{m})),
\end{aligned}
\end{equation}
where the second line follows from the $Z_k$ being uncorrelated and the third line follows since $\sum_{k=1}^\infty \mathbb{E} \left[Z_k^2 \right] = \sum_{k=1}^\infty\rho_k < \infty$. Since $\text{span}(\boldsymbol{\zeta}_{\boldsymbol{m}})  \subset \mathcal{H}_{\boldsymbol{m}}$, it is easy to see that $\text{Term}_2 = 0$ and thus we need only to deal with $\text{Term}_1$.
\par 
The mapping $\iota:\mathcal{H}_{\boldsymbol{m}}\mapsto \bigotimes_{d=1}^D\mathbb{R}^{m_{d}}$ defined by $\iota(u)_{j_{1},...,j_{D}} = a_{j_{1},...,j_{D}}$ is an isometry between inner product spaces $(\mathcal{H}_{\boldsymbol{m}}, \langle\cdot, \cdot\rangle_{\mathcal{H}_{\boldsymbol{m}}})$ and $(\bigotimes_{d=1}^D\mathbb{R}^{m_{d}}, \langle\cdot,\cdot\rangle_{\tilde{F}})$, where $a_{j_{1},...,j_{D}}$ is the coefficient of $u$ associated with basis element $\prod_{d=1}^D\phi_{d,j_{d}}$. Recall that any $u\in \text{span}(\boldsymbol{\zeta}_{\boldsymbol{m}})$ has the representation 
$$
u(\boldsymbol{x}) = \sum_{k=1}^{K}b_{k}\prod_{d=1}^{D}\sum_{j=1}^{m_{d}} c_{k,d,j} \phi_{d,j}(x_{d})
$$
and hence, under $\iota$, is identified with the tensor rank-$K$ tensor $\sum_{k=1}^K b_k\bigotimes_{d=1}^D\boldsymbol{c}_{d,k}$, where $\boldsymbol{c}_{d,k}$ are the $m_d$-vectors of coefficients for the $k$th marginal function in the $d$th dimension. It follows that  
$$
\begin{aligned}
    \mathbb{E}\left\| P_{\mathcal{H}_{\boldsymbol{m}}}(U) - P_{\boldsymbol{\zeta}_{\boldsymbol{m}}}(U)\right\|_{\mathcal{H}_{\boldsymbol{m}}}^2
    &=  \mathbb{E}\left\|\sum_{l=1}^\infty Z_l\mathcal{A}_l - \sum_{k=1}^Kb_k\bigotimes_{d=1}^D\boldsymbol{c}_{d,k} \right\|_{\tilde{F}}^2 \\
   &= \mathbb{E}\left\|\sum_{l=1}^\infty Z_l\mathcal{A}_l - \sum_{k=1}^K \sum_{j=1}^\infty Z_l b_{j,k}\bigotimes_{d=1}^D\boldsymbol{c}_{d,k}\right\|_{\tilde{F}}^2  \\ 
    &= \mathbb{E}\left\|\sum_{l=1}^\infty Z_l\mathcal{A}_l - \sum_{k=1}^KZ_lb_{l,k}\bigotimes_{d=1}^D\boldsymbol{c}_{d,k}\right\|_{\tilde{F}}^2 \\
    &= \sum_{l=1}^\infty \mathbb{E}\left[Z_l^2\right]\left\|\mathcal{A}_l - \sum_{k=1}^Kb_{l,k}\bigotimes_{d=1}^D\boldsymbol{c}_{d,k}\right\|_{\tilde{F}}^2 \\
    & \quad + \sum_{j\neq r}\mathbb{E}\left[Z_jZ_r\right]\left\langle \mathcal{A}_j - \sum_{k=1}^Kb_{j,k}\bigotimes_{d=1}^D\boldsymbol{c}_{d,k}, \mathcal{A}_r - \sum_{k=1}^Kb_{r,k}\bigotimes_{d=1}^D\boldsymbol{c}_{d,k}\right\rangle_{\tilde{F}} \\ 
    &= \sum_{l=1}^\infty \rho_l\left\| \mathcal{A}_l - \sum_{k=1}^Kb_{l,k}\bigotimes_{d=1}^D\boldsymbol{c}_{d,k}\right\|_{\tilde{F}}^2 \\ 
    &= \sum_{l=1}^K \rho_l\left\| \mathcal{A}_l - \sum_{k=1}^Kb_{l,k}\bigotimes_{d=1}^D\boldsymbol{c}_{d,k}\right\|_{\tilde{F}}^2 + O(h(K)) \\
    &= \min_{\boldsymbol{B}}\sum_{k=1}^K\rho_l\left\| \mathcal{A}_l - \sum_{k=1}^K\boldsymbol{B}_{l,k}\bigotimes_{d=1}^D\boldsymbol{c}_{d,k}\right\|_{\tilde{F}}^2 + O(h(K)) \\ 
    &= \min_{\boldsymbol{B}}\left\| \mathcal{A}^{(K)} - \sum_{k=1}^K\left[\bigotimes_{d=1}^D\boldsymbol{c}_{d,k}\right]\otimes \boldsymbol{B}_{:,k}\right\|_{\tilde{F}, C}^2 + O(h(K)).
\end{aligned}
$$
\end{proof}

\begin{lemma}\label{lem:param_converge_prob}
Let $$
    L_N(\boldsymbol{C}) := N^{-1}\sum_{i}^N\left\|U_i - P_{\boldsymbol{C}}(U_i)\right\|_{\mathcal{H}}^2; \quad L(\boldsymbol{C}) := \mathbb{E}\left\|U - P_{\boldsymbol{C}}(U)\right\|_{\mathcal{H}}^2.
    $$ 
   where $P_{\boldsymbol{C}}$ is the reparameterization of the projection operator $P_{\boldsymbol{\zeta}_{\boldsymbol{m}}}$ for $\boldsymbol{\zeta}_{\boldsymbol{m}}$ defined by $\boldsymbol{C} = (\boldsymbol{C}_1,...,\boldsymbol{C}_{D})\in\boldsymbol{\Theta}_{K,\boldsymbol{m}}$. Define $\Breve{\boldsymbol{C}}_{N}, \boldsymbol{C}^{*}\in \boldsymbol{\Theta}_{K,\boldsymbol{m}}$ to be the minimizers of $ L_N(\boldsymbol{C})$ and $ L(\boldsymbol{C})$, respectively. Then 
  $$
  \Breve{\boldsymbol{C}}_{N}\overset{P}{\rightarrow} \boldsymbol{C}^{*}
  $$
\end{lemma}
\begin{proof}
    The strong law of large numbers ensures $L_N(\boldsymbol{C}) \rightarrow L(\boldsymbol{C})$ for every $\boldsymbol{C}$, almost surely. By Theorem 5.7 of \cite{vaart_1998}, the desired convergence holds if the following conditions are met:
    \begin{enumerate}
        \item \textit{Uniform Convergence}:
        $$
        \sup_{\boldsymbol{C}\in\boldsymbol{\Theta}_{K,\boldsymbol{m}}}\left|L_N(\boldsymbol{C})  - L(\boldsymbol{C})\right| \overset{P}{\rightarrow} 0 
        $$
        \item \textit{Uniqueness}: For any $\epsilon > 0$
        $$
        \sup_{\boldsymbol{C}:\text{dist}(\boldsymbol{C},\boldsymbol{C}^{*})\ge \epsilon} L(\boldsymbol{C}) > L(\boldsymbol{C}^{*})
        $$
        \item \textit{Near Minimum}:
        $$
            L_{N}(\Breve{\boldsymbol{C}}_N) \le L_{N}(\boldsymbol{C}^{*}) + o_P(1)
        $$
    \end{enumerate}
    \textbf{Condition 1}: This can be verified by using Glivenko-Cantelli theory. Denote $l_{\boldsymbol{C}}(U) = \left\|U - P_{\boldsymbol{C}}(U)\right\|_{\mathcal{H}}^2$, i.e. $L(\boldsymbol{C}) = \mathbb{E}\left[l_{\boldsymbol{C}}(U)\right]$. Denote the set of functions 
    $$
        \Gamma = \{l_{\boldsymbol{C}}: \boldsymbol{C}\in\boldsymbol{\Theta}_{K,\boldsymbol{m}}\}.
    $$
    The uniform convergence requirement is equivalent to $\Gamma$ being Glivenko-Cantelli. 
    We can express each element of the function set as 
    $$
        l_{\boldsymbol{C}}(u) = \|u\|_{\mathcal{H}}^2 - 2\langle u, P_{\boldsymbol{C}}(u)\rangle_{\mathcal{H}} + \langle P_{\boldsymbol{C}}(u), P_{\boldsymbol{C}}(u)\rangle_{\mathcal{H}} 
    $$
    Here, we work with the equivalent formulation of 
    \begin{equation}\label{eqn:polynomial_formulation}
    \begin{aligned}
          l_{\boldsymbol{C}}(u) &= \|u\|_{\mathcal{H}}^2 - 2\langle u, P_{\boldsymbol{C}}(u)\rangle_{\mathcal{H}} + \langle P_{\boldsymbol{C}}(u), P_{\boldsymbol{C}}(u)\rangle_{\mathcal{H}}  \\
          &= \sum_{l=1}^\infty Z_q^2 -2\sum_{k=1}^Kb_k\langle \sum_{q=1}^\infty Z_q\mathcal{A}_q, \bigotimes_{d=1}^D\boldsymbol{c}_{d,k}\rangle_{\tilde{F}} + \sum_{k=1}^K\sum_{p=1}^K b_kb_p\langle \bigotimes_{d=1}^D\boldsymbol{c}_{d,k}, \bigotimes_{d=1}^D\boldsymbol{c}_{d,jp}\rangle_{\tilde{F}}
    \end{aligned}
    \end{equation}
    Recalling the definition of $\langle,\rangle_{\tilde{F}}$ and we have  
    $$
    \begin{aligned}
    \left\langle \sum_{q=1}^\infty Z_q\mathcal{A}_q, \bigotimes_{d=1}^D\boldsymbol{c}_{d,k}\right\rangle_{\tilde{F}} &= \sum_{i_{1}=1}^{m_{1}}\cdots\sum_{i_{D}=1}^{m_{d}}\sum_{j_{1}=1}^{m_{1}}\cdots\sum_{j_{D}=1}^{m_{D}}\left(\sum_{q=1}^\infty Z_q\mathcal{A}_q(i_1,...,i_{D})\right)\prod_{d=1}^D\boldsymbol{J}_{\boldsymbol{\phi}_{d}}(i_d,j_d)\prod_{d=1}^D\boldsymbol{c}_{d,k,j_{d}} \\ 
    \left\langle \bigotimes_{d=1}^D\boldsymbol{c}_{d,k}, \bigotimes_{d=1}^D\boldsymbol{c}_{d,p}\right\rangle_{\tilde{F}} &= \sum_{i_{1}=1}^{m_{1}}\cdots\sum_{i_{D}=1}^{m_{d}}\sum_{j_{1}=1}^{m_{1}}\cdots\sum_{j_{D}=1}^{m_{D}}\prod_{d=1}^D\boldsymbol{J}_{\boldsymbol{\phi}_{d}}(i_d,j_d)\prod_{d=1}^D \boldsymbol{c}_{d,k,i_{d}}\boldsymbol{c}_{d,p,j_{d}}
    \end{aligned}
    $$
    which are polynomials of order $D$ and $2D$ in $\boldsymbol{C}$, respectively. From the definition given in proposition~\ref{prop:kOMPB_projection_coefficients}, we can see that each $b_k$ is also a finite degree polynomial in $\boldsymbol{C}$. As a result, $l_{\boldsymbol{C}}(u)$ is isomorphic to a polynomial with finitely many terms. From the boundedness of the sum of the second moments of the $Z_k$'s along with proposition~\ref{prop:KOMPB_ip_bounded}, it follows that $\mathbb{E}\left[l_{\boldsymbol{C}}\right] < \infty$. Therefore, the $\Gamma$ is VC-class, which follows from Lemma 2.6.15 of \cite{vaart1996}, and hence Glivenko-Cantelli.
    \par 
    \textbf{Condition 2:} This condition indicates $\boldsymbol{C}^{*}$ is a well separated minimum of $L$. Using Lemma~\ref{eqn:coef_tensor_decomposition}, we have that 
    $$
    \begin{aligned}
    \min_{\boldsymbol{\zeta}_{\boldsymbol{m}}\in\mathcal{V}_{K,\boldsymbol{m}}}  \mathbb{E}\left\|U - P_{\boldsymbol{\zeta}_{\boldsymbol{m}}}(U)\right\|_{\mathcal{H}}^2 = \min_{\boldsymbol{C}\in\boldsymbol{\Theta}_{K,{\boldsymbol{m}}}}\min_{\boldsymbol{B}} \quad& 
    \left\|\mathcal{A}^{(K)} - \sum_{k=1}^K \left[\bigotimes_{d=1}^D\boldsymbol{c}_{d,k}\right]\otimes \boldsymbol{B}_{:,k}\right\|_{\tilde{F}, C}^2 + \\ &O(w_{\bm{\tau}_{\boldsymbol{m}}}(\boldsymbol{m})) +  O(h(K))
    \end{aligned}
    $$
    and therefore the minimizer of $L$ is given by the rank $K$ decomposition of the tensor $\mathcal{A}^{(K)}$ under the $\|\cdot\|_{\tilde{F},C}$ norm. Under Assumption~\ref{asm:existance_assumption}, this minimizer is unique in  $\boldsymbol{\Theta}_{K,\boldsymbol{m}}$. Coupled with the compactness of $\boldsymbol{\Theta}_{K,\boldsymbol{m}}$ and continuity of $L$, the desired condition follows. 
    \par 
    \textbf{Condition 3:} This follows trivially, as  
    $$
     L_N(\Breve{\boldsymbol{C}}_N) = \min_{\boldsymbol{C}\in\boldsymbol{\Theta}_{K,\boldsymbol{m}}}N^{-1}\sum_{i}^N\left\|U_i - P_{\boldsymbol{C}}(U_i)\right\|_{\mathcal{H}}^2 \le N^{-1}\sum_{i}^N\left\|U_i - P_{\boldsymbol{C}^{*}}(U_i)\right\|_{\mathcal{H}}^2 = L_N(\boldsymbol{C}^{*}) 
    $$
\end{proof}

\begin{theorem}[Consistency]\label{thm:consistency}
Under Assumptions 1 and \ref{asm:existance_assumption}, we have the following (component-wise) convergence result:
$$
\Breve{\boldsymbol{\zeta}}_{\boldsymbol{m},N}^{*}(\boldsymbol{x}) \overset{P}{\rightarrow} \boldsymbol{\zeta}_{\boldsymbol{m}}^{*}(\boldsymbol{x}) \quad \forall \boldsymbol{x}\in\mathcal{M}.
$$
\end{theorem}
\begin{proof}
    Note that $\boldsymbol{\zeta}_{\boldsymbol{m}}(\boldsymbol{x})$ is a continuous functions of $\boldsymbol{C}\in\boldsymbol{\Theta}_{K,\boldsymbol{m}}$ for all $\boldsymbol{x}\in\mathcal{M}$. The desired result follows directly from the convergence established in Lemma~\ref{lem:param_converge_prob} and the continuous mapping theorem.
\end{proof}

\noindent{\textbf{Convergence Rate}}
\par\bigskip

\begin{lemma}[Lipschitz Map]\label{lem:lipshitz}
    Under Assumptions 1, \ref{asm:spectral_decay_rate} and  \ref{asm:existance_assumption}, there exists functional $F(u)$ such that $\mathbb{E}\left[ F^2\right] <\infty$ and 
    $$
    | l_{\boldsymbol{C}^{(1)}}(u) - l_{\boldsymbol{C}^{(2)}}(u)| \le F(u)\text{dist}(\boldsymbol{C}^{(1)}, \boldsymbol{C}^{(2)})
    $$
    for any $\boldsymbol{C}^{(1)}, \boldsymbol{C}^{(2)}\in\Theta_{K,\boldsymbol{m}}$, where $\text{dist}(\boldsymbol{C}_{1}, \boldsymbol{C}_{2}) = \|\text{vec}\left(\boldsymbol{C}_{1}\right) - \text{vec}\left(\boldsymbol{C}_{2}\right)\|_{2}$.
\end{lemma}
\begin{proof}
Notice that 
  $$
  \begin{aligned}
    &| l_{\boldsymbol{C}^{(1)}}(u) - l_{ \boldsymbol{C}^{(2)}}(u)| = |-2\langle u, P_{\boldsymbol{C}^{(1)}}(u) - P_{ \boldsymbol{C}^{(2)}}(u) \rangle_{\mathcal{H}} + \|P_{\boldsymbol{C}^{(1)}}(u)\|_{\mathcal{H}}^2 -  \|P_{ \boldsymbol{C}^{(2)}}(u)\|_{\mathcal{H}}^2 | \\
    &\le 2 |\langle u, P_{\boldsymbol{C}^{(1)}}(u) - P_{ \boldsymbol{C}^{(2)}}(u) \rangle_{\mathcal{H}}| + |\|(P_{\boldsymbol{C}^{(1)}}(u)\|_{\mathcal{H}} -  \|P_{ \boldsymbol{C}^{(2)}}(u)\|_{\mathcal{H}}) (P_{\boldsymbol{C}^{(1)}}(u)\|_{\mathcal{H}} + \|P_{ \boldsymbol{C}^{(2)}}(u)\|_{\mathcal{H}})| \\
    & \le 2\|u\|_{\mathcal{H}}\|P_{\boldsymbol{C}^{(1)}}(u) - P_{ \boldsymbol{C}^{(2)}}(u)\|_{\mathcal{H}} + \|P_{\boldsymbol{C}^{(1)}}(u) - P_{ \boldsymbol{C}^{(2)}}(u)\|_{\mathcal{H}}(\|P_{\boldsymbol{C}^{(1)}}(u)\|_{\mathcal{H}} + \|P_{ \boldsymbol{C}^{(2)}}(u)\|_{\mathcal{H}})\\
    & \le 4\|u\|_{\mathcal{H}} \|P_{\boldsymbol{C}^{(1)}}(u) - P_{ \boldsymbol{C}^{(2)}}(u)\|_{\mathcal{H}}.
  \end{aligned}
  $$
Additionally, we have that 
\begin{equation}\label{eqn:Lipshitz_projection_operator}
\begin{aligned}
  \|P_{\boldsymbol{C}_{1}}(u) - P_{\boldsymbol{C}_{2}}(u)\|_{\mathcal{H}} &= \|\sum_{l=1}^\infty Z_l\sum_{k=1}^K (b_{l,k}^{(1)}\bigotimes_{d=1}^D \boldsymbol{c}_{d,k}^{(1)} - b_{l,k}^{(2)}\bigotimes_{d=1}^D \boldsymbol{c}_{d,k}^{(2)})\|_{\tilde{F}} \\
  &\le \sum_{l=1}^\infty |Z_l|\sum_{k=1}^K\|(b_{l,k}^{(1)}\bigotimes_{d=1}^D \boldsymbol{c}_{d,k}^{(1)} - b_{l,k}^{(2)}\bigotimes_{d=1}^D \boldsymbol{c}_{d,k}^{(2)})\|_{\tilde{F}} \\
  & \le R\sum_{l=1}^\infty |Z_l|\sum_{k=1}^K\|\bigotimes_{d=1}^D \boldsymbol{c}_{d,k}^{(1)} - \bigotimes_{d=1}^D \boldsymbol{c}_{d,k}^{(2)}\|_{\tilde{F}}.
\end{aligned}
\end{equation}
Clearly, the mapping defined by $\boldsymbol{C} \mapsto \sum_{k=1}^K\|\bigotimes_{d=1}^D \boldsymbol{c}_{d,k}\|_{\tilde{F}}$ has bounded partial derivatives on $\Theta_{K,\boldsymbol{m}}$ and therefore is Lipschitz and hence  
$$
\sum_{k=1}^K\|\bigotimes_{d=1}^D \boldsymbol{c}_{d,k}^{(1)} - \bigotimes_{d=1}^D \boldsymbol{c}_{d,k}^{(2)}\|_{\tilde{F}} \le R\|\text{vec}\left(\boldsymbol{C}_{1}\right) - \text{vec}\left(\boldsymbol{C}_{2}\right)\|_{2}.
$$
Define $F(u) := R\|u\|_{\mathcal{H}} \sum_{l=1}^\infty|Z_l|$ for generic constant $ 0 < R < \infty$. We have that 
$$
\begin{aligned}
  \mathbb{E}\left[ F^2 \right] &= R \mathbb{E}\left[ \left(\sum_{l=1}^\infty Z_l^2\right) \left(\sum_{l=1}^\infty|Z_l|\right)^2\right] \\ 
  & = R\mathbb{E}\left[\sum_{l=1}^\infty\sum_{j=1}^\infty\sum_{q=1}^\infty Z_l^2|Z_j||Z_q|\right] \\
  &= R\sum_{l=1}^\infty\sum_{j=1}^\infty\sum_{q=1}^\infty \mathbb{E}\left[ Z_l^2\right]\mathbb{E}\left[ |Z_j|\right]\mathbb{E}\left[ |Z_q|\right] \mathbb{I}\{l\neq j, l\neq q, q\neq j\} + \\
  & \qquad\qquad \mathbb{E}\left[ Z_l^2\right]\mathbb{E}\left[ Z_q^2\right] \mathbb{I}\{l\neq j, j = q\} + \\
  & \qquad\qquad \mathbb{E}\left[ |Z_l|^3\right]\mathbb{E}\left[ |Z_j|\right]\mathbb{I}\{l\neq j, l = q\} + \\
  & \qquad\qquad  \mathbb{E}\left[ |Z_l|^3\right]\mathbb{E}\left[ |Z_q|\right]\mathbb{I}\{l = j, l \neq q\} + \\
  & \qquad\qquad  \mathbb{E}\left[ Z_l^4\right]\mathbb{I}\{l = j = q\}. \\
\end{aligned}
$$
Therefore, $\sum_{l=1}^\infty\mathbb{E}\left[|Z_l|^r\right] < \infty\text{ for }  r=1,2,3,4 \Longrightarrow \mathbb{E}\left[ F^2 \right] < \infty$. Since $\mathbb{E}\left[|Z_l|\right] \le \sqrt{\rho_l}$ by Jensen's inequality, Assumptions 1 and \ref{asm:spectral_decay_rate} ensure each of these series are convergent and the desired result follows.
\end{proof}
\begin{theorem}[Convergence Rate]\label{thm:convergence_rate}
    Under Assumptions 1, \ref{asm:spectral_decay_rate} and \ref{asm:existance_assumption}, we have 
    $$
    \text{vec}\left(\Breve{\boldsymbol{C}}_{N}\right) -\text{vec}\left(\boldsymbol{C}^{*}\right) = O_{p}\left(N^{-1/2}\right)
    $$
\end{theorem}
\begin{proof}
This follows directly from combining lemmas~\ref{lem:lipshitz} and \ref{lem:param_converge_prob} along with corollary 5.53 of \cite{vaart_1998}.
\end{proof}

\noindent{\textbf{Generalization Error}}
\par\bigskip

\textbf{Proof of Theorem 2.1}
\begin{proof}
By the triangle inequality, 
\begin{equation}\label{eqn:bound_original}
    \mathbb{E}\left\| U - P_{\Breve{\boldsymbol{\zeta}}_{\boldsymbol{m},N}^{*}}(U)\right\|_{\mathcal{H}}^2\le     \mathbb{E}\left\| U - P_{\boldsymbol{\zeta}_{\boldsymbol{m}}^{*}}(U)\right\|_{\mathcal{H}}^2 +     \mathbb{E}\left\| P_{\boldsymbol{\zeta}_{\boldsymbol{m}}^{*}}(U) - P_{\Breve{\boldsymbol{\zeta}}_{\boldsymbol{m},N}^{*}}(U)\right\|_{\mathcal{H}}^2 
\end{equation}
Let $\boldsymbol{\psi}_K = (\psi_1, ..., \psi_K)^\prime$ and denote $P_{\boldsymbol{\psi}_K}$ the projection operator onto $\text{span}(\boldsymbol{\psi}_K)$. For the first term in the bound \eqref{eqn:bound_original}, we have that 
\begin{equation}\label{eqn:gen_error_master_eqn}
  \begin{aligned}
  \mathbb{E}\left\| U - P_{\boldsymbol{\zeta}_{\boldsymbol{m}}^{*}}(U)\right\|_{\mathcal{H}}^2 &= \mathbb{E}\left\| U - P_{\boldsymbol{\zeta}_{\boldsymbol{m}}^{*}}(U) + P_{\boldsymbol{\psi}_{K}}(U) - P_{\boldsymbol{\psi}_{K}}(U) \right\|_{\mathcal{H}}^2 \\
  & \le   \mathbb{E}\left\| U - P_{\boldsymbol{\psi}_{K}}(U) \right\|_{\mathcal{H}}^2+ \mathbb{E}\left\| P_{\boldsymbol{\psi}_{K}}(U) - P_{\boldsymbol{\zeta}_{\boldsymbol{m}}^{*}}(U) \right\|_{\mathcal{H}}^2
\end{aligned}  
\end{equation}
Clearly, $\mathbb{E}\left\| U - P_{\boldsymbol{\psi}_{K}}(U) \right\|_{\mathcal{H}}^2 = \sum_{k=K+1}^\infty\rho_k$. Considering now the second term in the sum on line two of \eqref{eqn:gen_error_master_eqn}, observe that 
$$
\begin{aligned}
  \mathbb{E}\left\| P_{\boldsymbol{\psi}_{K}}(U) - P_{\boldsymbol{\zeta}_{\boldsymbol{m}}^{*}}(U) \right\|_{\mathcal{H}}^2 &= \mathbb{E}\left\|P_{\mathcal{H}_{m}}\left( P_{\boldsymbol{\psi}_{K}}(U) - P_{\boldsymbol{\zeta}_{\boldsymbol{m}}^{*}}(U)\right) + P_{\mathcal{H}_{m}^{\perp}}\left(P_{\boldsymbol{\psi}_{K}}(U) - P_{\boldsymbol{\zeta}_{\boldsymbol{m}}^{*}}(U)\right) \right\|_{\mathcal{H}}^2 \\
  &= \mathbb{E}\left\|P_{\mathcal{H}_{m}}\left(P_{\boldsymbol{\psi}_{K}}(U)\right) - P_{\boldsymbol{\zeta}_{\boldsymbol{m}}^{*}}(U)\right\|_{\mathcal{H}_{m}}^2 + \mathbb{E}\left\|P_{\mathcal{H}_{m}^{\perp}}\left(P_{\boldsymbol{\psi}_{K}}(U)\right)\right\|_{\mathcal{H}_{m}^{\perp}}^2\\
  &= \text{Term}_1 + \text{Term}_2
\end{aligned}
$$
Clearly, $\text{Term}_2 = O(w_{\bm{\tau}_{\boldsymbol{m}}}(\boldsymbol{m}))$. In regard to $\text{Term}_1$, using the same logic as in the proof of Lemma~\ref{lem:generalization_error}, we have that 
$$
\begin{aligned}
    \text{Term}_1 &= \min_{\boldsymbol{B}}\sum_{k=1}^K\rho_l\left\| \mathcal{A}_l - \sum_{k=1}^K\boldsymbol{B}_{l,k}\bigotimes_{d=1}^D\boldsymbol{c}_{d,k}^{*}\right\|_{\tilde{F}}^2 \\
  &= \min_{\boldsymbol{C}\in\boldsymbol{\Theta}_{K,{\boldsymbol{m}}}}\min_{\boldsymbol{B}}\sum_{l=1}^K \rho_l\left\| \mathcal{A}_l - \sum_{k=1}^K\boldsymbol{B}_{l,k}\bigotimes_{d=1}^D\boldsymbol{c}_{d,k}\right\|_{\tilde{F}}^2 \\
  &=  \min_{\boldsymbol{C}\in\boldsymbol{\Theta}_{K,{\boldsymbol{m}}}}\min_{\boldsymbol{B}} \left\|\mathcal{A}^{(K)} - \sum_{k=1}^K\boldsymbol{B}_{l,k}\bigotimes_{d=1}^D\boldsymbol{c}_{d,k}\right\|_{\tilde{F},C}^2 \\
  &= \left\|\mathcal{A}^{(K)} - \widehat{\mathcal{A}}^{(K)}_{K}\right\|_{\tilde{F},C}^2
\end{aligned}
$$
where the last equality follows from the definition of the canonical polyadic decomposition, and the $O(h(K))$ term from Lemma~\ref{lem:generalization_error} is avoided due to the finite truncation of $\boldsymbol{\psi}_K$. For the second term in the bound \eqref{eqn:bound_original}, using the Lipshitz property of the projection operators along with the rate in established in Theorem~\ref{thm:convergence_rate}, it is easy to see that 
\begin{equation}\label{eqn:operator_convergence_rate}
\begin{aligned}
  \mathbb{E}\left\| P_{\boldsymbol{\zeta}_{\boldsymbol{m}}^{*}}(U) - P_{\Breve{\boldsymbol{\zeta}}_{\boldsymbol{m},N}^{*}}(U)\right\|_{\mathcal{H}} = O_{p}(N^{-1/2})
\end{aligned}
\end{equation}
The desired result follows from plugging the derived forms of $\mathbb{E}\left\| U - P_{\boldsymbol{\psi}_{K}}(U) \right\|_{\mathcal{H}}^2$, $\text{Term}_1$, $\text{Term}_2$ and \eqref{eqn:operator_convergence_rate} into Equation~\ref{eqn:bound_original}.
\end{proof}
\newpage
\noindent{\textbf{Proof of Theorem 3.1}}
\par\bigskip
\begin{proof}
  $$
  \begin{aligned}
    &\sum_{i=1}^{N} \left\| \mathcal{Y}_i - \sum_{k=1}^{K} \boldsymbol{B}_{ik} \bigotimes_{d=1}^{D} \boldsymbol{\Phi}_{d} \boldsymbol{c}_{d,k}  \right\|_{F}^{2} =  \left\| \mathcal{Y} - \sum_{k=1}^{K} \bigotimes_{d=1}^{D} \boldsymbol{\Phi}_{d} \boldsymbol{c}_{d,k} \otimes \boldsymbol{b}_{k} \right\|_{F}^{2} \\
    &= \left\| \mathcal{Y} - \Big[\sum_{k=1}^{K} \bigotimes_{d=1}^{D} \boldsymbol{D}_d\boldsymbol{V}_d^\prime \boldsymbol{c}_{d,k} \otimes \boldsymbol{b}_{k} \Big] \times_{1} \boldsymbol{U}_{1} \times_{2} \boldsymbol{U}_{2} \dots \times_{D} \boldsymbol{U}_{D} \right\|_{F}^2 \\
    &= \left\| \widehat{\mathcal{G}}- \Big[\sum_{k=1}^{K} \bigotimes_{d=1}^{D} \boldsymbol{D}_d\boldsymbol{V}_d^\prime \boldsymbol{c}_{d,k} \otimes \boldsymbol{b}_{k} \Big] \right\|_{F}^2.
  \end{aligned}
  $$
  Here the first equality is from properties of the Frobenius norm, the second comes from properties of $d$-mode multiplication, and the third from invariance of the Frobenius norm to orthogonal transformation. Therefore, solving Equation (9) is equivalent to solving
  \begin{equation}\label{eqn:Common_Basis_Optimizer_Transform}
    \min_{\boldsymbol{B}, \boldsymbol{C}}\left\|  \widehat{\mathcal{G}}- \sum_{k=1}^{K} \bigotimes_{d=1}^{D} \boldsymbol{D}_d\boldsymbol{V}_d^\prime \boldsymbol{C}_d \otimes \boldsymbol{b}_{k} \right\|_{F}^2.
  \end{equation}
  Using the mapping $\boldsymbol{\tilde{C}}_d = \boldsymbol{D}_d \boldsymbol{V}^\prime_d \boldsymbol{C}_d$, Equation~\eqref{eqn:Common_Basis_Optimizer_Transform} can be reparameterized as
  \begin{equation}\label{eqn:Common_Basis_Optimizer_Transform_Tilde}
    \min_{\boldsymbol{B}, \boldsymbol{\tilde{C}}}\left\| \widehat{\mathcal{G}}- \sum_{k=1}^{K} \bigotimes_{d=1}^{D} \boldsymbol{\tilde{C}}_d \otimes \boldsymbol{b}_{k} \right\|_{F}^2,
  \end{equation}
  which is solved by the rank-$K$ CPD of $\widehat{\mathcal{G}}$. Comparing Equations~\eqref{eqn:Common_Basis_Optimizer_Transform} and \eqref{eqn:Common_Basis_Optimizer_Transform_Tilde}, we see that $\widehat{\boldsymbol{B}}=\boldsymbol{B}$ and
  $\boldsymbol{D}_d\boldsymbol{V}_d\boldsymbol{\widehat{C}}_d = \boldsymbol{\tilde{C}}_d$, or equivalently, $\widehat{\boldsymbol{C}}_d=\boldsymbol{V}_{d}\boldsymbol{D}_{d}^{-1}\boldsymbol{\tilde{C}}_d$.
\end{proof}
\noindent{\textbf{Proof of Proposition 3.2}}
\par\bigskip
\begin{proof}
  Let $\boldsymbol{T}_d:=\boldsymbol{D}_d^{-1}\boldsymbol{V}_d^{\prime}\boldsymbol{R}_d\boldsymbol{V}_d\boldsymbol{D}_d^{-1}\in S^{M}_{+}$, with $\boldsymbol{R}_d(i,j) = \int_{\mathcal{M}_{d}}L_d(\phi_{d,i})L_d(\phi_{d,j})$, then 
  \begin{equation*}
    \begin{aligned}
      &\sum_{k=1}^K\int_{\mathcal{M}_{d}}\sum_{d=1}^D\lambda_d L^2_d(\xi_{k,d}) = \sum_{d=1}^D\lambda_d\sum_{k=1}^K\boldsymbol{c}_{d,k}^\prime \boldsymbol{R}_d\boldsymbol{c}_{d,k} \\
      &= \sum_{d=1}^D\lambda_d\sum_{k=1}^K\boldsymbol{\tilde{c}}_{d,k}^\prime \boldsymbol{D}_d^{-1}\boldsymbol{V}_d^{\prime}\boldsymbol{R}_d\boldsymbol{V}_d\boldsymbol{D}_d^{-1}\boldsymbol{\tilde{c}}_{d,k} = \sum_{d=1}^D\lambda_d\mathrm{tr}(\boldsymbol{\tilde{C}}^\prime_d\boldsymbol{T}_{d}\boldsymbol{\tilde{C}}_d), \\
    \end{aligned}
  \end{equation*}
\end{proof}
\section{\texttt{MARGARITA} Algorithm}\label{sec:algorithm}
\revised{Algorithm~\ref{alg:ss_bcd} provides pseudocode for \texttt{MARGARITA}. A couple comments are in order.}
\begin{itemize}
    \item \revised{For the ADMM subproblem, we adopt the stopping criteria proposed in \cite{boyd2011} based on the primal and dual residuals at the $r^{th}$ iteration, which have the form
    \begin{equation}\label{eqn:primal_dual_residual}
      r_{primal}^{(r)} =  \|\boldsymbol{B}^{(r)} - \boldsymbol{Z}^{(r)\prime}\|_{F}, \qquad r_{dual}^{(r)} = \|\gamma\left(\boldsymbol{Z}^{(r)} - \boldsymbol{Z}^{(r-1)}\right)\|_F.
    \end{equation}}
    \item \revised{As discussed, we can only guarantee convergence to a local minimum, and thus in practice it may be desirable to run Algorithm~\ref{alg:ss_bcd} for multiple random initializations and keep the best solution, e.g. evaluated using the proportion of variance explained criteria discussed in Section 3.5. In simulation results not reported, we did not find much difference in performance for a single random initialization vs. multiple random initializations, though this may become more important as the dimension of the domain $D$ increases.}
    \item \revised{In theory, the same ADMM scheme can be used to solve the special case of $l()=\|\cdot\|_{F}^2$, but this is not necessary in practice as an analytic solution exits. That said, for very large $N$, it may be the desirable to avoid an analytic solution as well due to the requirement of a large matrix inverse, i.e. see discussion in Section~\ref{apx:penalty_strengths}. In such cases, a more scalable solver, e.g. stochastic gradient descent, may be plugged into solve the $\boldsymbol{B}$ sub-problem.}
\end{itemize}
\begin{algorithm}[t]
  \caption{\texttt{MARGARITA}: \underline{MARG}inal-product b\underline{A}sis \underline{R}epresentation w\underline{I}th \underline{T}ensor \underline{A}nalysis}
  \label{alg:ss_bcd}
   \scriptsize 
  \begin{algorithmic}[1]
    \State \textbf{Input} $\mathcal{Y}$, $\mathcal{X}$, $\{\boldsymbol{\phi}_{m_{1},1},...,\boldsymbol{\phi}_{m_{D},D}\}$, $\{L_1, ...,\
    L_D\}$, $\{\lambda_{1}, ..., \lambda_{D+1}\}$, $K$
    \State \textbf{Output} $\boldsymbol{C}_1, ..., \boldsymbol{C}_D$, $\boldsymbol{B}$
    \For{d = 1,...,D}
    \State Compute $\boldsymbol{T}_d$ using Proposition 3.2 and  $\boldsymbol{\Phi}_d = \boldsymbol{U}_d\boldsymbol{D}_d\boldsymbol{V}_d^{\prime}$ using $\boldsymbol{\phi}_d$ and $\mathcal{X}$ 
    \State \revised{Randomly initialize $\boldsymbol{\tilde{C}}_d$}
    \EndFor
    \State \revised{Compute $\widehat{\mathcal{G}} = \mathcal{Y}\times_1\boldsymbol{U}_{1}^\prime\times_2\cdots\times_{D}\boldsymbol{U}_{D}^{\prime}$, randomly initialize $\boldsymbol{B}$ and set $\boldsymbol{A}^{*}$ as zero matrix} 
    \While {\text{change in $\boldsymbol{\tilde{C}}_1, ..., \boldsymbol{\tilde{C}}_D$, $\boldsymbol{B}$ is non-negligible}}
    \For{d = 1,...,D}
    \State \text{Update $\boldsymbol{\tilde{C}}_d$ according to (12), by way of (14)} and re-scale to unit norm
    \EndFor
    \While{$r_{primal}$ $>$ tol$_{primal}$ or $r_{dual}$ $>$ tol$_{dual}$}
    \State \text{Update $\boldsymbol{B}$ according to (16)}
    \State \text{Update $\boldsymbol{Z}$ according to (19)}
    \State \text{Update $\boldsymbol{A}^{*}$ according to (18)}
    \State \text{Update $r_{primal}$, $r_{dual}$ according to~\eqref{eqn:primal_dual_residual}}
    \EndWhile
    \EndWhile
    \For{d = 1,...,D}
    \State Get coefficient matrices using transformation $\boldsymbol{C}_d=\boldsymbol{V}_d\boldsymbol{D}_d^{-1}\boldsymbol{\tilde{C}}_d$
    \EndFor
  \end{algorithmic}
\end{algorithm}

\section{Hyper-Parameter Selection}\label{sec:hyper_param_selection}

\subsection{Rank Selection}
\revised{Data-driven methods for both marginal and global rank selection are important in practice, as they directly determine the approximation power of the resulting marginal product basis system, see Theorem 2.1. In the following,  we provide an elaboration of and justification for the criteria used for the data-driven rank selection and evaluate their performance on simulated data from Section 5.}
\subsubsection{Marginal Rank}\label{ssec:marginal_rank_selection}

\revised{Recall that the proposed marginal rank selection criteria is defined as: $\text{PVM}(\boldsymbol{m}):=\|\mathcal{Y}\bigtimes_{d=1}^D\boldsymbol{U}_d^\prime\|_{F}^2/\|\mathcal{Y}\|_{F}^2$. Denote tensor $\mathcal{U}_i\in\mathbb{R}^{n_{1}\times\cdots\times n_{D}}$ with element-wise definition $\mathcal{U}_i(i_1,...,i_{D}) = U_i(x_{1,i_{1}},...,x_{D,i_{D}})$. Then the observation model in Equation (7) can be written as 
$$
\mathcal{Y}_i = \mathcal{U}_i + \mathcal{E}_i,
$$
where $\mathcal{E}_i$ is a tensor of isotropic normal errors with variance $\sigma^2$. The regression of $\mathcal{U}_i$ onto the tensor product basis $\boldsymbol{\tau}_{\boldsymbol{m}}$
can be defined as
\begin{equation}\label{eqn:tensor_product_optimization}
\begin{aligned}
    \mathcal{A}_i = \underset{\mathcal{A}\in\mathbb{R}^{m_1\times\cdots\times m_{d}}}{\text{argmin}} \left\|\mathcal{U}_i - \mathcal{A}\bigtimes_{d=1}^D\boldsymbol{\Phi}_d\right\|_{F}^2,
\end{aligned}
\end{equation}
where $\mathcal{A}_i$ is the coefficient tensor of $\boldsymbol{\tau}_{\boldsymbol{m}}$. Recalling the notation of the SVD of the basis evaluation matrix $\boldsymbol{\Phi}_{d} := \boldsymbol{U}_d\boldsymbol{D}_d\boldsymbol{V}_{d}^{\prime}$, using properties of invariance of the norm, the least squares objective can be written as 
$$
\begin{aligned}
   &\left\|\mathcal{U}_i - \mathcal{A}\bigtimes_{d=1}^D\boldsymbol{\Phi}_d\right\|_{F}^2  =\left\|\mathcal{U}_i\bigtimes_{d=1}^{D}\boldsymbol{U}_{d}^{\prime} - \left(\mathcal{A}\bigtimes_{d=1}^D\boldsymbol{\Phi}_d\right)\bigtimes_{d=1}^{D}\boldsymbol{U}_{d}^{\prime}\right\|_{F}^2
    =\left\|\mathcal{U}_i\bigtimes_{d=1}^{D}\boldsymbol{U}_{d}^{\prime} - \left(\mathcal{A}\bigtimes_{d=1}^D\boldsymbol{D}_d\boldsymbol{V}_{d}^\prime\right)\right\|_{F}^2\\
    &= \left\|\boldsymbol{u}^{*}_{i,(d)} - (\boldsymbol{D}_{D}\boldsymbol{V}_{D}^{\prime}\otimes\cdots\otimes\boldsymbol{D}_{d+1}\boldsymbol{V}_{d+1}^{\prime}\otimes\boldsymbol{D}_{d-1}\boldsymbol{V}_{d-1}^{\prime}\otimes\cdots\otimes\boldsymbol{D}_1\boldsymbol{V}_{1}^{\prime})^{\prime}\otimes\boldsymbol{D}_{d}\boldsymbol{V}_{d}^{\prime}\text{vec}(\mathcal{A}_{(d)})\right\|_{F}^2
\end{aligned}
$$
where $\boldsymbol{u}^{*}_{i,(d)}$ is shorthand for the vectorization of the $d$-mode unfolding of tensor, i.e. 
$\boldsymbol{u}^{*}_{i,(d)} := \text{vec}\left([\mathcal{U}_i\bigtimes_{j=1}^{D}\boldsymbol{U}_{j}^{\prime}]_{(d)}\right)$. By the properties of the Kronecker product, the design matrix 
$(\boldsymbol{D}_{D}\boldsymbol{V}_{D}^{\prime}\otimes\cdots\otimes\boldsymbol{D}_{d+1}\boldsymbol{V}_{d+1}^{\prime}\otimes\boldsymbol{D}_{d-1}\boldsymbol{V}_{d-1}^{\prime}\otimes\cdots\otimes\boldsymbol{D}_1\boldsymbol{V}_{1}^{\prime})^{\prime}\otimes\boldsymbol{D}_{d}\boldsymbol{V}_{d}^{\prime}\in\mathbb{R}^{\prod_{d=1}^Dm_{d}\times \prod_{d=1}^Dm_{d}}$ is invertible, and so we have the exact solution:
$$
\text{vec}(\mathcal{A}_{i,(d)}) = ((\boldsymbol{D}_{D}\boldsymbol{V}_{D}^{\prime}\otimes\cdots\otimes\boldsymbol{D}_{d+1}\boldsymbol{V}_{d+1}^{\prime}\otimes\boldsymbol{D}_{d-1}\boldsymbol{V}_{d-1}^{\prime}\otimes\cdots\otimes\boldsymbol{D}_1\boldsymbol{V}_{1}^{\prime})^{\prime}\otimes\boldsymbol{D}_{d}\boldsymbol{V}_{d}^{\prime})^{-1}\boldsymbol{u}^{*}_{i,(d)},
$$
from which it follows that 
$$
 \left\|\left(\mathcal{U}_i- \mathcal{A}_i\bigtimes_{d=1}^D\boldsymbol{\Phi}_d\right)\bigtimes_{d=1}^D\boldsymbol{U}_d^{\prime} \right\|_{F}^2 = 0.
$$
Defining the tensor 
$\mathcal{U}_i^{\perp} := \mathcal{U}_i - \mathcal{A}_i\bigtimes_{d=1}^D\boldsymbol{\Phi}_d$,
the model for the $i$'th subject can be equivalently written as
$$
\mathcal{Y}_i = \mathcal{A}_i\bigtimes_{d=1}^D\boldsymbol{\Phi}_d + \mathcal{U}_i^{\perp} + \mathcal{E}_i. 
$$   
Now, the numerator of $\text{PVM}(\boldsymbol{m})$ can be written as  
$$
\begin{aligned}
  \left\|\mathcal{Y}\bigtimes_{d=1}^D\boldsymbol{U}_d^\prime\right\|_{F}^2 &= \sum_{i=1}^N\left\|\mathcal{Y}_i\left(\bigtimes_{d=1}^D\boldsymbol{U}_d\boldsymbol{U}_d^\prime\right)\right\|_{F}^2 \\
  & =\sum_{i=1}^N\left\|\left(\mathcal{A}_i\bigtimes_{d=1}^D\boldsymbol{\Phi}_d + \mathcal{U}_i^{\perp} + \mathcal{E}_i\right)\bigtimes_{d=1}^D\boldsymbol{U}_d\boldsymbol{U}_d^\prime\right\|_{F}^2\\
    & =\sum_{i=1}^N\left\|\mathcal{A}_i\bigtimes_{d=1}^D\boldsymbol{\Phi}_d + \mathcal{E}_i\bigtimes_{d=1}^D\boldsymbol{U}_d\boldsymbol{U}_d^\prime\right\|_{F}^2.\\
\end{aligned}
$$
Putting this altogether, since $\mathcal{U}_i \perp \mathcal{E}_i$ and independent for all $i$, asymptotically, we have that 
\begin{equation}\label{eqn:marg_rank_selection_criteria}
\begin{aligned}
    \text{PVM}(\boldsymbol{m}) &= \frac{N^{-1}\|\mathcal{Y}\bigtimes_{d=1}^D\boldsymbol{U}_d^\prime\|_{F}^2}{N^{-1}\|\mathcal{Y}\|_{F}^2} \\ 
    &= \frac{N^{-1}\sum_{i=1}^N\left\|\mathcal{A}_i\bigtimes_{d=1}^D\boldsymbol{\Phi}_d + \mathcal{E}_i\bigtimes_{d=1}^D\boldsymbol{U}_d\boldsymbol{U}_d^\prime\right\|_{F}^2}{N^{-1}\sum_{i=1}^N\left\|\mathcal{U}_i + \mathcal{E}_i\right\|_{F}^2} \\
    & \asymp  \frac{N^{-1}\sum_{i=1}^N\left\|\mathcal{A}_i\bigtimes_{d=1}^D\boldsymbol{\Phi}_d\right\|_{F}^2 + N^{-1}\sum_{i=1}^N\left\|\mathcal{E}_i\bigtimes_{d=1}^D\boldsymbol{U}_d\boldsymbol{U}_d^\prime\right\|_{F}^2}{N^{-1}\sum_{i=1}^N\left\|\mathcal{U}_i\right\|_{F}^2+N^{-1}\sum_{i=1}^N\left\|\mathcal{E}_i\right\|_{F}^2} \\
    & \asymp  \frac{N^{-1}\sum_{i=1}^N\left\|\mathcal{A}_i\bigtimes_{d=1}^D\boldsymbol{\Phi}_d\right\|_{F}^2 + \sigma^2\prod_{d=1}^{D}n_{d}}{N^{-1}\sum_{i=1}^N\left\|\mathcal{U}_i\right\|_{F}^2+\sigma^2\prod_{d=1}^{D}n_{d}} \\
\end{aligned}  
\end{equation}
where the final line is due to 
$$
N^{-1}\sum_{i=1}^N\left\|\mathcal{E}_i\right\|_{F}^2 = N^{-1}\sum_{i=1}^N\left\|\mathcal{E}_i\bigtimes_{d=1}^D\boldsymbol{U}_d\boldsymbol{U}_d^\prime\right\|_{F}^2 \asymp \sigma^2\prod_{d=1}^{D}n_{d},
$$
which results from the permutation invaraince of the trace and the strong law of large numbers. }
\par 
\revised{Now, assuming a sequence of equipartitioned grids $\mathcal{X}$,  
$[\prod_{d=1}^{D}n_{d}]^{-1}\left\|\mathcal{U}_i\right\|_{F}^2$ can considered proportional to  a Riemann sum approximation to the integral $\left\|\mathcal{U}_i\right\|_{\mathcal{H}}^2$. Hence, under a fine grid regime, using the strong law of large numbers, recalling the notation from Section 2, we have the approximations
$$
\begin{aligned}
N^{-1}\sum_{i=1}^N[\prod_{d=1}^{D}n_{d}]^{-1}\left\|\mathcal{A}_i\bigtimes_{d=1}^D\boldsymbol{\Phi}_d\right\|_{F}^2 \approx \mathbb{E}[\left\|P_{\mathcal{H}_{\boldsymbol{m}}}(U)\right\|_{\mathcal{H}}^2] \\     
N^{-1}\sum_{i=1}^N[\prod_{d=1}^{D}n_{d}]^{-1}\left\|\mathcal{U}_i\right\|_{F}^2 \approx \mathbb{E}[\left\|U\right\|_{\mathcal{H}}^2].
\end{aligned}
$$
For ease of presentation, from here on we take $m_1 = ... = m_{D}=m$ and $n_1 = ... n_{D}=n$. Then in the fine grid limit, we have the approximation
$$
\text{PVM}(m) \asymp \frac{\mathbb{E}[\left\|P_{\mathcal{H}_{\boldsymbol{m}}}(U)\right\|_{\mathcal{H}}^2] + \sigma^2}{\mathbb{E}[\left\|U\right\|_{\mathcal{H}}^2]+\sigma^2},
$$
which is a monotonically increasing function of $m$, with $\text{PVM}(m) \rightarrow 1$, and measures the degree of irreducible bias incurred by the finite trunctional of the marginal ranks. By analogy, on the discretely observed grid, $\text{PVM}(m)$ is a monotonic function of $m\le n$, with $\text{PVM}(n)=1$ which is an approximation to the proportion of irreducible bias from the finite truncation of ranks under the discrete projection \eqref{eqn:tensor_product_optimization}.}
\par 
\revised{In practice, in the absence of strong a-priori knowledge, we suggest dividing each marginal domain into $m=1,...,M$ equispaced candidate ranks for each $d$ and compute the $\text{PVM}(m)$ for each of the $M$ candidate marginal ranks $\boldsymbol{m}:=(\text{floor}(\frac{m}{M}n_1), ...,\text{floor}(\frac{m}{M}n_{D}))$.
Figure~\ref{fig:marg_rank_selection} shows $\text{PVM}(m)$ as a function of $m$ under several of the simulation set-ups considered in Section 5. Cubic b-splines were used as the marginal basis of the fits for all the experiments. In order to study the effects of noise in both set-ups, we added idd Gaussian noise to the simulated fields in Section 5.2 with $\sigma^2=0.1$ and $\sigma^2=1$ for both relatively high and low SNR:=$\frac{\sum_{k}\rho_k}{\sigma^2}$. We make the following observations:}
\begin{enumerate}
    \item \revised{We see that in the high SNR regimes, a proportion of variance explained or elbow criteria will both work. Alternatively, in the low signal to noise ratio regimes, the proportion of variance explained will select a larger model than necessary. Though ``flatter'' than that of the high SNR, the elbow of $\text{PVM}(m)$ can still be consistently identified (even visually) in the low SNR simulations and thus we suggest using an elbow-like criteria in practice, especially when degrees of freedom need to be conserved.}
    \item \revised{For the simulation setup from Section 5.1 (top two plots), the true marginal basis are the Fourier functions while the marginal basis in fitting are the cubic b-splines. For both SNR's, we see that a marginal b-spline basis with rank $m=15$ is consistently indentified with an elbow in the PVM for both noise levels. This is echoed in Table \ref{tab:mpf_vs_tpb}, where we see that the performance jump from $m_d=8$ to $m_d=15$ is substantial, while the over-parameterized regime $m_d=25$ performs quite similarly to $m_d=15$.}
    \item \revised{For the simulation setup from Section 5.2 (bottom two plots), the true marginal ranks $m_1\neq m_2$, hence, in theory, one direction will be ``smoother'' than the other. We again see a clear elbow at max($m_1, m_2$)=10. We also notice an inflection point at  max($m_1, m_2$)=8, which may be a useful way of identifying when differing ranks are required in different directions, though the issue of which marginal domain requires which rank would need subsequent exploration.}
\end{enumerate}
\revised{We conclude this section by reiterating two significant advantages of the PVM criteria: i) it can be precomputed using only the SVD of a set of relatively small matrices ii) it is independent of the global rank $K$.}
\begin{figure}
  \centering
  \begin{tabular}{@{}c@{}}
    \includegraphics[scale=0.4]{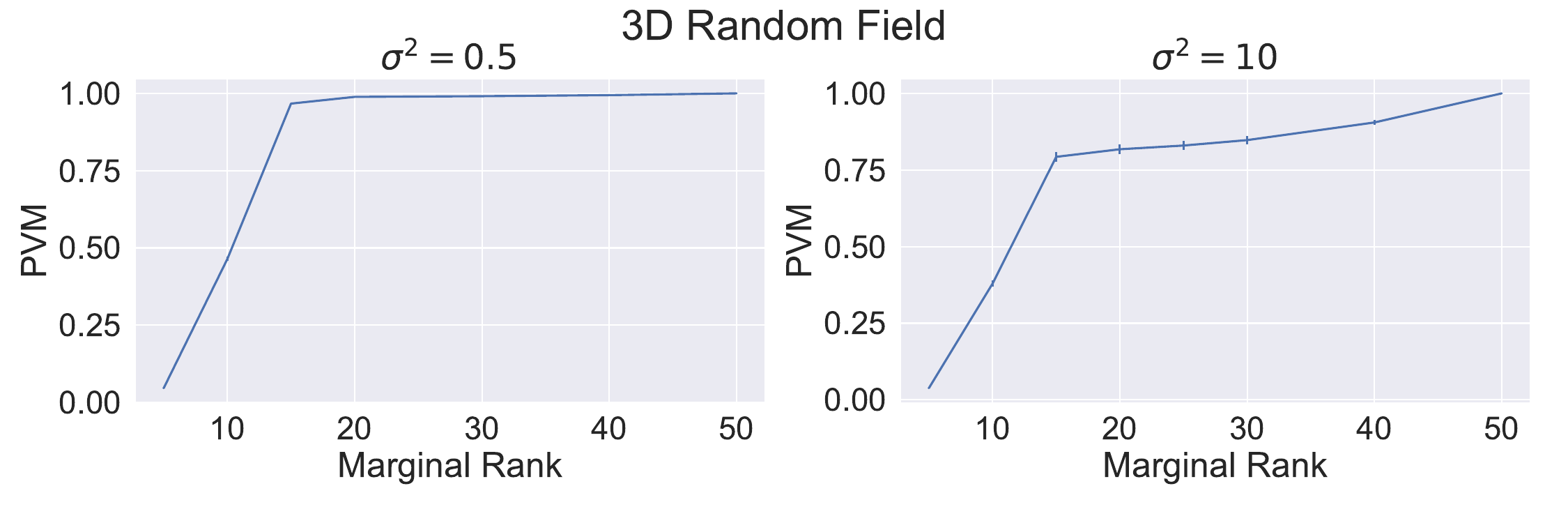} \\
  \end{tabular}
  \vspace{\floatsep}
  \begin{tabular}{@{}c@{}}
    \includegraphics[scale=0.4]{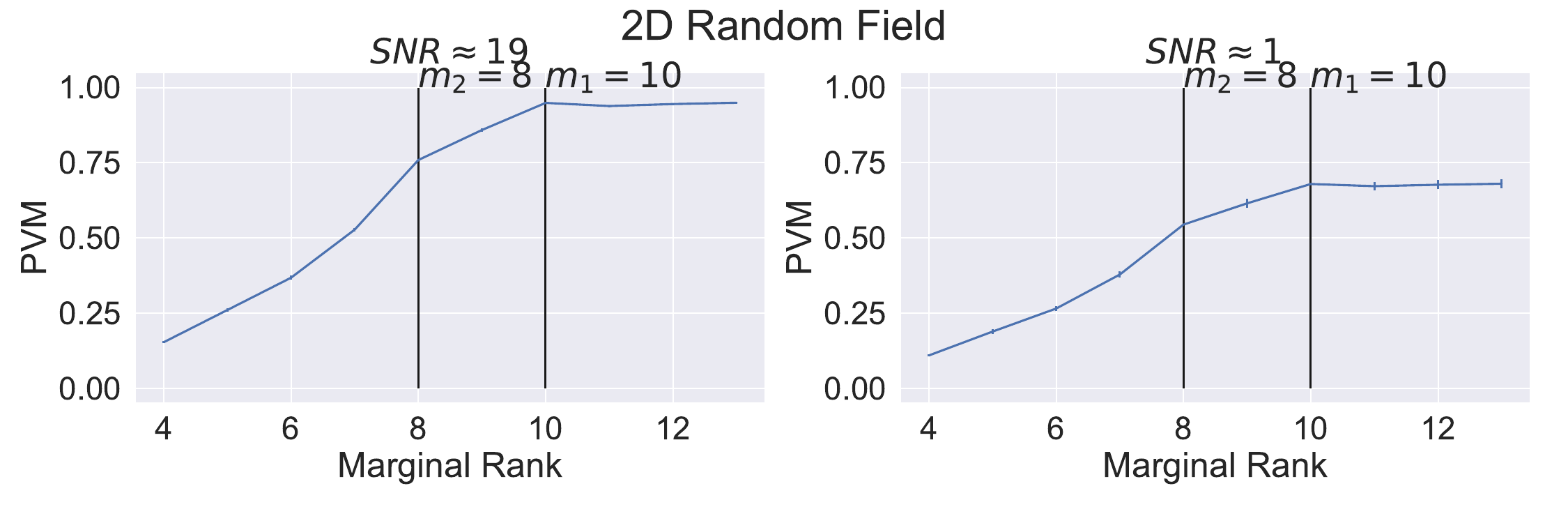} \\
  \end{tabular}
  \caption{\revised{Monte-Carlo average PVM as a function of increasing $m$. The top two plots show the results from the 3D simulation set-up detailed in Section 5.1 with $n_d=50$, $N=50$ and $K_t=20$. The bottom two plots show the results from the 2D simulation set-up from Section 5.2 with $N=100$.} }\label{fig:marg_rank_selection}
\end{figure}
\subsubsection{Global Rank}
\revised{As discussed in the main text, a standard criteria for global rank selection is the proportion of variance explained, which in this context is given by 
$$
PVG(K) = \|\mathcal{G} -\sum_{k=1}^K\boldsymbol{b}_k\otimes(\bigotimes_{d=1}^D\boldsymbol{\tilde{c}_{d,k}})\|_{F}^2/\|\mathcal{G}\|_{F}^2.
$$
Computing the PVG requires multiple runs of \texttt{MARGARITA}, too many of which may want to be avoided in the super high-dimensional case.  To avoid too many runs, one can start with some user-defined maximum rank $K_{max}$ and then step backward in increments until the smallest model which still meets the desired PVG(K) is obtained. In simulations, we find that the PVG criteria is robust to changes in the SNR, see Figure~\ref{fig:global_rank_selection}.  Computational resources permitting, we suggest the user to be generous with the PVG criteria, e.g. setting it very close to 1, to ``soak up'' most of the variance and then regularize via the cross-validation procedure discussed in Section~\ref{apx:penalty_strengths}. This is especially true if subsequent FPCA is desired and was the tactic for the two-stage estimation of the eigenfunctions in Section 5.2, which produces strong estimates for at least the first three eigenfunctions, see Figure~\ref{fig:eigenfunction_estimates}. If desired, model selection in this case can be subsequently performed by selecting the rank on the second stage eigen-basis. 
}
\par 
\revised{An interesting potential alternative path for global model selection is to encode the rank selection via a penalty operator and perform rank selection via selecting the penalty parameter $\lambda_{D+1}$. Specifically, we speculate that this may be accomplished by taking $l(\boldsymbol{B}) = \sum_{k=1}^K\|\boldsymbol{B}(:,k)\|_{2}$, the group lasso.  Similar $l_1$-penalty based optimization strategies have been proposed for automatic approximate rank determination in the tensor decomposition literature \citep{wang2015_tdc}. \texttt{MARGARITA} can seamlessly integrate this penalty by simply specifying the corresponding proximal operator of the group lasso. That said, our approach to penalty parameter selection (discussed in Section~\ref{apx:penalty_strengths}) would need to be augmented for this purpose. As the rank of the model is not constant, for this case, it is desirable to adequately penalize a model complexity term in addition to a model fit term. As our method does not invoke a likelihood framework, there is no straightforward application of a standard information criteria for model selection, e.g. AIC or BIC. Therefore, developing a model selection criteria for our case is an important avenue for future research, but is beyond the scope of the current work. }
\begin{figure}
  \centering
  \begin{tabular}{@{}c@{}}
    \includegraphics[scale=0.4]{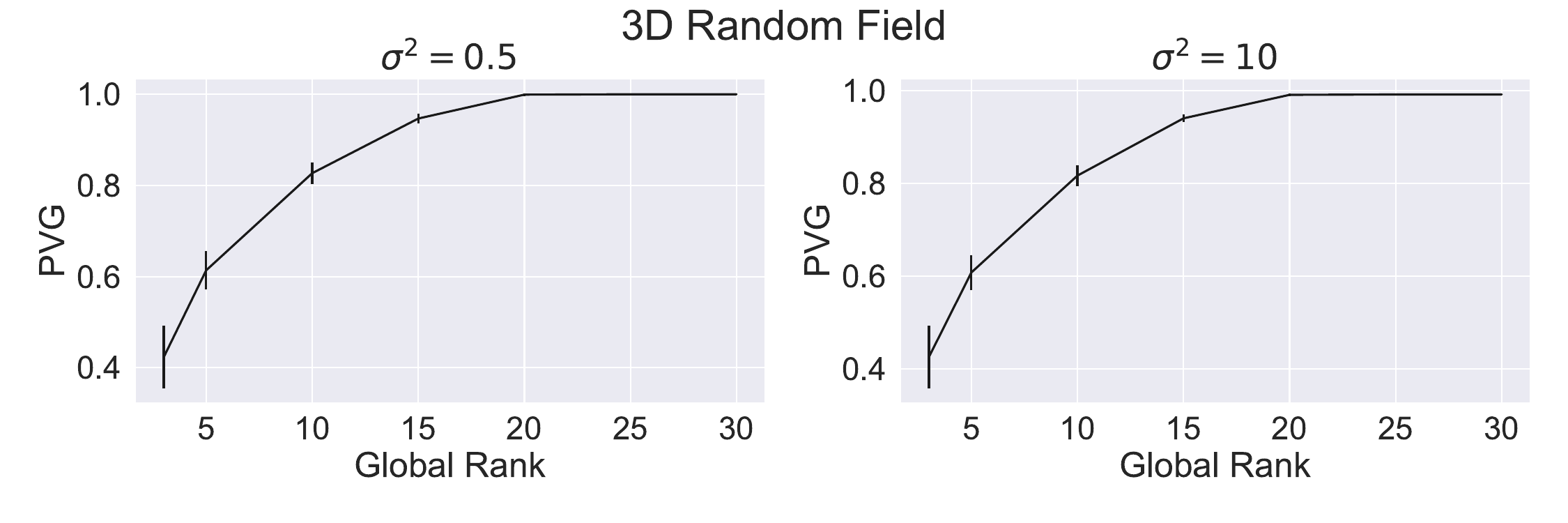} \\
  \end{tabular}
  \vspace{\floatsep}
  \begin{tabular}{@{}c@{}}
    \includegraphics[scale=0.4]{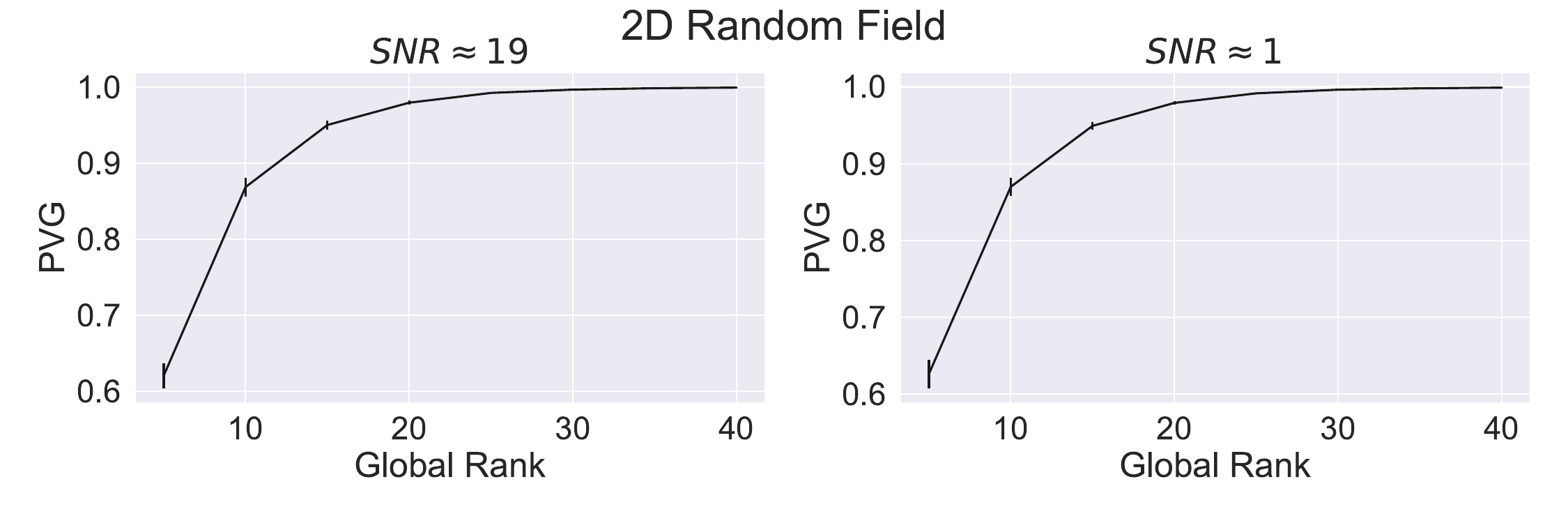}\\
  \end{tabular}
  \caption{\revised{Monte-Carlo average PVG as a function of increasing $K_{fit}$. The top two plots show the results from the 3D simulation set-up detailed in Section 5.1 with $n_d=50, N=50$ and $K^{t}=20$. The bottom two plots show the results from the 2D simulation set-up from Section 5.2 with $N=100$.} }\label{fig:global_rank_selection}
\end{figure}

\subsection{Penalty Strength Selection via Cross Validation}\label{apx:penalty_strengths}
\revised{Multidimensional functional data analysis is notoriously high-dimensional, and therefore the consideration of the statistical validatity of a selection procedure must always be balanced with considerations of the computational efficiency. Although \texttt{MARGARITA} allows all $D+1$ penalty parameters to be specified independently, in the absence of problem specific a-priori information, we suggest to let $\lambda_d = \lambda_f$ for $d=1,...,D$ and then select $(\lambda_f, \lambda_{D+1})^{\prime}$ by minimizing the $T$-fold cross-validation error over a 2-dimensional grid of potential values. To make things explicit, we provide pseudo-code for the automated smoothing parameter selection, referred to as  \texttt{CV-MARGARITA}, in Algorithm~\ref{alg:cv_margarita}. Notice that using the tensor unfolding operator \citep{kolda2009}, the high-dimensional regression \eqref{eqn:tensor_regression} can be re-written as
\begin{equation}\label{eqn:tensor_regression_vectorized}
 \begin{aligned}
     \widehat{\boldsymbol{B}}^{(t)} &:= \min_{\boldsymbol{B}^{(t)}} \|\boldsymbol{Y}_{(D+1)}^{(t)} - \boldsymbol{B}^{(t)}\boldsymbol{Z}^{(-t)}_{(D+1)} \|_{F}^2 + \lambda_{D+1}l(\boldsymbol{B}^{(t)}) \\
     & = \min_{\boldsymbol{B}^{(t)}} \|\boldsymbol{Y}_{(D+1)}^{(t)\prime} - \boldsymbol{Z}^{(-t)\prime}_{(D+1)}\boldsymbol{B}^{(t)\prime} \|_{F}^2 + \lambda_{D+1}l(\boldsymbol{B}^{(t)})
 \end{aligned}   
\end{equation}}
\revised{Now, \eqref{eqn:tensor_regression_vectorized} has the form of Equation (13) from the main text, and hence can be solved (globally) using the same ADMM scheme discussed in Section 3.4 for general $l(\cdot)$, or in closed from for the special case of $l(\cdot) = \|\cdot\|_{F}^2$.}
\begin{algorithm}[t]
  \caption{\texttt{CV-MARGARITA}}
  \label{alg:cv_margarita}
   \scriptsize 
  \begin{algorithmic}[1]
    \State \textbf{Input} $\mathcal{Y}$, $\mathcal{X}$, $\{\boldsymbol{\phi}_{m_{1},1},...,\boldsymbol{\phi}_{m_{D},D}\}$, $\{L_1, ...,\
    L_D\}$, candidate set $\{(\lambda_{f}, \lambda_{D+1}\})$, global rank $K$, number of folds $T$
    \State \textbf{Output} Selected smoothing parameters $(\lambda_f^{*}, \lambda_{D+1}^{*})$
    \State  Partition the observed data tensor into $T$-folds $\{\mathcal{Y}^{(1)}, ..., \mathcal{Y}^{(T)}\}$, where $\mathcal{Y}^{(t)}\in\mathbb{R}^{n_{1}\times\cdots\times n_{D}\times N_{t}}$. Denote $\mathcal{Y}^{(-t)}$ the tensor with all but the $\mathcal{Y}^{(t)}$ fold.
    \State Compute $\boldsymbol{T}_d$ using Proposition 3.2 and  $\boldsymbol{\Phi}_d = \boldsymbol{U}_d\boldsymbol{D}_d\boldsymbol{V}_d^{\prime}$ using $\boldsymbol{\phi}_d$ and $\mathcal{X}$ 
    \For{each candidate parameter $(\lambda_f, \lambda_{D+1})$}
    \For{t = 1, ..., T}
    \State Apply Algorithm~\ref{alg:ss_bcd} with $\mathcal{Y}^{(-t)}$ penalty parameters $(\lambda_1 = \lambda_f, ..., \lambda_{D} = \lambda_{f}, \lambda_{D+1})$ to obtain $\boldsymbol{C}_1^{(-t)}, ..., \boldsymbol{C}_D^{(-t)}$, $\boldsymbol{B}^{(-t)}$
    \State Form tensor $\mathcal{Z}^{(-t)}\in\mathbb{R}^{ n_{1}\times\cdots\times n_{D}\times K}$, with element-wise definition $$\mathcal{Z}^{(-t)}_{k,i_{1},...,i_{D}} = \prod_{d=1}^D\xi_{k,d}^{(-t)}(x_{d,i_{d}}),  \quad \xi_{k,d} = \sum_{j=1}^{m_{d}}\boldsymbol{C}^{(-t)}_{d}(k,j)\phi_{d,j}$$.
    \State Compute the tensor regression on the held-out data:
    \begin{equation}\label{eqn:tensor_regression}
            \widehat{\boldsymbol{B}}^{(t)}:= \min_{\boldsymbol{B}^{(t)}} \left\|\mathcal{Y}^{(t)} - \mathcal{Z}^{(-t)} \times_{D+1} \boldsymbol{B}^{(t)}\right\|_{F}^2 + \lambda_{D+1}l(\boldsymbol{B}^{(t)}),
    \end{equation}
    \State  Form the cross-validation error
    $$
    \text{CV}^{(t)}(\lambda_f, \lambda_{D+1}) =\left\|\mathcal{Y}^{(t)} - \mathcal{Z}^{(-t)} \times_{D+1} \boldsymbol{B}^{(t)}\right\|_{F}^2,
    $$
    \EndFor
    \EndFor
    \State  Set $(\lambda_f^{*}, \lambda_{D+1}^{*})$ which minimizes $\sum_{t=1}^T\text{CV}^{(t)}$
  \end{algorithmic}
\end{algorithm}
\par 
\revised{We now provide some justification for our approach to smoothing parameter selection by considering some of the tactics used in related work.
Several works \citep{huang2009,allen2013a,allen2019} propose the integration of penalty selection into the coordinate-wise updates of the estimation algorithm, using the so-called nested generalized cross validation. In the following, we derive an analogous procedure for coordinate-wise updating for all factor matrices and then discuss why we recommend our \texttt{CV-MARGARITA} instead.}
\par 
\revised{We start with the $\tilde{\boldsymbol{C}}_d$ parameters. Using properties of the $\text{vec}$ operator and the permutation invariance of the trace, it is easy to show that the objective function for the optimization problem in Equation (12) can be vectorized into a high dimensional ridge-type regression as follows:
\begin{equation}\label{eqn:equiv_ridge}
\begin{aligned}
  &\|\text{vec}(\boldsymbol{G}_{(d)}) -  \text{vec}(\boldsymbol{W}_d^{(r)}\boldsymbol{\tilde{C}}_d^{\prime})\|_{F}^2 + \lambda_d\mathrm{tr}(\boldsymbol{\tilde{C}}_d\boldsymbol{T}_{d}\boldsymbol{\tilde{C}}^{\prime}_d) \\
  &= \|\text{vec}(\boldsymbol{G}_{(d)}) -  (\boldsymbol{I}_{m_{d}}^\prime\otimes\boldsymbol{W}_d^{(r)})\text{vec}(\boldsymbol{\tilde{C}}_d^{\prime})\|_{F}^2 + \lambda_d\mathrm{tr}((\boldsymbol{T}_{d}^{1/2}\boldsymbol{\tilde{C}}_d^{\prime})^{\prime}\boldsymbol{T}_{d}\boldsymbol{\tilde{C}}^{\prime}_d) \\
   & = \|\text{vec}(\boldsymbol{G}_{(d)}) -  (\boldsymbol{I}_{m_{d}}^\prime\otimes\boldsymbol{W}_d^{(r)})\text{vec}(\boldsymbol{\tilde{C}}_d^{\prime})\|_{F}^2 + \lambda_d\text{vec}(\boldsymbol{\tilde{C}}_d^{\prime})^{\prime}(\boldsymbol{I}_{K}\otimes \boldsymbol{T}_{d})\text{vec}(\boldsymbol{\tilde{C}}_d^{\prime}).
\end{aligned}
\end{equation}}
\revised{The corresponding ``hat'' matrix is given by
$$
\boldsymbol{H}_d(\lambda) = (\boldsymbol{I}_{m_{d}}^\prime\otimes\boldsymbol{W}_d)\left[(\boldsymbol{I}_{m_{d}}^\prime\otimes\boldsymbol{W}_d)^{\prime}(\boldsymbol{I}_{m_{d}}^\prime\otimes\boldsymbol{W}_d) + \lambda_d(\boldsymbol{I}_{K}\otimes \boldsymbol{T}_{d})\right]^{-1}(\boldsymbol{I}_{m_{d}}^\prime\otimes\boldsymbol{W}_d^{\prime})
$$
In order to compute a generalized cross validation selection criteria, we need to obtain a measure of the degrees of freedom of the model. In the ridge regression set-up, this requires the computation of the trace, given by
\begin{equation}\label{eqn:trace}
    \begin{aligned}
  \text{tr}(\boldsymbol{H}_d(\lambda)) = \text{tr}((\boldsymbol{I}_{m_{d}}\otimes \boldsymbol{W}_d^{\prime}\boldsymbol{W}_{d})\left[(\boldsymbol{I}_{m_{d}}\otimes\boldsymbol{W}_d^{\prime}\boldsymbol{W}_d) + \lambda_d(\boldsymbol{I}_{K}\otimes \boldsymbol{T}_{d})\right]^{-1}).
\end{aligned}
\end{equation}}
\par 
\revised{For the $\boldsymbol{B}$ parameters, under the special case of $l(\cdot) = \|\cdot\|_2^2$, we have a similar structure to \eqref{eqn:equiv_ridge}, with $\boldsymbol{T}_{D+1} := \boldsymbol{I}_{N}$, and hence the previous analysis holds. For the case when $l(\cdot) = \|\cdot\|_{1}$, Equation (13) can be unfolded into a high-dimensional Lasso regression problem, and hence the degrees of freedom can be quantified using, e.g. the measure from \cite{tibshirani2012}. }
\par 
\revised{All this said, integrating this nested procedure to solve the hyper-parameter selection problem is non-trivial. First, notice that computation of the trace requires the inversion of a $m_{d}K \times m_{d}K$ matrix, for each candidate $\lambda_{d}$, at each iteration of \texttt{MARGARITA}. This may not be a problem for relatively small models, but is undesirable for large models. Note that the approaches in \citep{allen2013a,huang2009,allen2019} are deflationary approaches which perform a series of rank-$1$ approximations, and thus do not encounter this issue. Furthermore, care must be taken to avoid convergence issues, as the convergence of the sub-problems in \texttt{MARGARITA}, which is currently guaranteed as discussed in Section 3.4 of the main text, may no longer hold. This potential issue is noted in \cite{allen2019}, though no rigorous proposal is made in order to guard against it. Integrating a nested procedure may be an interesting avenue for future work, but the potential convergence problems as well as the previously highlighted potential computational issues related to the matrix inversion for the large $m_d,K$ case will need to be rigorously handled, which is well beyond the scope of the current work.}

\section{Brief Overview of Competing Methods}\label{sec:competing_methods}
The so-called \textit{sandwich smoother}, introduced by \cite{xiao2013}, is a method for estimating the coefficients of a tensor product approximation to an unknown deterministic function from noisy observations on a grid. The main contribution is in a clever formulation of the penalty term, which allows for the fast computation of the GCV statistic and hence a computationally efficient technique for selecting the roughness penalty strength. For more information, see the aforementioned paper or the \texttt{hero} package in R \citep{hero2020}.
\par 
We give a brief overview of the FCP-TPA algorithm \citep{allen2013a}, which is essentially a $D$-dimensional extension of the 2-dimensional regularization scheme from \cite{huang2009}. Using our notation, the FCP-TPA estimates the $k$th \textit{MPB basis evaluation vectors} $\boldsymbol{\Xi}_{d,k}$ and associated coefficient vector $\boldsymbol{b}_{k}$ by solving a series of $K$ rank-one penalized decompositions of the residual tensor. That is, at the $k$th iteration, FCP-TPA solves problem 
\begin{equation}\label{eqn:FPC_TPA_objective}
    \min_{\boldsymbol{\Xi}_{1,k},...,\boldsymbol{\Xi}_{D,K},\boldsymbol{b}_{k}}\Big\|\mathcal{Y}_{resid} - \bigotimes_{d=1}^D\boldsymbol{\Xi}_{d,k}\otimes\boldsymbol{b}_{k}\Big\|_{F}^2 - \prod_{d=1}^D\Big\|\boldsymbol{\Xi}_{d,k}\Big\|_2^2 + \prod_{d=1}^D\boldsymbol{\Xi}_{d,k}^\prime\boldsymbol{P}_d^{-1}\boldsymbol{\Xi}_{d,k}
\end{equation}
where $\mathcal{Y}_{resid} = \mathcal{Y} - \sum_{j=1}^{k-1}\bigotimes \boldsymbol{\Xi}_{d,j}\otimes \boldsymbol{b}_{j}$ and $\boldsymbol{P}_d\in\mathbb{R}^{n_d\times n_d}$ is a smoothing matrix, e.g. derived using squared second order differences. The solution to \eqref{eqn:FPC_TPA_objective} is approximated using a series of rank-1 approximations, each of which are solved using tensor power iterations which are shown to converge to a stationary point. 
\par 
Note that FCP-TPA does not directly construct a continuous representation but rather the discrete evaluations of the optimal marginal product functions on the observed marginal grid, i.e. the $\boldsymbol{\Xi}_d$'s. In order to obtain a continuous representation from the output of FCP-TPA, a ``decompose-then-represent'' approach is used in which the marginal basis functions are estimated from the basis expansion of the $\boldsymbol{\Xi}_d$'s. 

\section{Additional Simulation Studies and Details}\label{sec:addition_sims}
\revised{All simulations were performed using R/4.0.2 and Python/3.8.16 on a Linux machine equipped with a 2.4 GHz Intel Xeon CPU E5-2695 and 24GB of RAM.}
\subsection{Additional Comparisons to TPB and FCP-TPA}
\revised{
Figure \ref{fig:dof_mise_comparison} and Table \ref{tab:mpf_vs_tpb} display comparisons of the performance between \texttt{MARGARITA} and the tensor product basis (TPB) estimated by the sandwhich smoother for all simulation settings considered. We note substantially better performance for \texttt{MARGARITA} over TPB for comparable degrees of freedom, given by $\prod_{d=1}^D m_d$ and $K_{\text{fit}}\sum_{d=1}^D m_d$ for the TPB and \texttt{MARGARITA}, respectively.}

\begin{figure}
  \centering
  \includegraphics[scale=0.4]{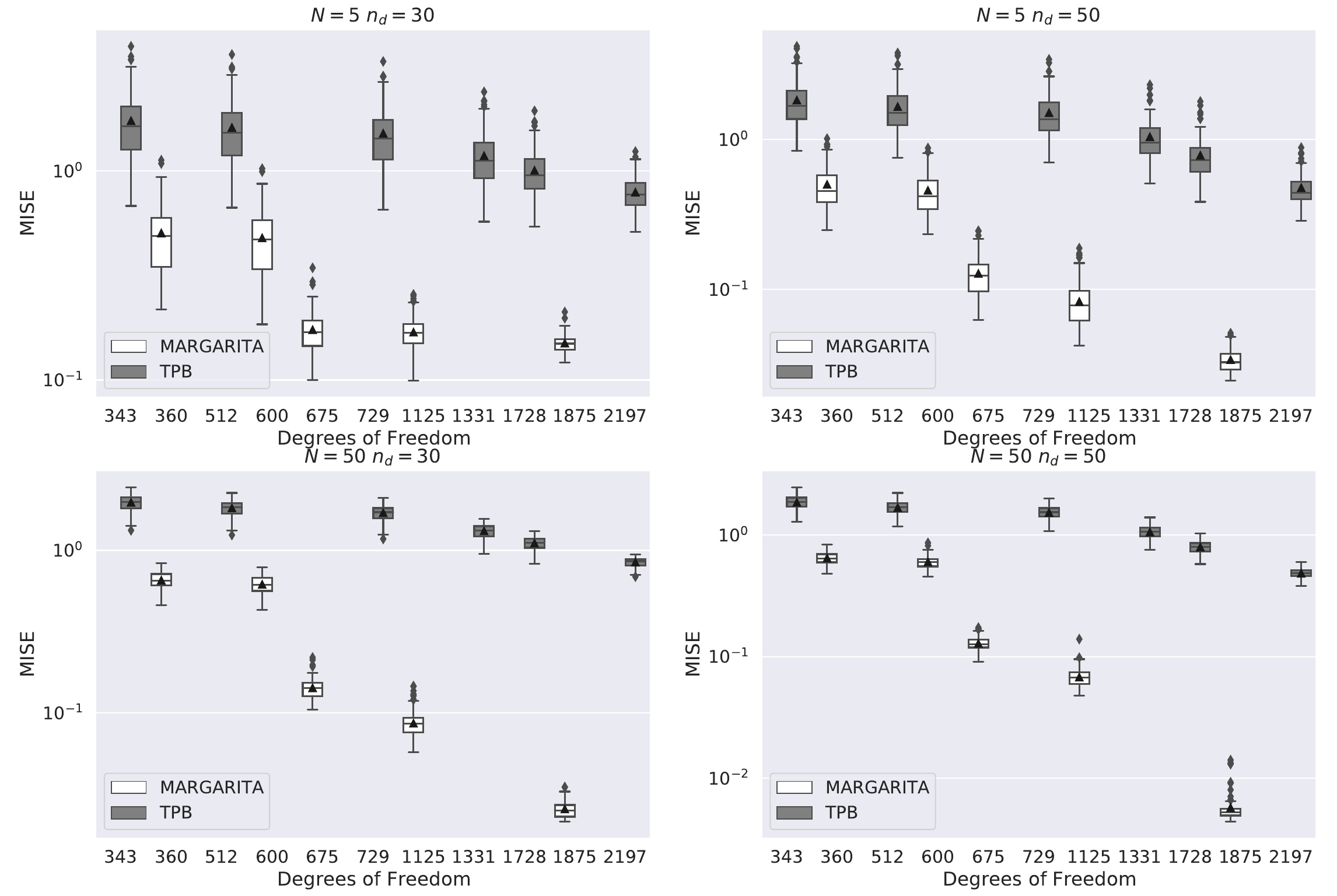}
  \caption{Comparison of the fit performance measured by MISE between the TPB estimated by the sandwhich smoother (gray) and \texttt{MARGARITA} (white) as a function of the total number of degrees of freedom. For each panel $K_t=20$ and $\sigma^2=10$. The Y-axis is plotted on log-scale for clarity.}
  \label{fig:dof_mise_comparison}
\end{figure}
\begin{table}
\centering
\scriptsize
\caption{moMISE comparison of \texttt{MARGARITA} (a), to the tensor product basis estimated by sandwhich smoother (b).}
\subcaption*{(a) Marginal Product Basis}
\begin{tabular}{rrrr|rrrrrr}
\toprule
 \multicolumn{4}{l}{} & \multicolumn{3}{l}{$K_{fit}=15$} & \multicolumn{3}{l}{$K_{fit}=25$} \\
 \multicolumn{4}{l}{} & \multicolumn{6}{l}{$m_{d}$ (model d.o.f.)}\\
$K_{true}$ & $\sigma^2$ &   $N$ &  $n_d$ &  8 (360) &  15 (675) &  25 (1,125) & 8 (600) &  15 (1,125) & 25 (1,875) \\
\midrule
     10 &     0.5 &   5 &  30 &       0.0890 &        0.0433 &        0.0433 &       0.0511 &        0.0030 &        0.0033 \\
     10 &     0.5 &   5 &  50 &       0.0932 &        0.0405 &        0.0402 &       0.0560 &        0.0006 &        0.0006 \\
     10 &     0.5 &  50 &  30 &       0.1186 &        0.0426 &        0.0411 &       0.0790 &        0.0017 &        0.0007 \\
     10 &     0.5 &  50 &  50 &       0.1196 &        0.0398 &        0.0394 &       0.0757 &        0.0001 &        0.0001 \\
     10 &    10.0 &   5 &  30 &       0.1148 &        0.1049 &        0.1634 &       0.1039 &        0.0919 &        0.1579 \\
     10 &    10.0 &   5 &  50 &       0.0976 &        0.0530 &        0.0647 &       0.0616 &        0.0180 &        0.0312 \\
     10 &    10.0 &  50 &  30 &       0.1166 &        0.0525 &        0.0632 &       0.0800 &        0.0146 &        0.0263 \\
     10 &    10.0 &  50 &  50 &       0.1149 &        0.0423 &        0.0441 &       0.0745 &        0.0019 &        0.0033 \\
     20 &     0.5 &   5 &  30 &       0.4910 &        0.1158 &        0.0684 &       0.4418 &        0.0511 &        0.0059 \\
     20 &     0.5 &   5 &  50 &       0.4872 &        0.1075 &        0.0627 &       0.4429 &        0.0475 &        0.0015 \\
     20 &     0.5 &  50 &  30 &       0.6334 &        0.1239 &        0.0679 &       0.5850 &        0.0564 &        0.0025 \\
     20 &     0.5 &  50 &  50 &       0.6388 &        0.1182 &        0.0646 &       0.5998 &        0.0539 &        0.0009 \\
     20 &    10.0 &   5 &  30 &       0.5058 &        0.1587 &        0.1723 &       0.4736 &        0.1200 &        0.1492 \\
     20 &    10.0 &   5 &  50 &       0.4994 &        0.1185 &        0.0841 &       0.4534 &        0.0594 &        0.0286 \\
     20 &    10.0 &  50 &  30 &       0.6544 &        0.1353 &        0.0867 &       0.6008 &        0.0682 &        0.0220 \\
     20 &    10.0 &  50 &  50 &       0.6415 &        0.1231 &        0.0674 &       0.5915 &        0.0579 &        0.0039 \\
\bottomrule
\end{tabular}
\bigskip
\subcaption*{(b) Tensor Product Basis}
\begin{tabular}{rrrr|rrrrrr}
\toprule
 \multicolumn{4}{l}{} & \multicolumn{6}{l}{$m_{d}$ (model d.o.f.)}\\
$K_{true}$ & $\sigma^2$ &   $N$ &  $n_d$ &   7 (343) &   8 (512) &   9 (729) &  11 (1,331) &  12 (1,728) & 13 (2,197) \\
\midrule
     10 &     0.5 &   5 &  30 &  1.1409 &  0.9520 &  0.9096 &  0.5327 &  0.4817 &  0.2136 \\
     10 &     0.5 &   5 &  50 &  1.1251 &  0.9326 &  0.8870 &  0.5067 &  0.4581 &  0.1823 \\
     10 &     0.5 &  50 &  30 &  1.1457 &  0.9556 &  0.9130 &  0.5338 &  0.4840 &  0.2134 \\
     10 &     0.5 &  50 &  50 &  1.1170 &  0.9251 &  0.8789 &  0.5020 &  0.4516 &  0.1832 \\
     10 &    10.0 &   5 &  30 &  1.1579 &  1.0224 &  1.0136 &  0.7705 &  0.7559 &  0.6051 \\
     10 &    10.0 &   5 &  50 &  1.1476 &  0.9677 &  0.9300 &  0.5822 &  0.5492 &  0.3027 \\
     10 &    10.0 &  50 &  30 &  1.1822 &  1.0465 &  1.0369 &  0.7861 &  0.7721 &  0.6109 \\
     10 &    10.0 &  50 &  50 &  1.1518 &  0.9679 &  0.9334 &  0.5857 &  0.5517 &  0.3036 \\
     20 &     0.5 &   5 &  30 &  1.8024 &  1.6254 &  1.4738 &  0.9819 &  0.7065 &  0.3806 \\
     20 &     0.5 &   5 &  50 &  1.7445 &  1.5693 &  1.4194 &  0.9342 &  0.6623 &  0.3464 \\
     20 &     0.5 &  50 &  30 &  1.8352 &  1.6530 &  1.4973 &  1.0030 &  0.7218 &  0.3842 \\
     20 &     0.5 &  50 &  50 &  1.8207 &  1.6349 &  1.4739 &  0.9748 &  0.6911 &  0.3580 \\
     20 &    10.0 &   5 &  30 &  1.8076 &  1.6757 &  1.5688 &  1.2221 &  1.0375 &  0.8092 \\
     20 &    10.0 &   5 &  50 &  1.8592 &  1.6807 &  1.5290 &  1.0597 &  0.7971 &  0.4817 \\
     20 &    10.0 &  50 &  30 &  1.9789 &  1.8289 &  1.7073 &  1.3210 &  1.1113 &  0.8472 \\
     20 &    10.0 &  50 &  50 &  1.8527 &  1.6765 &  1.5280 &  1.0560 &  0.7946 &  0.4833 \\
\bottomrule
\end{tabular}
\label{tab:mpf_vs_tpb}
\end{table}
\par
\revised{Table \ref{tab:F_ADMM_vs_FCP_TPA} displays the relative difference in moMISE, defined as 
\begin{equation}\label{eqn:rel_moMISE}
    \frac{\text{moMISE}_{\text{FCP-TPA}} - \text{moMISE}_{\mathtt{MARGARITA}}}{\text{moMISE}_{\text{FCP-TPA}}}
\end{equation}
between the fits resulting from the FCP-TPA algorithm and \texttt{MARGARITA} for all marginal and global ranks and simulation settings. As all but one of the entries in the table is positive, \texttt{MARGARITA} is nearly uniformly outperforming the FCP-TPA. We see that the relative boost in performance from \texttt{MARGARITA} generally increases with $K_{fit}$ and marginal rank $m_d$}.
\begin{table}
    \centering
    \small
    \caption[Relative difference in moMISE for FCP-TPA and \texttt{MARGARITA} for marginal ranks 15 and 25 and $K_{fit}=8, 15, 25$.]{Relative difference in moMISE for FCP-TPA and \texttt{MARGARITA} for marginal ranks 15 and 25 and $K_{fit}=8, 15, 25$. Positive values indicate lower moMISE for \texttt{MARGARITA}. A grid search to select $\lambda_d$ was performed for each fit and the results from the optimal value are reported. The entry in bold face indicates the \textbf{only case} that FCP-TPA outperformed \texttt{MARGARITA}.}
    \centering
    \begin{tabular}{rrrr|rrrrrr}
    \toprule
     \multicolumn{4}{l}{} & \multicolumn{3}{l}{$K_{fit}=15$} & \multicolumn{3}{l}{$K_{fit}=25$} \\
     \multicolumn{4}{l}{} & \multicolumn{6}{l}{$m_{d}$}\\
    $K_{true}$ & $\sigma^2$ &   $N$ &  $n_d$ &  8  &  15 &  25 & 8 &  15 & 25  \\
    \midrule
     10 &     0.5 &   5 &  30 &  0.2113 &  0.2506 &  0.1738 &  0.3062 &  0.7949 &  0.7160 \\
     10 &     0.5 &   5 &  50 &  0.2544 &  0.2433 &  0.1534 &  0.3634 &  0.9328 &  0.8518 \\
     10 &     0.5 &  50 &  30 &  0.2062 &  0.2118 &  0.1555 &  0.2323 &  0.9228 &  0.8734 \\
     10 &     0.5 &  50 &  50 &  0.1597 &  0.1850 &  0.1167 &  0.2954 &  0.9757 &  0.9437 \\
     10 &    10.0 &   5 &  30 &  0.3684 &  0.2533 & \textbf{-0.0770} &  0.4413 &  0.4426 &  0.4895 \\
     10 &    10.0 &   5 &  50 &  0.3086 &  0.3824 &  0.2120 &  0.4409 &  0.6580 &  0.4486 \\
     10 &    10.0 &  50 &  30 &  0.3487 &  0.3509 &  0.1540 &  0.3529 &  0.6836 &  0.5607 \\
     10 &    10.0 &  50 &  50 &  0.2080 &  0.3300 &  0.2446 &  0.3079 &  0.7229 &  0.7629 \\
     20 &     0.5 &   5 &  30 &  0.0946 &  0.4304 &  0.4429 &  0.1152 &  0.6157 &  0.8706 \\
     20 &     0.5 &   5 &  50 &  0.1139 &  0.4149 &  0.4505 &  0.1243 &  0.6088 &  0.9428 \\
     20 &     0.5 &  50 &  30 &  0.0708 &  0.4452 &  0.4624 &  0.0743 &  0.6107 &  0.9699 \\
     20 &     0.5 &  50 &  50 &  0.0835 &  0.4274 &  0.4554 &  0.0880 &  0.6092 &  0.9816 \\
     20 &    10.0 &   5 &  30 &  0.1235 &  0.4651 &  0.3642 &  0.1385 &  0.5306 &  0.6020 \\
     20 &    10.0 &   5 &  50 &  0.1329 &  0.4600 &  0.5117 &  0.1382 &  0.5928 &  0.7476 \\
     20 &    10.0 &  50 &  30 &  0.0954 &  0.4489 &  0.5313 &  0.0753 &  0.5856 &  0.8161 \\
     20 &    10.0 &  50 &  50 &  0.0915 &  0.4349 &  0.5089 &  0.0983 &  0.6155 &  0.9299 \\
\bottomrule
    \end{tabular}
   \label{tab:F_ADMM_vs_FCP_TPA}
\end{table}
\par 
\revised{Figure~\ref{fig:mpf_computational_time_comparison} displays a comparison of the computational time between FCP-TPA and \texttt{MARGARITA} for a variety of model and data sizes. We note that the methods exhibit comparable performance for the small sample, small domain case, while \texttt{MARGARITA} is much faster for the large sample, large domain case.}
\begin{figure}[t]
    \centering
    \includegraphics[scale=0.37]{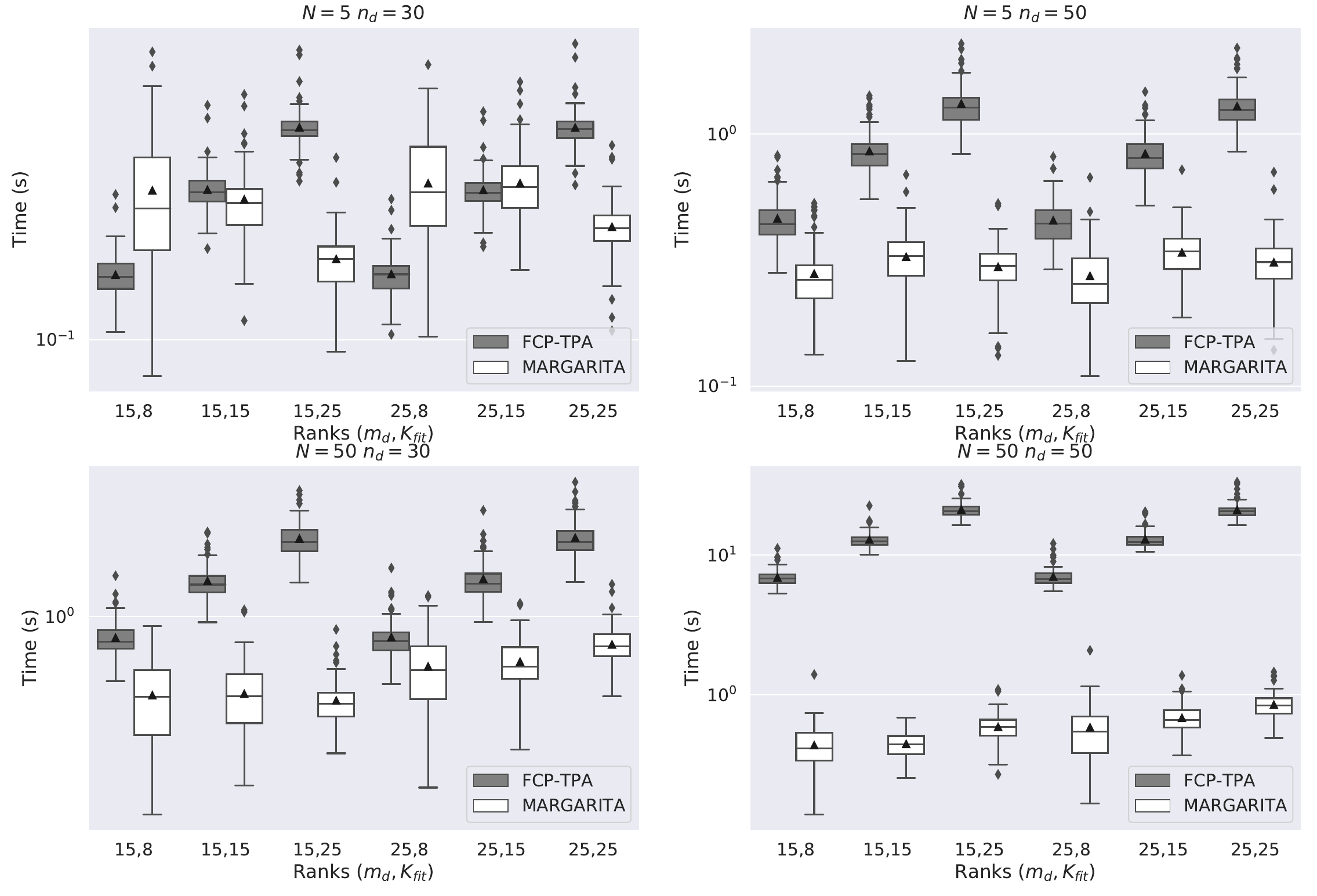}
    \caption[Computational time comparison between FCP-TPA and \texttt{MARGARITA}]{Computational time comparison between FCP-TPA and \texttt{MARGARITA}. The Y-axis is plotted on log-scale for clarity.}
    \label{fig:mpf_computational_time_comparison}
\end{figure}
\subsection{MargFPCA and Two-Stage FPCA Analysis}\label{ssec:margfpca_vs_margarita}
\revised{The eigenfunctions determining a non-stationary, non-separable anisotropic covariance function $C(\mathbf{x},\mathbf{y})$ over $\mathcal{M}$ are defined as follows. Denote the eigen-decomposition of the pairwise $\mathbb{L}^2$ inner product matrix of the tensor product basis system $\boldsymbol{J}_{\boldsymbol{\phi}_{1}\otimes\boldsymbol{\phi}_{2}} = \boldsymbol{P}\boldsymbol{\Gamma}\boldsymbol{P}^\prime$. The collection of $m_1m_2$ orthonormal eigenfunctions are defined according to $\boldsymbol{\psi} = \boldsymbol{\Gamma}^{-1/2}\boldsymbol{P}^\prime\text{vec}(\boldsymbol{\phi}_1\otimes\boldsymbol{\phi}_2)$. The eigenvalue corresponding to the $k$th eigenfunction is given by an exponential decay model $\rho_k = \text{exp}(-0.5k)$. Realizations of the random function $U_i \sim U$ are simulated using a Gaussian process assumption and then evaluated on an equispaced $200\times 200$ grid on $\mathcal{M}$. $\boldsymbol{\phi}_{1}$ and $\boldsymbol{\phi}_{2}$ are used as the marginal basis for fitting and are taken to be equispaced cubic b-splines with $m_1=10$ and $m_2=8$. Note that these marginal ranks can be consistently identified via our rank selection procedure (see the bottom panel of Figure~\ref{fig:marg_rank_selection}).}
\par 
\revised{Table~\ref{tab:MARGARITA_vs_MARG_FPCA} shows the results of the simulation study in Section 5.2 for additional training sample sizes and ranks. The interpretation of the results echo those in the main text. Table~\ref{tab:two_stage_eigenfunction_estimates} shows the performance of the proposed two-stage FPCA procedure for estimating the first three eigenfunctions. The estimation performance is quantified using the angular error:
$AE(\hat{\psi}_i) :=  1 - |\langle\psi_i,\hat{\psi}_i\rangle_{\mathbb{L}^2(\mathcal{M})}|$. As the eigenfunctions are unit norm, the $AE(\hat{\psi}_i) \in [0,1]$, with $0$ indicating perfect recovery and $1$ indicating orthogonality. As desired, we see that the $AE$ converging to $0$ as $N$ increases. Figure~\ref{fig:eigenfunction_estimates} shows the first three eigenfunctions (right column) and estimates from the two stage FPCA for $N_{train}=100$ (left column). Notice that the eigenfunctions are extremely high frequency, making this a challenging estimation problem. Despite this, the fits are nearly visually identical.}
\begin{table}
\centering
\small
\caption{Monte Carlo average MISE for representing a new realization for both \texttt{MARGARITA} and \texttt{MargFPCA}, for a variety of ranks and training sample sizes.}
\begin{tabular}{rr|rrrr|rrrr}
\toprule
 & & \mc{MARGARITA} & \mc{MargFPCA} \\
&  $N_{train}$ & $10$ & $20$ & $50$ & $100$ & $10$ & $20$ & $50$ & $100$  \\
\midrule
\multirow{5}{0.1mm}{$K$} & 5  &     1.5582 &     1.4204 &     1.3652 &      1.3607 &    2.2583 &    2.1277 &    2.1387 &     2.1317 \\
& 10 &     0.8247 &     0.7333 &     0.6373 &      0.6104 &    1.9119 &    1.8598 &    1.8380 &     1.8800 \\
& 15 &     0.5177 &     0.3554 &     0.3037 &      0.3009 &    1.6870 &    1.6224 &    1.5539 &     1.6279 \\
& 20 &     0.3230 &     0.1904 &     0.1443 &      0.1324 &    1.4540 &    1.4180 &    1.4030 &     1.3785 \\
& 30 &     0.1205 &     0.0535 &     0.0319 &      0.0288 &    1.0688 &    1.0780 &    1.0772 &     1.0796 \\
\bottomrule
\end{tabular}
\label{tab:MARGARITA_vs_MARG_FPCA}
\end{table}
\begin{table}[]
\centering
\begin{tabular}{lrrr}
\toprule
$N_{\text{train}}$ &     20  &     50 &     100 \\
\midrule
$\psi_1$ & 0.2248 $\pm$ 0.0566 &  0.0929 $\pm$ 0.0231&  0.0341 $\pm$ 0.0099\\
$\psi_2$ &  0.3401 $\pm$ 0.0530&  0.1677 $\pm$ 0.0306 &  0.0751 $\pm$ 0.0150 \\
$\psi_3$ & 0.3898 $\pm$ 0.0537 &  0.1861 $\pm$ 0.0387 &  0.0842 $\pm$  0.0153\\
\bottomrule
\end{tabular}
    \caption{\revised{Average angular distance with standard errors for the two stage eigenfunction estimates.}}
    \label{tab:two_stage_eigenfunction_estimates}
\end{table}
\begin{figure}
    \centering
\includegraphics[scale=0.4]{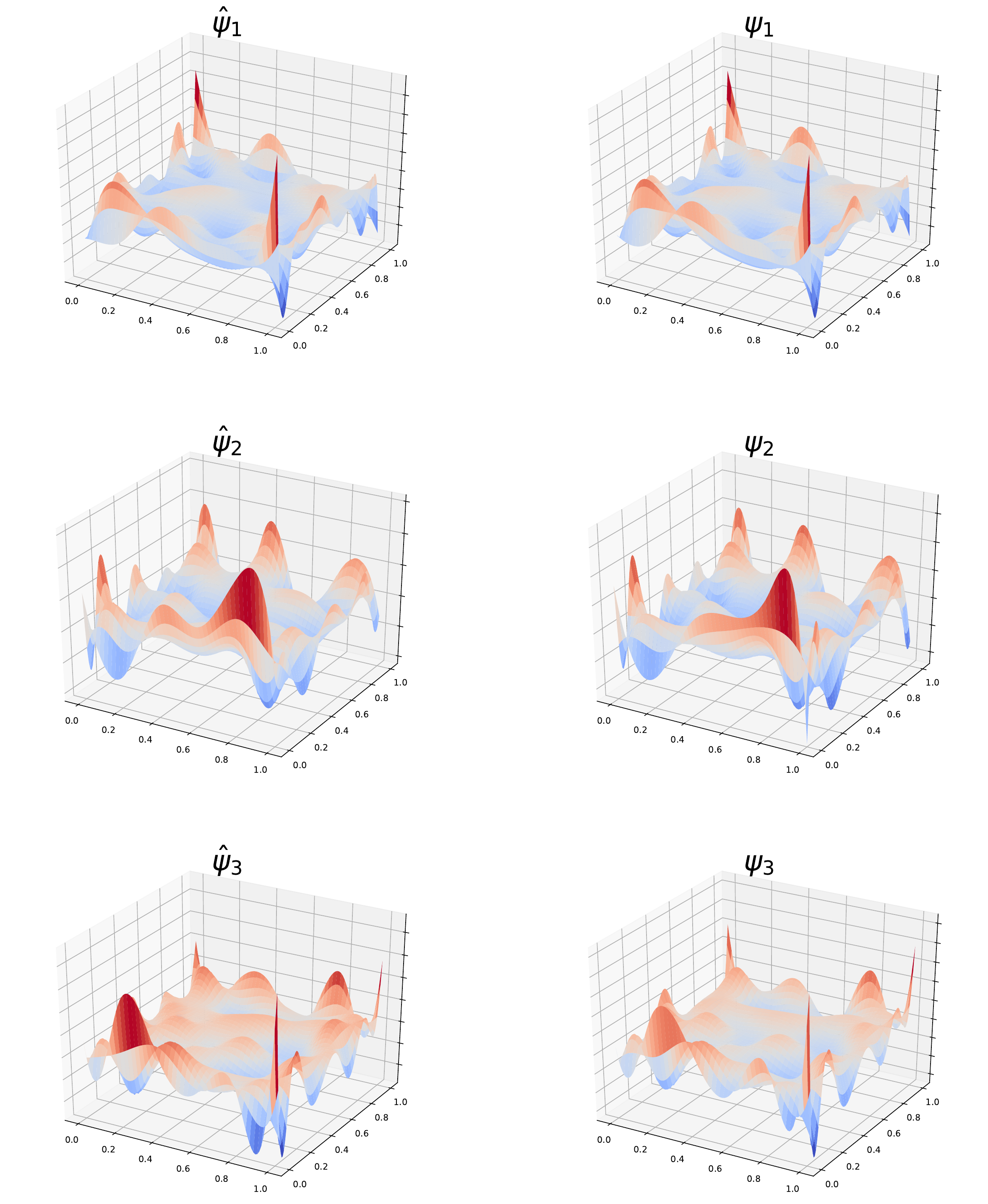}
    \caption{\revised{First three true eigenfunctions (right column) and their estimates using the two-stage FPCA with $N=100$ (left column).}}
    \label{fig:eigenfunction_estimates}
\end{figure}
\section{TBI Data Description and Preprocessing}\label{sec:additional_real_data_analysis}

All subjects in our study were referred to the University of Rochester Medical Imaging Center and imaged on the same 3T MRI scanner. Study inclusion criteria included history of concussion, while exclusion criteria included dental braces, prior brain surgery, ventricular shunt, skull fractures, or other standard contraindications for MR imaging. Diagnosis of concussion was made by neurologists, physical medicine and rehabilitation physicians, and sports medicine physicians. The control group consisted of young athletes with no history of concussion. The University Institutional Review Board approved this retrospective study. All MRI examinations were reviewed by an experienced neuroradiologist for any artifacts that might affect the quality of the study, as well as for the presence of recent or remote intracranial hemorrhage, signal abnormalities in the brain, hydrocephalus, congenital or developmental anomalies.
\par 
The diffusion MRI data was collected on a single 3T scanner using a 20-channel head coil (Siemens Skyra, Erlangen, Germany). Diffusion imaging was performed with a b-value of $1000 \frac{s}{mm^2}$, using 64 diffusion-encoding directions. In addition, a $b=0 \frac{s}{mm^2}$ image was collected for signal normalization. Additional dMRI parameters included: $FOV = 256 \times 256 mm$, number of slices = $70$, image resolution = $2\times 2 \times 2 mm^3$, $TR/TE = 9000/88 ms$, Generalized autocalibrating partially parallel acquisition (GRAPPA) factor $= 2$. Acquisition of dMRI data took 10 minutes and 14 seconds. A Gradient-recalled echo (GRE) sequence was also collected with $TEs = 4.92, 7.38 ms$ at the same resolution of the dMRI to correct for susceptibility-induced distortion effects. A diffusion tensor model (DTI) was fit to each subject's diffusion data and used to compute the per-voxel FA. Registration of the FA images to to the ICBM 2009c Nonlinear Symmetric 1mm template \citep{fonov2009} was then performed using the popular ANTS software \citep{avants2009advanced}. The domain of analysis was constrained to be the convex hull of a rectangular $115\times 140 \times 120$ voxel grid in the template space covering the white matter, i.e. the raw data tensor $\mathcal{Y} \in \mathbb{R}^{115\times 140 \times 120 \times 50}$. A white matter mask was also applied to the aligned data.

\end{document}